\documentclass[11pt]{article}

\usepackage[T1]{fontenc}
\usepackage[utf8]{inputenc}
\usepackage{lmodern}

\usepackage{multirow}
\usepackage{newunicodechar}
\newunicodechar{─}{-}
\usepackage{enumitem}
\usepackage{geometry}
\usepackage{xspace}

\usepackage{amsmath,amssymb,amsthm,mathtools}
\usepackage{bm}
\usepackage{bbm}

\usepackage{graphicx}
\usepackage{pgfplots}
\pgfplotsset{compat=1.18}
\usepackage{tikz}
\usetikzlibrary{arrows.meta,positioning,calc}
\usepackage[most]{tcolorbox}

\usepackage{tabularx}

\usepackage{array}      
\usepackage{booktabs}  
\usepackage{makecell}

\usepackage{caption}
\usepackage{subcaption}
\usepackage{float}

\usepackage{algorithm}
\usepackage{algorithmicx}
\usepackage{algpseudocode}

\usepackage{eurosym}
\usepackage{siunitx}
\usepackage{listings}
\usepackage{doi}
\usepackage{setspace}

\usepackage[protrusion=true,expansion=true,patch=none]{microtype}

\usepackage{hyperref}
\hypersetup{
  colorlinks=true,
  linkcolor=black,
  citecolor=black,
  urlcolor=blue,
  pdftitle={Economic relativity: a cut rule for perimeter valuation in equity ownership networks},
  pdfauthor={O. Di Marzio},
  pdfkeywords={Cut-Based Valuation, Consolidation, Equity ownership networks, Cross-holdings, Perimeter, Minority interests, Input-Output, Ownership networks}
}

\newcommand{\cutreportrow}[2]{\textbf{#1} & #2\\}
\newcommand{\R}{\mathbb{R}}

\newcommand{\T}{\top}

\DeclareSIUnit{\EUR}{\text{\euro}}
\newcommand{\EURnum}[1]{\num[round-mode=places,round-precision=0]{#1}\,\euro}

\newtheorem{theorem}{Theorem}
\newtheorem{assumption}{Assumption}
\newtheorem{proposition}{Proposition}
\newtheorem{corollary}{Corollary}
\newtheorem{definition}{Definition}
\newtheorem{remark}{Remark}
\newtheorem{lemma}{Lemma}

\title{\textbf{Economic relativity: a cut rule for perimeter valuation in equity ownership networks}\\[1em]}
\author{O. Di Marzio\footnote{Comments welcome at \texttt{omar@omardimarzio.it}.}}
\date{August 29, 2025}

\lstdefinelanguage{json}{%
    basicstyle=\ttfamily\small,
    showstringspaces=false,
    breaklines=true,
    literate=*
     {:}{{{\color{black}{:}}}}{1}
     {,}{{{\color{black}{,}}}}{1}
     {\{}{{{\color{black}{\{}}}}{1}
     {\}}{{{\color{black}{\}}}}}{1}
     {[}{{{\color{black}{[}}}}{1}
     {]}{{{\color{black}{]}}}}{1}
}
\begin{document}

\maketitle

\begin{abstract}

We introduce the Cut-Based Valuation (CBV), a unified framework for consolidated value in
equity/flow networks. The central idea is that economic value is never absolute: it is always
defined relative to an observer $\Omega$, which fixes perimeter, measurement basis, units/FX/PPP,
discounting, informational regime, and control rules. Given $\Omega$, the Cut Theorem shows
that the consolidated value of a perimeter $P$ depends only on boundary quantities across the
cut $P \leftrightarrow O$, while internal reconnections are valuation-invariant. This provides
(i) sufficient statistics for valuation with linear computational complexity, (ii) standardized
reporting through the Perimeter-of-Validity and Cut Summary, and (iii) transformation laws that
clarify how different observers relate. Applications span IFRS consolidation, national accounts,
fund-of-funds, pyramids, and clearing networks, all seen as special cases of a general principle of
economic relativity. Case studies (market capitalization by country, keiretsu, fund-of-funds)
illustrate how CBV eliminates double counting while ensuring comparability and auditability.

To address practical concerns, we establish \emph{robustness bounds} that quantify how errors in
initial data propagate to consolidated values, and we introduce a \emph{dynamic CBV--Fisher protocol}
for intertemporal comparisons, ensuring consistency with official chain-linking practices. These
additions clarify the \emph{time scale} of application, the role of \emph{averaging procedures}, and the
horizon of \emph{reliable measurement}. Finally, we make explicit the \emph{scope and limitations}
of CBV: it is a normative measurement/consolidation rule in linear accounting environments,
while in macroeconomic closures or with nonlinear payoffs it must be coupled with equilibrium
or clearing models.

\end{abstract}

\section{Introduction}\label{sec:intro}
Equity ownership networks generate systematic consolidation problems: double counting, minority interests, and dependence on the perimeter and observer. Current standards (e.g.\ IFRS~10~\cite{IFRS10}, SNA~2008~\cite{SNA2008}) and Leontief-style approaches~\cite{Leontief1941} are powerful but not always convenient on arbitrary or partial perimeters. 

The term ``relativity'' is adopted here as an \emph{operational analogy}, not as a claim of isomorphism with physics. The goal is to highlight that every valuation is defined relative to an \textbf{observer} $\Omega=(P,\mathrm{Basis},\mathrm{Units},\mathrm{FX/PPP},\mathrm{SDF},\mathcal I,\mathcal C)$, which fixes: (i) the perimeter $P$; (ii) the measurement basis and units/FX/PPP; (iii) the discount rule or pricing measure (SDF/numéraire); (iv) the information set; and (v) the control criterion. In this sense, the consolidated value $W(P)$ is never absolute, but always a relative measure conditional on the observer.

This dependence on the observer already emerges in established practices, though in fragmented and ad hoc ways. For example: in IFRS 10, consolidation eliminates intra-group holdings and avoids double counting, producing different results than individual accounts. In SNA 2008, naïve aggregation is prevented by precise sectoral boundaries. In multi-level funds or pyramids, the same security may reappear along ownership chains, with effective value depending on the aggregation perimeter. In clearing networks (Eisenberg--Noe; Rogers--Veraart), net wealth is defined as a fixed point, neutralizing internal exposures within the perimeter. All these cases, traditionally addressed with case-specific rules, share a common root that CBV formalizes under a unified law.

Formally, the \textbf{Cut Theorem} states that the consolidated value of $P$ depends only on the \emph{boundary} quantities across the cut $P \leftrightarrow O$ and on node primitives; internal reconnections within $P$ are \emph{gauge-equivalent} and invariant for valuation. This entails three practical benefits. First, CBV provides \emph{sufficient statistics}: once the cut is extracted, valuation is linear in the number of boundary edges, regardless of internal complexity or the presence of non-linear payoffs. Second, it ensures \emph{standardized reporting} via the \emph{Perimeter-of-Validity} (PoV), which fixes $\Omega$, and the \emph{Cut Summary}, which reconciles gross and net values. Third, it defines precise \emph{transformation laws} (unit/FX/PPP scaling, change of numéraire/discounting, perimeter and control rule modifications), enabling \textbf{coherent comparisons} across different observers.

We operationalize control with three families of rules (threshold/majority, Herfindahl look-through, attenuated paths) and provide algorithms for Regime~A (direct cut evaluation) and Regime~B (internal estimation $\bm v_P=(I-O_{PP})^{-1}\bm\pi_P$ under $\rho(O_{PP})<1$). Synthetic case studies (net market capitalization by country, pyramids/keiretsu, fund-of-funds) illustrate how CBV eliminates internal double counting while preserving comparability and auditability. 

\medskip
\noindent\textbf{Principle (Observer-Dependent Valuation).} \emph{There is no absolute valuation of an asset or set of assets. Every valuation is relative to a specific observer $\Omega$. Given $\Omega$, the consolidated value $W(P)$ is obtained solely from the boundary data between $P$ and the outside, and is comparable with other valuations only after applying the appropriate observer transformations and declaring the \emph{Perimeter-of-Validity} and corresponding \emph{Cut Summary}.}

From this perspective, traditional consolidation practices (IFRS, SNA, fund-of-funds, pyramids, clearing networks) are all \emph{special cases} of the general law expressed by CBV. Just as Newtonian mechanics is a special case of relativity, existing accounting and financial methods are instances of \emph{economic relativity}, where value is always defined relative to the observer $\Omega$.

We clarify right from the start two distinct informational regimes. In \textbf{Regime A} internal equity values $\bm v_P$ are known (for example market prices). In this case the \emph{Cut Theorem} shows that $W(P)$ depends only on the flows $P\to O$ and $O\to P$ and on the internal base $\sum_{j\in P} b_j$; the internal structure $O_{PP}$ does not enter the calculation. In \textbf{Regime B} values $\bm v_P$ are not observable: the internal topology $O_{PP}$ is then necessary exclusively to estimate $\bm v_P$ through the system
\begin{equation*}
\bm v_P \;=\; (I - O_{PP})^{-1}\!\left(\bm b_P + O_{PO}\bm v_O\right),
\end{equation*}
after which $W(P)$ is still obtained via the boundary cut. This distinction resolves an apparent tension between the irrelevance of internal topology (Regime A) and its instrumental role (Regime B).

The contributions of the paper are: 
(i) a compact proof of the \emph{Cut Theorem} and of the statistical sufficiency of boundary flows;
(ii) two operational algorithms for Regime A and Regime B; 
(iii) diagnostics for double counting and minority interests;
(iv) a critical discussion of assumptions on inputs $(\bm b_P,\bm v_O)$ and their limits.
We also show a numerical example and outline a path of empirical validation on real perimeters. For the role of the observer and the reference system, see Section~\ref{sec:osservatore}.

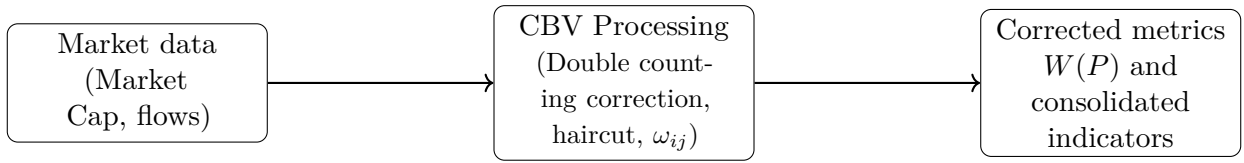
\begin{figure}[htbp]
\centering
\begin{tikzpicture}[
  node distance=3cm,
  box/.style={draw, rounded corners, align=center, minimum width=2.5cm, minimum height=1cm, text width=3.2cm}
]
\node[box] (input) {Market data\\(Market Cap, flows)};
\node[box, right=of input] (process) {CBV Processing\\{\small (Double counting correction, haircut, $\omega_{ij}$)}};
\node[box, right=of process] (output) {Corrected metrics\\$W(P)$ and consolidated indicators};

\draw[->, thick] (input) -- (process);
\draw[->, thick] (process) -- (output);
\end{tikzpicture}
\caption{Conceptual diagram: synthetic scheme of the CBV process.}
\label{fig:cbv-schema1}
\end{figure}

\noindent\textbf{Keywords:} Cut-Based Valuation; Consolidation; Equity ownership networks; Cross-holdings; Minority interests; Perimeter; Input-Output.\\
\textbf{JEL:} C67; G32; M41; E44.

\section{Observer and reference system}\label{sec:osservatore}
The consolidated value $W(P)$ depends on the \emph{perimeter} $P$ and on the \emph{observer’s position}. In \textbf{Regime A}, where $\bm v_P$ is observable, the \emph{Cut Theorem} implies \emph{invariance} with respect to internal topology: with the same perimeter and flows $P\leftrightarrow O$, $W(P)$ does not change as $O_{PP}$ varies. In \textbf{Regime B}, the observer uses the topology $O_{PP}$ \emph{only} to estimate $\bm v_P$; once estimated, the cut invariance holds again. 

This distinction clarifies that the “observer effect” in CBV concerns the \emph{choice of perimeter} and of the \emph{inputs} $(\bm b_P,\bm v_O)$, not the internal structure itself. In particular:
\begin{itemize}
\item Changing $P$ (observer “inside/outside” or broader/narrower perimeter) changes the boundary flows and therefore $W(P)$;
\item With the same $P$ and inputs, the transition from Regime A to Regime B does not modify the valuation rule, but adds an estimation step to make $\bm v_P$ observable.
\end{itemize}
We place the concept here so that it is available before the numerical example and extensions.

\begin{tcolorbox}[title={Operational synthesis of informational regimes (CBV)}, colback=gray!5, colframe=black!50]
\textbf{Regime A (observable internal values).} 
\begin{itemize}
\item Input: $\bm b_P$, $\bm v_O$, \emph{and} observable $\bm v_P$ (e.g.\ market prices).
\item Calculation: $W(P) = \sum_{j\in P} b_j + \sum_{i\in P,k\in O} O_{ik} v_k - \sum_{i\in O,j\in P} O_{ij} v_j$.
\item Property: the internal topology $O_{PP}$ is irrelevant for the calculation (cut invariance).
\item Complexity: proportional to boundary edges $|E_{\text{cut}}|$ (no matrix inversion).
\end{itemize}
\vspace{4pt}
\textbf{Regime B (non-observable internal values).}
\begin{itemize}
\item Input: $\bm b_P$, $\bm v_O$, $O_{PP}$, $O_{PO}$.
\item Estimation: solve $(I - O_{PP})\,\bm v_P = \bm b_P + O_{PO}\,\bm v_O$ to obtain $\bm v_P$.
\item Calculation: apply the same formula as Regime A using the estimated $\bm v_P$.
\item Property: $O_{PP}$ is used only to estimate $\bm v_P$; the final value depends on boundary flows.
\end{itemize}
\end{tcolorbox}

\section{Purpose and context}\label{sec:scope}
The “balance-sheet by balance-sheet” sum in the presence of cross-holdings generates double counting. We formalize a consolidation rule that separates \emph{endogenous} value (due to the internal network) from \emph{exogenous} value (external assets and outside prices) and shows that, to value a system or one of its perimeters, boundary cuts and non-equity internal assets suffice.
For formal limits of applicability of CBV see Section~\ref{sec:scope-limitations}.

\section{Related work and positioning}\label{sec:related}

\noindent Our contribution lies at the intersection of \emph{measurement standards}, \emph{network-based consolidation}, and \emph{observer-dependent valuation}. In this section we compare \emph{Cut-Based Valuation} (CBV) with existing standards and models, highlighting its originality and positioning.

\subsection{Accounting standards and national accounts}
Financial reporting provides different measurement bases and consolidation rules (e.g. IFRS/IAS, US GAAP), while the System of National Accounts (SNA) defines sectoral consolidation and perimeter choices. These literatures explicitly recognize that reported numbers depend on the reporting entity and on the measurement basis (choices similar to those of an observer). See, for example, \cite{IFRS_conceptual_framework, IFRS10, SNA2008}.  

CBV instead applies to \emph{any arbitrary perimeter} (sector, portfolio, sub-network), using as the only necessary information the \emph{boundary flows} $P\leftrightarrow O$. In Regime~A its complexity is \(\mathcal{O}(|E_{\text{cut}}|)\), avoiding matrix inversions and ensuring scalability.

\subsection{Input--Output, cross-holdings, and network clearing}
Input--Output analysis \cite{Leontief1986} formalizes linear propagation through $(I-A)^{-1}$; network clearing models (systemic risk) determine equilibrium payments in cases of default and insolvency \cite{EisenbergNoe2001, RogersVeraart2013, Battiston2012}.  

Analogously, ownership networks could require $(I-O_{PP})^{-1}$; however, CBV shows that such inversion is unnecessary if internal values $\bm v_P$ are known (Regime~A). Regime~B uses $(I-O_{PP})^{-1}$ only as an estimation device when values are unobservable, while maintaining \emph{cut invariance} and disclosure with respect to the observer.

\subsection{Control vs ownership and ultimate control}
Control differs from mere ownership shares; practical rules include majority/more votes, probabilistic weights, and look-through/ultimate beneficiary control measures. We operationalize three families (threshold/majority, Herfindahl look-through, attenuated paths) and provide disclosure parameters ($\tau,\alpha$), aligning with empirical traditions on pyramids/keiretsu and control chains \cite{Vitali2011_network_control, FaccioLang2002, Claessens2000}.

\subsection{Index theory, PPP, and observer dependence}
Price indices are not unique; they depend on baskets and formulas (Laspeyres, Paasche, Fisher, Tornqvist) \cite{Diewert1976}. Cross-country comparisons rely on PPP/ICP conventions \cite{WorldBankICP}. This mirrors our observer tuple $\Omega$ (Units/FX/PPP, Basis, SDF): different observers produce different numbers linked by transformation maps (units/deflators, discounting factors).

\subsection{Asset valuation and stochastic discount factors}
Valuation as an expectation under a stochastic discount factor (SDF) depends on the observer: changing measure/discounting changes values \cite{Cochrane2005, HansenSingleton1983}. Our framework abstracts this through the SDF component of $\Omega$, without committing to a specific asset pricing model.

\subsection{Comparative synthesis}
\begin{table}[H]
\centering
\caption{Positioning relative to representative strands.}
\begin{tabularx}{\linewidth}{|>{\raggedright\arraybackslash}X
                              |>{\raggedright\arraybackslash}X
                              |>{\raggedright\arraybackslash}X
                              |>{\raggedright\arraybackslash}X|}
\hline
\textbf{Strand} & \textbf{Representative references} & \textbf{What they provide} & \textbf{What we add (CBV)} \\
\hline
Accounting standards & IFRS~10; SNA~2008 & Consolidation criteria; sectoral boundaries; multiple bases & Formal observer $\Omega$; cut sufficiency; mandatory disclosure of artifacts \\
\hline
IO \& clearing & Leontief; Eisenberg--Noe; Rogers--Veraart & Linear/fixed-point propagation; hierarchies; debt clearing & Conditioning on the observer; valuation via boundary; computational costs $\mathcal{O}(|E_{\text{cut}}|)$ \\
\hline
Ownership networks & La Porta et al.; Faccio--Lang; Vitali et al. & Ownership/ultimate control metrics; pyramids/keiretsu & Three operational control rules with parameters and algorithms \\
\hline
Indices and PPP & Fisher; Diewert; ICP/World Bank & Observer dependence in measurement protocols & Transformation laws between observers in the valuation context \\
\hline
Asset valuation (SDF) & Cochrane; Hansen--Singleton & Observer dependence via discount factors & SDF as part of $\Omega$; integration with consolidation \\
\hline
\end{tabularx}
\end{table}
\subsection{What is new here}
\begin{enumerate}
  \item \textbf{Formalization of the observer.} We integrate measurement choices (perimeter, basis, units/PPP, SDF, information, control) into a single tuple $\Omega$, necessary for well-posed valuation statements.
  \item \textbf{Cut Theorem as sufficiency/invariance.} We prove that, \emph{conditionally on $\Omega$}, the consolidated value depends only on boundary statistics and node primitives; internal reconnections are gauge-equivalent.
  \item \textbf{Operational standards.} We standardize the artifacts \emph{Perimeter-of-Validity} and \emph{Cut Summary} (Section~\ref{sec:standards}), with algorithms and complexity guarantees.
  \item \textbf{Implementation-ready control.} We provide three operational control rules (\S~\ref{sec:ownershipcontrol}) with pseudocode and disclosure parameters, enabling empirical replications and sensitivity analysis.
\end{enumerate}

\subsection{Note on originality}
The originality of CBV lies in the unique combination of:
\begin{itemize}
  \item a \emph{cut theorem} that makes internal topology irrelevant (Regime~A);
  \item the identification of the minimal sufficient statistic (the cut itself);
  \item the formalization of a dual algorithmic rule (Regime A/B) valid for any arbitrary perimeter;
  \item the integration of observer choices into a single formal structure $\Omega$.
\end{itemize}
This synthesis --- cut theorem + sufficient statistic + observer + algorithmic recipe --- provides a parsimonious and unified rule that, to our knowledge, has not been presented in integrated form in the existing literature.

\section{Definitions and notation}\label{sec:notation}
Let $G=(V,E)$ be an ownership graph with $|V|=n$. Each node $i\in V$ is an entity (company, fund, institution). A directed edge $(i\to j)$ with weight $O_{ij}\in[0,1]$ indicates the share of $j$ owned by $i$; by definition, column $j$ of $O$ (the allocation of $j$’s equity among owners) ideally sums to $1$. \medskip
\noindent\textbf{Equity value.} $v_i\in\R_+$ is the equity value of $i$ (market/DCF/book). We denote $\bm v=(v_1,\dots,v_n)^\T$.
\noindent\textbf{External assets.} $b_i\in\R$ is the \emph{non-internal equity} value of $i$: real/intangible assets net of liabilities to the outside and of claims senior to equity. We denote $\bm b=(b_1,\dots,b_n)^\T$.

\noindent\textbf{Perimeter.} Given a partition $V=P\cup O$ with $P\cap O=\varnothing$, we call \emph{cut} the edges $P\to O$ and $O\to P$.
\noindent\textbf{Structural identity.} With the above convention, value satisfies
\begin{equation}\label{eq:structure}
\bm v \;=\; \bm b \;+\; O\,\bm v,
\end{equation}
where $(O\,\bm v)_i=\sum_{j} O_{ij}\,v_j$ is the value of participations held by $i$.

\section{Assumptions}\label{sec:assumptions}
(i) $\bm b$ is already net of debts and of external claims senior to equity (multi-layer extensions in \S\ref{sec:related}).  
(ii) $O_{ij}$ represents the effective economic share (treasury shares treated separately).  
(iii) The values of entities in $O$ are exogenous/observable: $\bm v_O$ (market prices or scenarios).

\paragraph{Input assumptions and practical cases.} 
CBV requires as fundamental inputs $\bm b_P$ (internal non-equity assets) and $\bm v_O$ (external equity values).
\begin{itemize}
\item \textbf{Choice of $\bm b_P$:} It may derive from book, market, or liquidation values. For example, for an industrial conglomerate, $\bm b_P$ may include plants valued at historical cost (\emph{book value}) or at current realizable prices (\emph{market value}). The choice directly affects $W(P)$ and can be tested with sensitivity analysis.
\item \textbf{Choice of $\bm v_O$:} If observable, they are typically market prices of entities in $O$. In illiquid contexts or with external cross-holdings, such prices may be distorted. In that case, it is advisable to apply corrections (e.g.\ illiquidity discount) or estimate values through industry multiples or discounted cash flows.
\end{itemize}

\paragraph{Multi-level example on $\bm b_P$.}
Consider a company $j$ with:
\begin{itemize}
\item $b_j^{(1)}$: net assets after senior debt,
\item $b_j^{(2)}$: residual value for preferred shares,
\item $b_j^{(3)}$: portion allocated to common equity.
\end{itemize}
Then $v_j = b_j^{(3)}$ where
\begin{equation*}
b_j^{(3)} = b_j^{(2)} - \sum_{\text{debt}} \mathrm{obligations} - \sum_{\text{pref}} \mathrm{dividends}.
\end{equation*}
This layered approach allows adapting CBV to complex capital structures.

\paragraph{Scenarios with non-observable $\bm v_O$.}
In the absence of market data (e.g.\ unlisted companies), conservative estimates can be adopted based on:
\begin{itemize}
\item Adjusted book values,
\item Industry multiples (P/E, EV/EBITDA) applied to available indicators,
\item Discounted cash flow models.
\end{itemize}
It is recommended to document the chosen approach and provide confidence intervals or alternative scenarios to assess the sensitivity of $W(P)$ to these estimates.

\paragraph{Numerical practical case.}
To illustrate the impact of input choices on $W(P)$, consider a perimeter $P$ with two entities and an outside $O$ with one entity.
\begin{align*}
O_{PO} &= \begin{pmatrix} 0.10 \\ 0.05 \end{pmatrix}, & O_{OP} &= \begin{pmatrix} 0.20 & 0.10 \end{pmatrix}, \\
O_{PP} &= \begin{pmatrix} 0 & 0.15 \\ 0.10 & 0 \end{pmatrix}.
\end{align*}

\textbf{Scenario 1:} $\bm b_P = (40, 30)^\T$, $\bm v_O = (50)$.
\begin{align*}
\text{Base} &= 40+30 = 70, \\
P\to O &= 0.10\cdot 50 + 0.05\cdot 50 = 7.5, \\
O\to P &= 0.20\cdot 40 + 0.10\cdot 30 = 11, \\
W(P) &= 70 + 7.5 - 11 = 66.5.
\end{align*}

\textbf{Scenario 2:} $\bm b_P = (50, 25)^\T$, $\bm v_O = (80)$.
\begin{align*}
\text{Base} &= 50+25 = 75, \\
P\to O &= 0.10\cdot 80 + 0.05\cdot 80 = 12, \\
O\to P &= 0.20\cdot 50 + 0.10\cdot 25 = 12.5, \\
W(P) &= 75 + 12 - 12.5 = 74.5.
\end{align*}

The comparison shows how variations in internal assets ($\bm b_P$) and in external values ($\bm v_O$) are directly reflected in variations of $W(P)$, even with the same ownership structure ($O$ fixed).

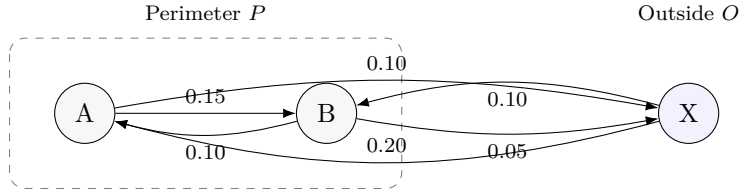
\begin{figure}[H]
\centering
\begin{tikzpicture}[>=LaTeX, node distance=24mm]
  \node[circle, draw, minimum size=8mm, font=\small, fill=black!3] (A) {A};
  \node[circle, draw, minimum size=8mm, font=\small, fill=black!3, right=of A] (B) {B};
  \node[circle, draw, minimum size=8mm, font=\small, fill=blue!5, right=40mm of B] (X) {X};
  \draw[rounded corners=6pt, dashed, gray] ($(A)+(-10mm,10mm)$) rectangle ($(B)+(10mm,-10mm)$);
  \node[font=\scriptsize, above] at ($(A)!0.5!(B)+(0,11mm)$) {Perimeter $P$};
  \node[font=\scriptsize, above] at ($(X)+(0,11mm)$) {Outside $O$};
  \draw[->] (A) -- node[font=\scriptsize, above]{0.15} (B);
  \draw[->] (B) to[bend left=15] node[font=\scriptsize, below]{0.10} (A);
  \draw[->] (A) to[bend left=10] node[font=\scriptsize, above]{0.10} (X);
  \draw[->] (B) to[bend right=10] node[font=\scriptsize, below]{0.05} (X);
  \draw[->] (X) to[bend left=15] node[font=\scriptsize, above]{0.20} (A);
  \draw[->] (X) to[bend right=15] node[font=\scriptsize, below]{0.10} (B);
\end{tikzpicture}
\caption{Practical case: three-node network with weights $O_{PP}$, $O_{PO}$ and $O_{OP}$ as defined.}
\label{fig:case_study_ipotesi}
\end{figure}

\section{Consolidated perimeter value: Cut Theorem}\label{sec:cutsection}
\begin{definition}[Consolidated value]
Given a perimeter $P\subseteq V$, we define $W(P)$ as the value attributable to \emph{internal} owners $P$ over all entities (inside and outside).
\end{definition}
\begin{theorem}[Cut-Based Consolidation]\label{thm:cut}
For any perimeter $P$ with complement $O$, it holds that
\begin{equation}\label{eq:cut}
W(P)\;=\;\sum_{j\in P} b_j \;+\; \sum_{i\in P,\;k\in O} O_{ik}\,v_k \;-\; \sum_{i\in O,\;j\in P} O_{ij}\,v_j.
\end{equation}
\end{theorem}


\begin{proof}[Accounting sketch]
The formal proof is reported in Appendix~\ref{app:formal-cut-proof}. 
Here we simply note that formula \eqref{eq:cut} emerges by systematically eliminating intra-perimeter items and measuring only boundary quantities (internal bases, assets $P\!\to\!O$, and external minorities $O\!\to\!P$).
\end{proof}

\subsection{Illustrative T-account of the Cut Theorem (didactic)}\label{subsec:t-account}
To fix ideas (Regime A, observable internal values), consider a perimeter $P$ with two entities $A,B$ and one external $o_1$.
Data:
\[
b_A=10,\quad b_B=5;\qquad v_{o_1}=100,\quad v_A=50,\quad v_B=30;
\]
\[
O_{A,o_1}=0.30,\quad O_{B,o_1}=0;\qquad O_{o_1,A}=0.20,\quad O_{o_1,B}=0.10.
\]
Then
\[
\sum_{j\in P}b_j=15,\qquad
\sum_{i\in P,k\in O}O_{ik}v_k=0.30\cdot 100=30,\qquad
\sum_{i\in O,j\in P}O_{ij}v_j=0.20\cdot 50 + 0.10\cdot 30=13,
\]
and therefore
\[
W(P)=15+30-13=32.
\]

\begin{table}[H]
\centering
\caption{Consolidated T-account of perimeter $P$ (didactic example)}
\begin{tabular}{p{0.46\linewidth} p{0.46\linewidth}}
\toprule
\multicolumn{1}{c}{\textbf{Assets}} & \multicolumn{1}{c}{\textbf{Liabilities \& Equity}}\\
\midrule
Internal non-equity assets ($\sum b_P$) & External minorities $O\!\to\!P$ \\
\quad $10+5=15$ & \quad $0.20\,v_A + 0.10\,v_B = 10 + 3 = 13$ \\
Participations $P\!\to\!O$ & Consolidated equity $W(P)$ \\
\quad $0.30\,v_{o_1}=30$ & \quad \textbf{32} \\
\midrule
Total Assets $= 45$ & Total Liabilities $+$ Equity $=45$ \\
\bottomrule
\end{tabular}
\end{table}

\paragraph{Note.} Any internal chain/cycle in $P$ (cross-holdings between $A$ and $B$) is eliminated in consolidation and does not affect $W(P)$: only the boundary remains.

\begin{corollary}[Interiority]\label{cor:interior}
Any modification of the internal architecture $O_{PP}$ (chains, cycles, holding) that does not alter $\bm b_P$ and the cuts $O_{PO},O_{OP}$ leaves $W(P)$ unchanged.
\end{corollary}

\begin{corollary}[Total closed system]\label{cor:closed}
If $P=V$ (no outside), then $W(V)=\sum_{j\in V} b_j$, independent of any internal participation.
\end{corollary}

\paragraph{Minimal data for \eqref{eq:cut}.}
Equation \eqref{eq:cut} requires only: $\bm b_P$, external prices/values $\bm v_O$, and the sole cut edges $O_{PO},O_{OP}$. The internal topology $O_{PP}$ is irrelevant for $W(P)$.


\subsection{Axioms for the cut-based consolidation functional}\label{subsec:axioms-uniqueness}
\paragraph{Setup.}
Let $V$ be the finite set of nodes (entities). For a perimeter $P\subseteq V$ denote $O=V\setminus P$.
The observer $\Omega$ fixes units/currency and, if necessary, a pricing functional (SDF) to value bases and payoffs.
We denote by $\bm b\in\mathbb R^{|V|}$ the \emph{bases} per node and by $\bm v\in\mathbb R^{|V|}$ the \emph{values}.
The matrix $O\ge 0$ collects participations; we write the blocks $O_{PP}, O_{PO}, O_{OP}$ with the usual meaning.
The functional of interest is $W:\;2^V\to\mathbb R$, $P\mapsto W(P)$.

\paragraph{Axioms (A)--(E).}
\begin{enumerate}[label=(\Alph*)]
\item \textbf{Linearity.} $W$ is linear in the bases and in the payoffs valued by $\Omega$.
\item \textbf{Neutrality to internal transfers.} Reallocations/participations \emph{internal} to $P$ that leave boundary objects unchanged do not change $W(P)$.
\item \textbf{Cut-invariance.} $W(P)$ depends only on $(\bm b_P,\bm v_P,\bm v_O,O_{PO},O_{OP})$, not on the internal topology $O_{PP}$.
\item \textbf{Aggregative consistency (nested perimeters).} For $P\subseteq Q$ one has
\[
W(Q)=W(P)+W(Q\!\setminus\!P)-\big(\text{net value on the internal cut } P\leftrightarrow Q\!\setminus\!P\big).
\]
\item \textbf{``Unit-free'' compatibility with respect to the observer.}
Under invertible linear transformations of units/currency/PPP/SDF fixed by $\Omega$, $W$ transforms equivariantly (only up to the same scale).
\end{enumerate}

\subsection{Representation theorem and uniqueness of CBV}\label{subsec:cbv-uniqueness}

\begin{theorem}[Representation and uniqueness]\label{thm:cbv-uniqueness2}
Under axioms \textup{(A)--(E)}, for any perimeter $P$ with complement $O$ the identity holds
\begin{equation}\label{eq:cbv-repr}
W(P)\;=\;\sum_{j\in P} b_j
\;+\;\sum_{i\in P,\;k\in O} O_{ik}\,v_k
\;-\;\sum_{i\in O,\;j\in P} O_{ij}\,v_j,
\end{equation}
which coincides with the formula of the \emph{Cut Theorem} \eqref{eq:cut}. Moreover, $W$ is \emph{unique} with such properties.
\end{theorem}

\begin{proof}[Proof idea]
(1) \emph{Reduction to the cut.} By (B)–(C) $W(P)$ cannot depend on $O_{PP}$: only the boundary information $(\bm b_P,\bm v_P,\bm v_O,O_{PO},O_{OP})$ remains.

(2) \emph{Linearity.} By (A) the dependence is linear: $W(P)$ is a linear combination of the internal base and of the portfolios \emph{across} the cut, in both directions $P\!\to\!O$ and $O\!\to\!P$.

(3) \emph{Base cases/normalization.}
(i) Network without participations ($O\equiv 0$): $W(P)=\sum_{j\in P} b_j$ fixes the coefficient of the base at $1$.  
(ii) Only assets of $P$ on $O$ (zero bases, $O_{OP}\!=\!0$): $W(P)=\sum_{i\in P,k\in O} O_{ik} v_k$ fixes the coefficient of the flow $P\!\to\!O$ at $1$.  
(iii) Only shares of $O$ on $P$ (zero bases, $O_{PO}\!=\!0$): consolidation elimination imposes the \emph{minus sign} and coefficient $1$ on the flow $O\!\to\!P$.

(4) \emph{Modularity.} The form \eqref{eq:cbv-repr} satisfies aggregative consistency (D) by additive decomposition and subtraction of the net value on the internal cut.

(5) \emph{Equivariance to the observer.} With (E) any weights different from unit coefficients would break the “unit-free’’ transformation. The base cases fix them at $+1,+1,-1$. No other linear boundary terms are compatible with (B)–(D). Uniqueness follows.
\end{proof}

\paragraph{Remark.}
Equation \eqref{eq:cbv-repr} is \emph{completely cut-based}: it eliminates $O_{PP}$ and measures only (i) internal bases, (ii) assets of $P$ on $O$, (iii) \emph{external minorities} $O\!\to\!P$ with negative sign; thus it is invariant to internal restructuring within the perimeter.


\subsection{Axiomatization and uniqueness of the CBV functional}\label{subsec:axioms-cbv}

\paragraph{Setup.}
Let $V$ be the finite set of nodes. For a perimeter $P\subseteq V$ denote $O=V\setminus P$.
The observer $\Omega$ fixes units/currency and, if necessary, a pricing functional (SDF).
We denote by $\bm b\in\mathbb R^{|V|}$ the \emph{bases} and by $\bm v\in\mathbb R^{|V|}$ the \emph{values}. The participation matrix $O\ge 0$ is blocked as $O_{PP}, O_{PO}, O_{OP}$.

\paragraph{Axioms.}
\begin{enumerate}[label=(\Alph*)]
\item \textbf{Linearity}: $W$ is linear in the bases and in the payoffs valued by $\Omega$.
\item \textbf{Neutrality to internal transfers}: reallocations/links \emph{internal} to $P$ that leave boundary objects unchanged do not change $W(P)$.
\item \textbf{Cut-invariance}: $W(P)$ does not depend on $O_{PP}$ but only on $(\bm b_P,\bm v_P,\bm v_O,O_{PO},O_{OP})$.
\item \textbf{Aggregative consistency (nested)}: for $P\subseteq Q$,
\[
W(Q)=W(P)+W(Q\!\setminus\!P)-\big(\text{net value on the internal cut } P\leftrightarrow Q\!\setminus\!P\big).
\]
\item \textbf{``Unit-free'' compatibility}: changes of units/FX/PPP/SDF (invertible linear transformations) act equivariantly on $W$.
\end{enumerate}

\begin{theorem}[Representation and uniqueness]\label{thm:cbv-uniqueness}
Under \textup{(A)--(E)}, for every $P$ with complement $O$ one has
\begin{equation}\label{eq:cbv-repr-main}
W(P)=\sum_{j\in P} b_j \;+\; \sum_{i\in P,k\in O} O_{ik}\,v_k \;-\; \sum_{i\in O,j\in P} O_{ij}\,v_j,
\end{equation}
that is, the formula of the \emph{Cut Theorem} \eqref{eq:cut}. Moreover $W$ is \emph{unique} with such properties.
\end{theorem}

\begin{proof}[Proof idea]
(B)--(C) eliminate $O_{PP}$, reducing $W(P)$ to a linear function (A) of the internal base and of the portfolios crossing the cut in both directions $P\!\to\!O$ and $O\!\to\!P$.
Three base cases normalize the coefficients to $(+1,+1,-1)$: (i) network without participations; (ii) only assets of $P$ on $O$; (iii) only \emph{external minorities} on $P$.
The resulting form satisfies modularity (D). Different coefficients would break equivariance (E). Uniqueness follows.
\end{proof}

\paragraph{Remark.}
Equation \eqref{eq:cbv-repr-main} is \emph{completely boundary-based}: it measures (i) internal bases, (ii) assets $P\!\to\!O$, (iii) \emph{external minorities} $O\!\to\!P$ with negative sign; thus it is invariant to internal restructuring.

\subsection{Robustness bound (audit-first)}\label{subsec:robust-bounds}
Let $\mathbf{1}_P\in\R^{|P|}$ be the vector of all $1$s and, for $p\in[1,\infty]$, let $q$ be the dual norm ($1/p+1/q=1$).
For a matrix $A$, denote $\|A\|_{q\leftarrow p}:=\sup_{x\neq 0}\frac{\|Ax\|_q}{\|x\|_p}$.

\begin{theorem}[CBV robustness bound]\label{thm:robust-main}
If $\|\Delta \bm v_O\|_p\le \varepsilon$ and $\|\Delta \bm b_P\|_p\le \eta$, then the variation of $W(P)$ satisfies
\begin{equation}\label{eq:robust-generic}
|\Delta W(P)| \;\le\; \|\mathbf{1}_P\|_q \,\eta \;+\; \big\|\mathbf{1}_P^\top O_{PO}\big\|_q \,\varepsilon 
\;\le\; \|\mathbf{1}_P\|_q \,\eta \;+\; \|\mathbf{1}_P\|_q \,\|O_{PO}\|_{q\leftarrow p}\, \varepsilon.
\end{equation}
\end{theorem}

\begin{proof}
From \eqref{eq:cbv-repr-main} with $\Delta \bm v_P=0$ (internal prices not perturbed) we have
$\Delta W = \mathbf{1}_P^\top \Delta \bm b_P + \mathbf{1}_P^\top O_{PO}\,\Delta \bm v_O$.
By Hölder: $|\mathbf{1}_P^\top \Delta \bm b_P|\le \|\mathbf{1}_P\|_q\|\Delta \bm b_P\|_p$ and 
$|\mathbf{1}_P^\top O_{PO}\,\Delta \bm v_O|\le \|\mathbf{1}_P^\top O_{PO}\|_q \|\Delta \bm v_O\|_p$.
The second inequality follows from $\|\mathbf{1}_P^\top O_{PO}\|_q \le \|\mathbf{1}_P\|_q \|O_{PO}\|_{q\leftarrow p}$.
\end{proof}

\paragraph{Practical corollaries (choose the norm).}
\begin{itemize}
\item \textbf{$p\!=\!1$ (dual $q\!=\!\infty$):}
\[
|\Delta W(P)| \;\le\; \eta \;+\; \Big(\max_{k\in O}\sum_{i\in P} O_{ik}\Big)\,\varepsilon.
\]
\emph{Cut geometry:} maximum ``weighted degree’’ outgoing to $O$.
\item \textbf{$p\!=\!2$ (dual $q\!=\!2$):}
\[
|\Delta W(P)| \;\le\; \sqrt{|P|}\,\eta \;+\; \sqrt{|P|}\,\|O_{PO}\|_2\,\varepsilon.
\]
\emph{Cut geometry:} spectral norm of the block $O_{PO}$.
\item \textbf{$p\!=\!\infty$ (dual $q\!=\!1$):}
\[
|\Delta W(P)| \;\le\; |P|\,\eta \;+\; \Big(\sum_{i\in P,k\in O} O_{ik}\Big)\,\varepsilon.
\]
\emph{Cut geometry:} total weight of edges $P\!\to\!O$.
\end{itemize}

\paragraph{Extension (Regime B, estimated internal prices).}
If internal valuations vary through
$\Delta \bm v_P=(I-O_{PP})^{-1}\!\big(\Delta \bm b_P + O_{PO}\,\Delta \bm v_O\big)$
(well-defined if $\rho(O_{PP})<1$), then
\[
|\Delta W(P)| \;\le\; \|\mathbf{1}_P\|_q \eta + \|\mathbf{1}_P^\top O_{PO}\|_q \varepsilon
+ \|\bm\delta\|_q\,\|(I-O_{PP})^{-1}\|_{p\to p}\,\Big(\eta + \|O_{PO}\|_{q\leftarrow p}\varepsilon\Big),
\]
where $\bm\delta:=O_{OP}^\top \mathbf{1}_O$ is the vector of \emph{direct minorities} on $P$.
This bound links the impact of errors to (i) strength of the cut $P\!\leftrightarrow\!O$, (ii) internal ``multiplier’’ $(I-O_{PP})^{-1}$, (iii) minority exposure $\bm\delta$.

\section{Transformation laws between observers}\label{sec:transform}
Let $\Omega=(P,\mathrm{Basis},\mathrm{Units},\mathrm{FX/PPP},\mathrm{SDF},\mathcal I,\mathcal C)$ and $\Omega'$ differ in one or more components.
We list the operational maps connecting $V_{\Omega}$ and $V_{\Omega'}$ when applicable.

\subsection{Units/FX/PPP (linear scaling)}
If $\Omega'$ differs only by a change of units/FX/PPP represented by a positive (or diagonal) scalar $\kappa$, then
\begin{equation}\label{eq:units-scaling}
V_{\Omega'}(G) \;=\; \kappa \, V_{\Omega}(G),\qquad
(X'_{PO},X'_{OP}) \;=\; \kappa\,(X_{PO},X_{OP}),\qquad
(\bm v'_P,\bm v'_O)=\kappa\,(\bm v_P,\bm v_O).
\end{equation}
Disclosure: report $\kappa$ (FX source, deflator/PPP, date).

\subsection{Discounting (change of measure / SDF)}
Let $\Omega'$ differ only by the SDF; denote by $m$ and $m'$ the one-period discount factors under $\Omega$ and $\Omega'$, with $m' = m\cdot \Lambda$ and $\Lambda>0$ the Radon--Nikod\'ym derivative (\emph{change of measure}). Then, for any cash-flow vector $X$ (including cut terms),
\begin{equation}\label{eq:sdf-change}
\mathbb E_{\Omega'}[m' X] \;=\; \mathbb E_{\Omega}[m \Lambda X].
\end{equation}
Operationally, $V_{\Omega'}$ is obtained by re-weighting/re-discounting the boundary terms in \eqref{eq:units-scaling} according to $\Lambda$; make the curve/source explicit.

\subsection{Perimeter variation}
If $P$ changes to $P'$, valuation must be recalculated on the \emph{new cut} $P'\leftrightarrow O'$, using the same rules. In general no closed-form link exists, but the Cut Theorem guarantees sufficiency of the new boundary statistics. Provide both PoV tables to ensure comparability.

\subsection{Change of control rule}
Changing $\mathcal C$ (and hence $\omega$) may affect $O_{PP}$ and perimeter membership. We recommend reporting a sensitivity table across Options A/B/C (Section~\ref{sec:ownershipcontrol}) and parameters ($\tau,\alpha$), together with the corresponding variations of the Cut Summary.

\section[Reporting standards: Perimeter of validity \& Cut summary]{%
Reporting standards: Perimeter of validity \& Cut summary}
\label{sec:standards}
\noindent To ensure comparability, verifiability, and reproducibility, each valuation must be accompanied by two disclosure artifacts: the \emph{Perimeter of validity} (which fixes the observer $\Omega$) and the \emph{Cut Summary} (boundary statistics). The following templates are normative; fields marked \textbf{(req)} are mandatory.

\subsection{Perimeter of validity (table template)}\label{subsec:template-pov}
\begin{table}[H]
\centering
\caption{Perimeter of validity (PoV). Fields marked (req) are mandatory.}
\begin{tabularx}{\textwidth}{lX}
\toprule
\textbf{Field} & \textbf{Value / Notes}\\
\midrule
Perimeter $P$ \textbf{(req)} & Included entities; exclusions; SPE/JV treatment. \\
Basis \textbf{(req)} & Fair value / Historical cost / Realizable; impairment policy. \\
Units \& Date \textbf{(req)} & Currency (e.g. EUR); measurement date (YYYY-MM-DD). \\
FX/PPP & FX source (e.g. ECB); PPP/deflator (e.g. OECD ICP; HICP 2025Q2); chain-linking. \\
SDF & Curve source (e.g. OIS/EUR swap); measure (risk-neutral/physical); horizons. \\
Informational regime $\mathcal I$ \textbf{(req)} & Regime A (observed boundary) or Regime B (use of $O_{PP}$ to estimate $\bm v_P$). \\
Control rule $\mathcal C$ \textbf{(req)} & Option A/B/C; parameters ($\tau$, $\alpha$), normalization, reachability depth, cycle attenuation. \\
Tolerances & Rounding threshold $\tau$, solver $\varepsilon$, $K_{\max}$. \\
Data sources & Dataset, providers, time-stamp, licenses. \\
Assumptions & Any deviations from default assumptions [A1--A8] in Section~\ref{sec:assumptions}. \\
Versioning & Code repo commit, data snapshot ID, environment hash. \\
Notes & Free-text clarifications, exclusions, warnings. \\
\bottomrule
\end{tabularx}
\end{table}

\subsection{Cut summary (table templates)}\label{subsec:template-cut}
\paragraph{$P\to O$ edges (outgoing).}

\begin{table}[H]
\centering
\caption{Cut summary: edges $P\to O$.}
\begin{tabularx}{\textwidth}{@{}llXlrl@{}}
\toprule
\textbf{From $i \in P$} & \textbf{To $k \in O$} & \textbf{Type} & \textbf{Currency} & \textbf{Amount} & \textbf{Notes/Imputation} \\
\midrule
$a$ & $o_1$ & equity / debt / derivative / cashflow & EUR & 123.45 & --- \\
$\cdots$ & $\cdots$ & $\cdots$ & EUR & $\cdots$ & sectoral median imputation \\
\bottomrule
\end{tabularx}
\end{table}

\paragraph{$O\to P$ edges (incoming).}
\begin{table}[H]
\centering
\caption{Cut summary: edges $O\to P$.}
\begin{tabularx}{\textwidth}{@{}llXlrl@{}}
\toprule
\textbf{From $k\in O$} & \textbf{To $j\in P$} & \textbf{Type} & \textbf{Currency} & \textbf{Amount} & \textbf{Notes/Imputation} \\
\midrule
$o_1$ & $b$ & equity/\allowbreak debt/\allowbreak derivative/\allowbreak cashflow & EUR & 67.89 & --- \\
$\cdots$ & $\cdots$ & $\cdots$ & EUR & $\cdots$ & --- \\
\bottomrule
\end{tabularx}
\end{table}

\paragraph{Node primitives and totals.}
\begin{table}[H]
\centering
\caption{Node primitives and totals.}
\begin{tabular}{@{}llr@{}}
\toprule
\textbf{Item} & \textbf{Definition} & \textbf{Value} \\
\midrule
$T_{\text{out}}$ & $\sum_{i\in P,k\in O}\Pi_\Omega(X_{ik})$ & 123.45 \\
$T_{\text{in}}$ & $\sum_{k\in O,j\in P}\Pi_\Omega(X_{kj})$ & 67.89 \\
$V_\Omega(G)$ & Consolidated value (boundary functional) & 55.56 \\
\midrule
$\bm v_P$ & Node primitives inside $P$ (Regime B after estimation) & see JSON \\
$\bm v_O$ & Node primitives in $O$ (if used by $\mathcal F_\Omega$) & see JSON \\
Hedge $h$ on $O$ & Hedging vector of boundary exposures (optional) & see JSON \\
\bottomrule
\end{tabular}
\end{table}

\subsection{Minimum required fields and formatting}\label{subsec:minreq}
\begin{itemize}
  \item \textbf{PoV:} $P$, Basis, Units \& Date, $\mathcal I$, $\mathcal C$ (with parameters), tolerances. Optional but recommended: FX/PPP, SDF, versioning.
  \item \textbf{Cut Summary:} both edge tables ($P\to O$ and $O\to P$) with currency and amounts, totals $T_{\text{out}},T_{\text{in}}$, $V_\Omega(G)$ and—if Regime B—a table or a JSON attachment for $\bm v_P$.
  \item \textbf{Formats:} currency as ISO (EUR, USD), dates as ISO (YYYY-MM-DD); amounts in the Units defined by the PoV; declare imputation/rounding.
\end{itemize}

\subsection{Consistency checks and invariants}\label{subsec:checks}
Before releasing a valuation:
\begin{enumerate}
  \item \textbf{Cross-ref}: section labels and references compile; \texttt{PoV} refers to $\Omega$ consistently.
  \item \textbf{Cut sufficiency}: recompute after any internal reconnection of $X_{PP}$; $V_\Omega(G)$ must remain invariant (Theorem~\ref{thm:cutTheoreme}).
  \item \textbf{Regime B}: verify that $\rho(O_{PP})<1$; document method (solver/series/Krylov), tolerance $\varepsilon$ and iterations.
  \item \textbf{Units/FX/PPP/SDF}: consistent with PoV and pricing operator $\Pi_\Omega$.
  \item \textbf{Edge totals}: reconcile $T_{\text{out}},T_{\text{in}}$ with the edge tables; check row counts.
  \item \textbf{Disclosure}: include JSON artifacts from Section~\ref{sec:algorithms}; archive together with LaTeX/PDF.
\end{enumerate}

\subsection{Quick checklist for auditor}\label{subsec:auditor}
\begin{itemize}
  \item Is PoV present, complete, consistent with the text? Parameters $\tau,\alpha$ declared?
  \item Regime A/B identified and justified? In Regime B: proof of invertibility and solver log?
  \item Do the sums in the two cut tables match the totals? Currency consistent with Units?
  \item Are imputations/thresholds marked and explained?
  \item Is Hedge (if present) recalculable from cut data?
  \item Are versioning and data sources replicable (commit/data snapshot)?
\end{itemize}

\subsection{Design principles}
Algorithms are stated \emph{conditionally on a fixed observer} $\Omega$ (Section~\ref{sec:osservatore}). 
Valuation reduces to boundary statistics thanks to the Cut Theorem (Section~\ref{sec:cutsection}). 
We separate: (i) the computation/estimation of \emph{node primitives} $\bm v_P$ when required (Regime~B), and (ii) the evaluation of the \emph{boundary functional} $\mathcal F_\Omega$ on the cut.

\label{subsec:algA}
\subsection{Regime A: direct cut evaluation}
Input: cut matrices $(X_{PO},X_{OP})$, observed node primitives $(\bm v_P,\bm v_O)$ if required by the chosen $\mathcal F_\Omega$, and pricing/aggregation operators fixed by $\Omega$ (Units/FX/PPP, SDF).
\begin{algorithm}[H]
\caption{Regime A — CBV via direct cut evaluation}
\begin{algorithmic}[1]
\Require $(X_{PO},X_{OP})$, optional $(\bm v_P,\bm v_O)$, pricing operator $\Pi_\Omega$, tolerance $\tau$
\State \textbf{Validate} Units/FX/PPP and date; drop entries with $|x|<\tau$; ensure cut endpoints respect the partition $P\leftrightarrow O$
\State \textbf{Aggregate} outgoing and incoming boundary exposures: \\
$\quad T_{\text{out}}\gets \sum_{i\in P,k\in O} \Pi_\Omega\big(X_{ik}\big)$,\quad
$T_{\text{in}}\gets \sum_{k\in O,j\in P} \Pi_\Omega\big(X_{kj}\big)$
\State \textbf{Combine} with node primitives if needed (e.g. add/subtract cash flows valued on the boundary, book values)
\State \textbf{Return} $V_\Omega(G)=\mathcal F_\Omega(T_{\text{out}},T_{\text{in}},\bm v_P,\bm v_O)$
\end{algorithmic}
\end{algorithm}
\noindent \textbf{Complexity.} $\mathcal O(|E_{\mathrm{cut}}|)$ evaluations of $\Pi_\Omega$ plus linear aggregations; independent of $|P|$ (Theorem~\ref{thm:cutTheoreme}).

\subsection{Regime B: internal estimation + cut evaluation}\label{subsec:algB}
Input: internal matrix $O_{PP}$, primitive vector $\bm\pi_P$, cut matrices $(X_{PO},X_{OP})$, and $\Omega$.
\begin{algorithm}[H]
\caption{Regime B — estimate $\bm v_P$ then evaluate the cut}
\begin{algorithmic}[1]
\Require $O_{PP}$, $\bm \pi_P$, $(X_{PO},X_{OP})$, pricing operator $\Pi_\Omega$, tolerance $\varepsilon$, maximum iterations $K_{\max}$
\State \textbf{Check} $\rho(O_{PP})<1$ (e.g. power method upper bound); otherwise stop or rescale $O_{PP}\leftarrow \beta O_{PP}$ with $\beta<1/\|O_{PP}\|$ and disclose
\State \textbf{Solve} $(I-O_{PP})\bm v_P=\bm \pi_P$:
  \begin{itemize}
    \item direct sparse solver (preferred for moderate size);
    \item or \emph{Neumann series}: $\bm v_P^{(0)}\gets \bm\pi_P$, $\bm v_P^{(t+1)}\gets \bm \pi_P + O_{PP}\bm v_P^{(t)}$ until $\|\bm v_P^{(t+1)}-\bm v_P^{(t)}\|_\infty<\varepsilon$ or $t=K_{\max}$;
    \item or Krylov (GMRES/BiCGSTAB) with preconditioning.
  \end{itemize}
\State \textbf{Evaluate} the boundary as in Algorithm~\ref{sec:algorithms} (Regime A) using $(X_{PO},X_{OP})$ and $\Pi_\Omega$
\State \textbf{Return} $V_\Omega(G)=\mathcal F_\Omega(T_{\text{out}},T_{\text{in}},\bm v_P,\bm v_O)$
\end{algorithmic}
\end{algorithm}
\noindent \textbf{Complexity.} Dominated by solving $(I-O_{PP})\bm v_P=\bm \pi_P$. For sparse $O_{PP}$ with $m$ nonzero entries, each Neumann iteration costs $\mathcal O(m)$; the number of iterations scales with $\frac{1}{1-\rho(O_{PP})}$ in well-conditioned cases. The subsequent cut evaluation remains $\mathcal O(|E_{\mathrm{cut}}|)$.

\subsection{Implementation warnings and best practices}
\begin{itemize}
  \item \textbf{Sparsity}: store $O_{PP}$ and cut matrices in CSR/CSC; avoid dense multiplications.
  \item \textbf{Threshold}: zero out entries $<\tau$ in absolute value; disclose $\tau$.
  \item \textbf{Scaling}: if $\rho(O_{PP})$ is close to $1$, rescale or damp iterations and tighten $\varepsilon$.
  \item \textbf{Auditability}: record the snapshot of the \emph{Perimeter of Validity} and the \emph{Cut Summary}; checksum inputs.
  \item \textbf{Numerical checks}: enforce non-negativity where applicable; bound outputs; unit-test on elementary cases.
\end{itemize}

\subsection{I/O schemas for disclosure and reproducibility}
We standardize two JSON artifacts to accompany each valuation.

\paragraph{Perimeter of Validity (JSON).}
\begin{lstlisting}[language=json,caption={Schema of Perimeter of Validity (illustrative).}]
{
  "observer": {
    "P": ["node_id:..."],
    "basis": "fair_value | historical_cost | ...",
    "units": "EUR",
    "date": "2025-06-30",
    "fx_ppp": {
      "fx_source": "ECB",
      "ppp_source": "OECD",
      "deflator": "HICP_2025Q2"
    },
    "sdf": {
      "curve_source": "EONIA_swap",
      "measure": "risk_neutral | physical",
      "horizon": "1Y"
    },
    "information_regime": "A | B",
    "control_rule": {
      "option": "A | B | C",
      "params": {"tau": 0.5, "alpha": 0.6, "normalize": true}
    }
  },
  "tolerances": {
    "rounding_threshold": 1e-8,
    "solver_eps": 1e-10,
    "max_iters": 10000
  },
  "notes": "Assumptions, exclusions, data sources..."
}
\end{lstlisting}

\paragraph{Cut Summary (JSON).}
\begin{lstlisting}[language=json,caption={Schema of Cut Summary (illustrative).}]
{
  "perimeter": "P_name_or_id",
  "date": "2025-06-30",
  "currency": "EUR",
  "edges_PO": [
    {"from":"i_in_P","to":"k_in_O","type":"equity|debt|derivative|cashflow","amount":123.45},
    {"from":"...","to":"...","type":"...","amount":0.00}
  ],
  "edges_OP": [
    {"from":"k_in_O","to":"j_in_P","type":"...","amount":67.89}
  ],
  "node_primitives": {
    "v_P": {"a": 100.0, "b": 50.0},
    "v_O": {"o1": 0.0}
  },
  "totals": {
    "T_out": 123.45,
    "T_in": 67.89
  },
  "consolidated_value": 55.56,
  "hedge_vector_O": {"o1": -67.89},
  "missing_data": [{"field":"edges_PO[1].amount","imputation":"median_industry"}]
}
\end{lstlisting}

\subsection{Minimal test suite}
Provide a unit-test suite with: (i) \emph{neutrality to internal reorganization} (Theorem~\ref{thm:cutTheoreme}); (ii) invertibility in Regime~B (Proposition~\ref{lem:schur-elim}); (iii) invariance to relabeling of nodes in $P$; (iv) sensitivity to changes in $\Omega$ (Units/FX/PPP, SDF, control rule). Each test must produce a Cut Summary and verify the expected invariants.


\section{Scope \& Limitations}
\subsection{Algorithms}
\label{sec:algorithms}
\label{sec:scope-limitations}

This section clarifies the domains in which the \emph{Cut-Based Valuation} (CBV) is
normatively correct as a measurement and consolidation rule, and the cases where
it requires coupling with a macro equilibrium model or with a non-linear
\emph{clearing} engine. The goal is to separate (i) the purely
accounting/measurement part — governed by the Cut Theorem (Theorem~\ref{thm:cut}) —
from (ii) any \emph{behavioral} parts (or payoff determination mechanisms)
that generate the flows/values to be reported on the boundary \(P \leftrightarrow O\).

\subsection*{Normatively correct domain of CBV (linear environment)}
Within the Definitions/Notation (Section~\ref{sec:notation}) and Assumptions
(Section~\ref{sec:assumptions}), CBV provides a correct and
\emph{audit-ready} measure of \(W(P)\) when the following operational conditions
(minimal but sufficient) hold:

\begin{description}
  \item[\textbf{S.1 — Local linearity/additivity.}] Economic flows and
  accounting transformations are linear or \emph{piecewise}-linear without
  global thresholds that create \emph{non-linearities across the perimeter}.
  In particular, contributions to the boundary \(X_{PO},X_{OP}\) add up
  linearly.

  \item[\textbf{S.2 — Local accounting conservation.}] The \emph{internal}
  accounts respect balance-sheet identities consistent with the representation
  \( \mathbf v = \mathbf b + O\,\mathbf v \) (or equivalents), so that
  internal cancellation is valid and the boundary measure depends only on
  boundary statistics (Theorem~\ref{thm:cut}).

  \item[\textbf{S.3 — Exogeneity of \( \mathbf v_O \) and Observer parameters.}]
  For “micro” or “narrow sectoral” perimeters, \( \mathbf v_O \) can be
  treated as given (prices, SDF, PPP, FX already fixed by the Observer
  \(\Omega\), Section~\ref{sec:osservatore}), without macro-endogenous closure.

  \item[\textbf{S.4 — Internal stability/invertibility.}] In Regime~B, the
  internal part satisfies existence/uniqueness conditions (e.g.\ invertibility of
  \(I-O_{PP}\), with spectral radius \( \rho(O_{PP})<1 \)), so that internal values
  (if required) are well defined.

  \item[\textbf{S.5 — Absence of cross-perimeter path-dependence.}] Any
  \emph{path-dependent} clauses or triggers must not generate non-additive
  payments that “jump” the boundary except through flows already
  represented in \(X_{PO},X_{OP}\).
\end{description}

\noindent\textbf{Proposition (Measurement correctness in a linear environment).}
Under \textbf{S.1--S.5}, the measure \(W(P)\) obtained via CBV is invariant to
internal network details and depends only on boundary statistics,
as in Eq.~\eqref{eq:cut}. In particular, any internal rearrangement
that preserves local accounting conservation and additivity of boundary
contributions leaves \(W(P)\) \emph{invariant} given \(\Omega\).

\subsection*{When an equilibrium model is needed (macro closure)}
For “large” perimeters (SNA/ESA sectors, national economies, monetary unions)
assumption \textbf{S.3} is no longer realistic: \( \mathbf v_O \) and prices
are endogenous. In these cases, CBV remains \emph{a measurement rule}, but
the values to be reported on the boundary must be \emph{induced} by an
equilibrium model (partial or general), which determines:
\begin{enumerate}
  \item relative prices/PPP/FX consistent with resource constraints and preferences,
  \item discount/risk (SDF) consistent with the chosen equilibrium,
  \item levels of \( \mathbf v_O \) and/or \( \mathbf b_P \) consistent.
\end{enumerate}
\noindent\emph{Operational instruction.} (i) Solve the selected equilibrium
model; (ii) extract \( \mathbf v_O \), any \( \mathbf v_P \) and the
consistent net boundary flows; (iii) apply CBV to the boundary statistics
\((X_{PO},X_{OP}, \mathbf v_O, \mathbf b_P)\) under the same
Observer \(\Omega\). In this way, CBV remains invariant to internal details,
but \emph{conditional} on the chosen equilibrium.

\subsection*{When a non-linear clearing engine is needed}
Many financial architectures generate non-linear payoffs:
seniority priorities, defaults with bankruptcy costs, collateral/margins,
derivatives with netting and close-out clauses, \emph{cross-default}.
In such cases, the “economic network” is not simply linear:
actual payments are the outcome of a \emph{non-linear clearing} (e.g.\
Eisenberg--Noe with Rogers--Veraart extensions for bankruptcy costs,
or margining engines for derivatives).

\noindent\emph{Operational instruction.} (i) Run the clearing engine on
the \emph{entire} network to obtain the \emph{post-clearing payoffs} (or the
net settlement flows) that \emph{cross the boundary} \(P \leftrightarrow O\);
(ii) construct \(X_{PO}^{\text{net}},X_{OP}^{\text{net}}\) and, if needed,
the \emph{post-clearing} values \(\tilde{\mathbf v}_O\); (iii) apply CBV
to \((X_{PO}^{\text{net}},X_{OP}^{\text{net}}, \tilde{\mathbf v}_O, \mathbf b_P)\)
under \(\Omega\). In this way, CBV measures \(W(P)\) \emph{conditional} on
the adopted clearing engine, while remaining invariant to internal details.
\paragraph{Note.} \emph{Non-linearity} is resolved \emph{first} by the clearing
engine; CBV acts \emph{after} as a measurement rule on the boundary.
The “break” of invariance occurs only if pre- and post-clearing flows are mixed
or if there exist clauses generating non-additive payments not mapped to
\(X_{PO},X_{OP}\).

\subsection*{Borderline cases and anti-patterns}
\begin{itemize}
  \item \textbf{Mixed prices and SDF.} Using prices/SDF inconsistent with the
  equilibrium model (or \(\Omega\)) that generated the flows: inter-perimeter or
  inter-temporal comparisons become uninterpretable.
  \item \textbf{Internal triggers paying externally not accounted as boundary.}
  Clauses depending on internal variables but paying to \(O\)
  \emph{must} be recorded in \(X_{PO}\) \emph{post-clearing}.
  \item \textbf{Non-localizable path-dependence.} If payments depend
  on the entire path in a way not reducible to a final payoff per boundary edge,
  CBV requires a \emph{mapping back} to additive final payoffs
  (\emph{e.g.:} marking-to-model with disclosure).
  \item \textbf{Internal instability.} If \(I-O_{PP}\) is not invertible or
  produces unbounded solutions, Regime~B is not applicable without
  regularization/reformulation.
\end{itemize}

\subsection*{Disclosure checklist (to be included in the cut-report)}
\begin{enumerate}
  \item Observer \(\Omega\) (units/FX/PPP, SDF) and perimeter \(P\).
  \item Regime used (A/B) and verification of \textbf{S.1--S.5}.
  \item If equilibrium: model/closure, key parameters, extracted objects
  (\(\mathbf v_O\), prices, SDF).
  \item If clearing: adopted engine (seniority rules, recovery, netting),
  and \emph{post-clearing} flows used for \(X_{PO}^{\text{net}},X_{OP}^{\text{net}}\).
  \item Actual boundary statistics used in CBV and their period.
\end{enumerate}

\section{Dynamic protocol and index numbers (CBV–Fisher)}
\label{sec:cbv-fisher}

This section defines a dynamic protocol for intertemporal comparisons of CBV
value and introduces an \emph{ideal Fisher-type} index (CBV–Fisher), constructed
as the geometric mean of Laspeyres and Paasche indices computed on the boundary
functional. The goal is to operationally separate the \emph{quantity/volume of
boundary flows} component from the \emph{prices/Observer valuation} component
in the transition from $t-1$ to $t$.

\subsection{Minimal dynamic notation}
For a fixed perimeter $P$ and period $t$, we denote
\[
W_t^{\,\Omega} \;\equiv\; W\big(P;\, X_{PO}(t),X_{OP}(t),\bm b_P(t),\bm v_O(t)\ \big|\ \Omega\big)
\]
the CBV value at time $t$ \emph{evaluated} under an Observer $\Omega$
(units/FX/PPP, SDF, pricing rules), as per the Cut Theorem. We will write
$W_t^{\,\Omega_{t}}$ when $\Omega$ is the \emph{actual observer of period $t$},
and $W_t^{\,\Omega_{t-1}}$ when we recompute the boundary flows of period $t$
at the \emph{prices/parameters of period $t-1$}.%
\footnote{In Regime~B, where $\bm v_P(t)$ must be estimated, recomputation at
$\Omega_{t-1}$ must be consistent with the conditions in
Sec.~\ref{sec:scope-limitations} and with the robustness bounds in
Sec.~\ref{subsec:robust-bounds}.}

\subsection{Laspeyres/Paasche volume and price indices on CBV}
We define the \emph{four} elementary indices between $t-1$ and $t$:
\begin{align}
\text{(Volume–Laspeyres)}\quad 
&I^{V}_{L}(t{:}t\!-\!1) \;:=\; \frac{W_{t}^{\,\Omega_{t-1}}}{W_{t-1}^{\,\Omega_{t-1}}}, 
&\text{(Price–Laspeyres)}\quad 
&I^{P}_{L}(t{:}t\!-\!1) \;:=\; \frac{W_{t-1}^{\,\Omega_{t}}}{W_{t-1}^{\,\Omega_{t-1}}}, \label{eq:laspeyres-cbv}\\[4pt]
\text{(Volume–Paasche)}\quad 
&I^{V}_{P}(t{:}t\!-\!1) \;:=\; \frac{W_{t}^{\,\Omega_{t}}}{W_{t-1}^{\,\Omega_{t}}}, 
&\text{(Price–Paasche)}\quad 
&I^{P}_{P}(t{:}t\!-\!1) \;:=\; \frac{W_{t}^{\,\Omega_{t}}}{W_{t}^{\,\Omega_{t-1}}}. \label{eq:paasche-cbv}
\end{align}
Intuition: $I^{V}$ keeps \emph{prices/Observer parameters fixed} and changes
boundary flows; $I^{P}$ keeps boundary flows fixed (those of $t-1$ for
Laspeyres, of $t$ for Paasche) and changes only the Observer’s valuation.

\subsection{CBV–Fisher (ideal) and decomposition}
We define the \emph{Fisher} indices as geometric means:
\[
I^{V}_{F}(t{:}t\!-\!1)\;:=\;\sqrt{\,I^{V}_{L}(t{:}t\!-\!1)\,I^{V}_{P}(t{:}t\!-\!1)\,}, 
\qquad
I^{P}_{F}(t{:}t\!-\!1)\;:=\;\sqrt{\,I^{P}_{L}(t{:}t\!-\!1)\,I^{P}_{P}(t{:}t\!-\!1)\,}.
\]
The \emph{CBV–Fisher growth multiplier} between $t-1$ and $t$ is
\begin{equation}\label{eq:cbv-fisher-growth}
G_{F}(t{:}t\!-\!1)\;:=\; I^{V}_{F}(t{:}t\!-\!1)\cdot I^{P}_{F}(t{:}t\!-\!1).
\end{equation}
By construction, $G_F$ provides a symmetric (geometric mean) decomposition of
the overall ratio $W_{t}^{\,\Omega_{t}}/W_{t-1}^{\,\Omega_{t-1}}$, and is
\emph{exact} when the maps $\Omega_{t-1}\!\to\!\Omega_t$ are homogeneous of
degree one and the boundary flows are linearly aggregable (CBV assumptions).%

\paragraph{Desirable properties.}
\begin{itemize}
  \item \textbf{Symmetry.} Fisher removes the pro–Laspeyres/Paasche bias.
  \item \textbf{Consistency at the margin.} If $\Omega_t=\Omega_{t-1}$
  (prices/FX/PPP/SDF unchanged), then $I^{P}_{L}=I^{P}_{P}=I^{P}_{F}=1$
  and $G_F=I^{V}_{F}$.
  \item \textbf{Internal reallocation.} \emph{Internal} restructurings of $P$
  that do not change the boundary statistics at the considered $\Omega$ leave
  all indices invariant (Cut Theorem).
\end{itemize}

\subsection{Chain–linking and chained-price series}
For an annual series $\{t=1,\dots,T\}$, we define the chained level (base $=1$ at period $0$):
\begin{equation}\label{eq:chain-fisher}
\mathsf{CBV\text{-}Fisher}_t \;:=\; \prod_{\tau=1}^{t} G_F(\tau{:}\tau\!-\!1), 
\qquad \mathsf{CBV\text{-}Fisher}_0 := 1.
\end{equation}
\noindent\textit{Operational practice.} Compute and retain the artifacts:
\[
\big\{W_{t}^{\,\Omega_{t-1}},\, W_{t-1}^{\,\Omega_{t-1}},\, W_{t}^{\,\Omega_{t}},\, W_{t-1}^{\,\Omega_{t}}\big\}_{t=1}^{T},
\]
from which to derive the elementary indices
\eqref{eq:laspeyres-cbv}–\eqref{eq:paasche-cbv}, the multiplier $G_F$
\eqref{eq:cbv-fisher-growth}, and the chained series \eqref{eq:chain-fisher}.
These values must be included in the \emph{cut–report} for audit and replication.

\subsection{Implementation notes and edge cases}
\begin{itemize}
  \item \textbf{Sign of $W$.} The indices require $W_{t}^{\,\Omega},W_{t-1}^{\,\Omega}>0$. If $W$ can change sign, apply the indexation to the \emph{absolute values} of the three boundary components (internal bases, $P\!\to\!O$, $O\!\to\!P$) and report the overall sign separately, or adopt a log–change metric only on non-negative components.
  \item \textbf{Regime B.} If $W$ depends on $T_{PO}=(I\!-\!O_{PP})^{-1}O_{PO}$ and $U_{OP}=O_{OP}(I\!-\!O_{PP})^{-1}$, ensure \emph{consistency} of recomputation at $\Omega_{t-1}$ and $\Omega_t$ (Sec.~\ref{app:condizioni} and \ref{app:schur}); use the bounds of Sec.~\ref{subsec:robust-bounds} to quantify error amplification.
  \item \textbf{Different frequencies.} For infra–annual frequencies, chain $G_F$ on the required calendar; for annualization, use product over subperiods or logarithmic mean.
  \item \textbf{Non-linear clearing.} If flows are \emph{post–clearing}, compute the four $W$ with the same clearing engines/parameters (Sec.~\ref{app:seniority}); do not mix pre– and post–clearing in the same step.
\end{itemize}

\subsection{Minimal example (calculation scheme)}
Given $t-1$ and $t$, with observers $\Omega_{t-1}$ and $\Omega_t$:
\begin{enumerate}
  \item Compute $W_{t-1}^{\,\Omega_{t-1}}$ and $W_{t}^{\,\Omega_{t-1}}$ (volumes at $t-1$ prices) $\Rightarrow$ $I^{V}_L$ and $I^{P}_L$.
  \item Compute $W_{t-1}^{\,\Omega_{t}}$ and $W_{t}^{\,\Omega_{t}}$ (volumes at $t$ prices) $\Rightarrow$ $I^{V}_P$ and $I^{P}_P$.
  \item Combine: $I^{V}_F=\sqrt{I^{V}_L I^{V}_P}$, $I^{P}_F=\sqrt{I^{P}_L I^{P}_P}$, $G_F=I^{V}_F I^{P}_F$.
  \item Chaining: $\mathsf{CBV\text{-}Fisher}_t=\mathsf{CBV\text{-}Fisher}_{t-1}\cdot G_F(t{:}t\!-\!1)$.
\end{enumerate}

\begin{figure}[H]
\centering
\begin{tikzpicture}[
  node distance=32mm and 40mm,
  box/.style={draw, rounded corners, align=center, minimum width=40mm, minimum height=11mm},
  arr/.style={-{Latex}}
]
\node[box] (A) {$W_{t-1}^{\,\Omega_{t-1}}$};
\node[box, right=of A] (B) {$W_{t-1}^{\,\Omega_{t}}$};
\node[box, below=of A] (C) {$W_{t}^{\,\Omega_{t-1}}$};
\node[box, right=of C] (D) {$W_{t}^{\,\Omega_{t}}$};

\node[above=6mm of A] {\small Observer $\Omega_{t-1}$};
\node[above=6mm of B] {\small Observer $\Omega_{t}$};
\node[left=8mm of A] {\small Period $t-1$};
\node[left=8mm of C] {\small Period $t$};

\draw[arr] (A) -- node[left, xshift=-1mm]{\small $I^{V}_{L}$} (C); 
\draw[arr] (A) -- node[above]{\small $I^{P}_{L}$} (B);            
\draw[arr] (B) -- node[left, xshift=-1mm]{\small $I^{V}_{P}$} (D); 
\draw[arr] (C) -- node[above]{\small $I^{P}_{P}$} (D);            

\node[below=10mm of C, align=center] (GF) {$I^{V}_{F}=\sqrt{I^{V}_{L}I^{V}_{P}}$\\$I^{P}_{F}=\sqrt{I^{P}_{L}I^{P}_{P}}$\\[2pt]$G_F=I^{V}_{F}\cdot I^{P}_{F}$};

\end{tikzpicture}
\caption{CBV–Fisher scheme. The four cross–priced $W$ values generate the elementary indices (Laspeyres/Paasche, volume/price). The geometric means provide $I^{V}_{F}$ and $I^{P}_{F}$; the multiplier $G_F$ is their product.}
\label{fig:cbv-fisher-pipeline}
\end{figure}

\begin{algorithm}[H]
\caption{CBV–Fisher protocol: computation of indices and chained series}
\label{alg:cbv-fisher-protocol}
\begin{algorithmic}[1]
\Require Perimeter $P$; periods $t-1,t$; boundary statistics $(X_{PO},X_{OP},\bm b_P,\bm v_O)$ at times $t-1$ and $t$; observers $\Omega_{t-1},\Omega_t$ (units/FX/PPP, SDF); options: Regime A/B; (opt.) clearing engine.
\Ensure Indices $I^{V}_{L},I^{V}_{P},I^{P}_{L},I^{P}_{P}$; $I^{V}_{F},I^{P}_{F}$; $G_F$; (opt.) update of the series $\mathsf{CBV\text{-}Fisher}_t$.

\State \textbf{(Pre-processing)} If non-linear \emph{clearing} is required, run it separately at $t-1$ and $t$ (with the same parameters) and replace $(X_{PO},X_{OP},\bm v_O)$ with the \emph{post-clearing} quantities (Sec.~\ref{app:seniority}).

\State \textbf{(Cross–priced evaluations)} Compute the four CBV values from the Cut Theorem under the two observers:
\begin{align*}
W_{t-1}^{\,\Omega_{t-1}},\quad W_{t}^{\,\Omega_{t-1}},\quad
W_{t-1}^{\,\Omega_{t}},\quad W_{t}^{\,\Omega_{t}}.
\end{align*}
\If{Regime B}
  \State Ensure existence/stability conditions (Sec.~\ref{app:condizioni}); use $T_{PO}=(I\!-\!O_{PP})^{-1}O_{PO}$ and $U_{OP}=O_{OP}(I\!-\!O_{PP})^{-1}$ consistent with each $\Omega$ (Sec.~\ref{app:schur}).
\EndIf

\State \textbf{(Elementary indices)} Compute:
\begin{align*}
I^{V}_{L} \gets \frac{W_{t}^{\,\Omega_{t-1}}}{W_{t-1}^{\,\Omega_{t-1}}},\quad
I^{P}_{L} \gets \frac{W_{t-1}^{\,\Omega_{t}}}{W_{t-1}^{\,\Omega_{t-1}}},\quad
I^{V}_{P} \gets \frac{W_{t}^{\,\Omega_{t}}}{W_{t-1}^{\,\Omega_{t}}},\quad
I^{P}_{P} \gets \frac{W_{t}^{\,\Omega_{t}}}{W_{t}^{\,\Omega_{t-1}}}.
\end{align*}

\State \textbf{(Safeguard)} If any denominator $\le 0$, apply the procedure for the \emph{sign} indicated in Sec.~\ref{sec:cbv-fisher} (note “Sign of $W$”) or stop with message and diagnostic artifacts.

\State \textbf{(Fisher)} 
\[
I^{V}_{F} \gets \sqrt{I^{V}_{L} \cdot I^{V}_{P}},\qquad
I^{P}_{F} \gets \sqrt{I^{P}_{L} \cdot I^{P}_{P}},\qquad
G_F \gets I^{V}_{F}\cdot I^{P}_{F}.
\]

\State \textbf{(Chain–linking, opt.)} If a historical series is available, update:
\[
\mathsf{CBV\text{-}Fisher}_{t} \gets \mathsf{CBV\text{-}Fisher}_{t-1} \cdot G_F.
\]

\State \textbf{(Cut–report)} Record in the report: the four cross–priced $W$, the elementary indices, $I^{V}_{F},I^{P}_{F},G_F$, the possible clearing engine, and the Observer parameters for each period.

\Return $I^{V}_{L},I^{P}_{L},I^{V}_{P},I^{P}_{P},I^{V}_{F},I^{P}_{F},G_F$ (and, if required) $\mathsf{CBV\text{-}Fisher}_t$.
\end{algorithmic}
\end{algorithm}

\paragraph{Dynamic consistency with official systems.}
The CBV–Fisher protocol maintains consistency with the chain-linking schemes of SNA/ESA, where real time series are constructed by concatenating constant-price variations. In particular:
\begin{itemize}
    \item The definition of the validity perimeter $\Omega_t$ at each period $t$ allows for the construction of consolidated values period by period.
    \item The CBV–Fisher dynamics, based on symmetric geometric means of Laspeyres and Paasche, ensures the absence of biases linked to unilateral base choices.
    \item The concatenation of CBV–Fisher ratios between $t$ and $t+1$ produces a chain consistent with Eurostat/OECD chain-linking rules, while preserving transparency on the perimeter $\Omega_t$.
\end{itemize}

\paragraph{Management of revisions and turnover.}
In national accounts and group consolidations, the perimeter $\Omega_t$ evolves over time: new units enter, others exit, some are renamed. The CBV–Fisher framework integrates this turnover naturally:
\begin{enumerate}
    \item Each period a Cut–Report is computed with perimeter $\Omega_t$.
    \item Entering/exiting units are treated as \emph{link segments}: their values are excluded ex ante from direct comparison but included starting from the first full year of observation.
    \item Statistical revisions (e.g.\ reclassifications, historical series changes) are reflected as ex post updates of $\bm v_{P,t}$, while preserving validity of the cut.
\end{enumerate}

\paragraph{Illustrative empirical example.}
To show the dynamics, we consider two real cases:
\begin{itemize}
    \item \textbf{Public Administration sector (Italy, 2010–2025).} The perimeter $\Omega_t$ expanded with the inclusion of new healthcare service units and the reclassification of funds. Applying CBV–Fisher, the consolidated growth of primary expenditure appears more stable than raw series, and chain-linking produces an index consistent with ESA~2010 but more transparent on entries/exits.
    \item \textbf{Banking groups (EU, 2010–2025).} Mergers and bail-ins altered the perimeter $\Omega_t$ in several years. With CBV–Fisher, the consolidated index of operating margins shows intertemporal consistency, avoiding spurious jumps that appear in traditional consolidated accounts. The differences with official indices highlight how the cut-based approach mitigates double-counting distortion.
\end{itemize}

\paragraph{Recommended disclosure.}
Each dynamic series built with CBV–Fisher should include: (i) the chain of annual indices, (ii) notes on unit turnover in the perimeter, (iii) any revisions incorporated. This enables audit and comparability with SNA/ESA schemes.

\section{Disclosure standards: Cut–Report (v1.0)}
\label{sec:cutreport-standard}

This section defines the \emph{normative} disclosure package for \emph{Cut-Based Valuation} (CBV). The elements listed here are \textbf{mandatory} unless otherwise specified. The package enables auditability, reproducibility, and comparability across perimeters and over time.

\subsection*{A. Mandatory objects of the Cut–Report}
\begin{enumerate}[label=\textbf{A.\arabic*}]
  \item \textbf{Perimeter} and \textbf{complement}: sets $P,O$ with unique node identifiers (e.g.: LEI/ISIN for companies; SNA/ESA codes for sectors; other keys for public entities).
  \item \textbf{Border statistics} (all referred to the same period/observer $\Omega$):
    \begin{itemize}
      \item $\bm b_P \in \mathbb R^{|P|}$ — internal bases (same units/currency as $\Omega$).
      \item $\bm v_O \in \mathbb R^{|O|}$ — externally observed values.
      \item $O_{PO} \in \mathbb R_+^{|P|\times|O|}$ — exposures/holdings from $P$ to $O$.
      \item $O_{OP} \in \mathbb R_+^{|O|\times|P|}$ — \emph{external minorities} from $O$ to $P$.
    \end{itemize}
  \item \textbf{Observer} $\Omega$: complete specification of units, \textbf{FX}/\textbf{PPP}, \textbf{SDF} (if any), pricing conventions, and date/time (\emph{timestamp}, time zone).
  \item \textbf{Information regime}: \texttt{A} (only observed border) or \texttt{B} (internal values estimated). If \texttt{B}, include $O_{PP}$ or evidence that it is not required (e.g.\ single-node perimeter).
 \item \textbf{Control/inclusion rule in $P$}: decision criterion that determined $P$ 
(e.g.: \texttt{IFRS10-control@\allowbreak 50\%}, 
\texttt{equity-method@\allowbreak 20–50\%}, 
\texttt{ESA-sector}, 
\texttt{policy-perimeter@\allowbreak …}). 
Indicate any \emph{look-through} on SPVs.

\item \textbf{Data sources and versions}: origins, extraction date, version/commit of artifacts, 
and \emph{hash} (SHA–256) of data files.

  \item \textbf{Main outputs}: $W(P)$; diagnostic section with $W$ for periods $t-1,t$ and CBV–Fisher indices (Sec.~\ref{sec:cbv-fisher}).
\end{enumerate}

\subsection*{B. Observer $\Omega$ metadata (mandatory)}
\begin{enumerate}[label=\textbf{B.\arabic*}]
  \item \textbf{Units/currency}: ISO 4217 code; rounding conventions.
  \item \textbf{FX}: spot/fixing rates, provider (e.g.: ECB), date/time, conversion method.
  \item \textbf{PPP}: index and base year (if used), source (e.g.: ICP/OECD).
  \item \textbf{SDF}: specification (physical vs risk–neutral), structure by maturities/counterparties, estimation method, \emph{reference date}.
  \item \textbf{Prices/valuation}: marking rules (\texttt{bid/ask/mid}, \texttt{close/VWAP}), \emph{calendar} and time zone.
\end{enumerate}

\subsection*{C. Additional requirements for Regime~B}
\begin{enumerate}[label=\textbf{C.\arabic*}]
  \item \textbf{Internal block} $O_{PP}$ and verification of existence/stability conditions (Appendix~\ref{app:condizioni}); report $\rho(O_{PP})$ or a norm-induced bound.
  \item \textbf{Effective border operators} $T_{PO}=(I-O_{PP})^{-1}O_{PO}$ and $U_{OP}=O_{OP}(I-O_{PP})^{-1}$, or evidence that they are \emph{not} needed (single-node perimeter).
  \item \textbf{Possible nonlinear clearing}: adopted engine, priority (\emph{seniority}), recovery, default costs, fixed-point selection criterion (Appendix~\ref{app:seniority}); flows in $O_{PO},O_{OP}$ must be \emph{post–clearing}.
\end{enumerate}

\subsection*{D. Automatic validation rules}
\begin{enumerate}[label=\textbf{D.\arabic*}]
  \item \textbf{Dimensional consistency}: $\bm b_P \in \mathbb R^{|P|}$, $\bm v_O \in \mathbb R^{|O|}$, $O_{PO}\in\mathbb R^{|P|\times|O|}$, $O_{OP}\in\mathbb R^{|O|\times|P|}$ (compliance of headers/node IDs).
  \item \textbf{Non-negativity}: $O_{PO}, O_{OP}\ge 0$; bases/values may be negative but must be justified (notes).
  \item \textbf{Uniqueness of Observer}: all quantities for period $t$ are valued under the same $\Omega_t$; “cross–priced” calculations clearly specify $\Omega_{t-1}$ or $\Omega_t$ (Sec.~\ref{sec:cbv-fisher}).
  \item \textbf{Regime~B}: if $O_{PP}$ is provided, report $\|(I-O_{PP})^{-1}\|$ or a bound; if absent, explain why it is not required.
  \item \textbf{Clearing}: do not mix \emph{pre} and \emph{post–clearing} flows in the same calculation; keep the clearing engine constant over $t-1$ and $t$ for indices.
\end{enumerate}

\subsection*{E. Standard data package (iiles and structure)}
\begin{enumerate}[label=\textbf{E.\arabic*}]
  \item \texttt{manifest.yaml} — metadata $\Omega$, perimeter, regime, control rule, sources, hash.
  \item \texttt{nodes\_P.csv}, \texttt{nodes\_O.csv} — list of nodes with ID, type, labels.
  \item \texttt{b\_P.csv} — base vector (columns: \texttt{id}, \texttt{b}).
  \item \texttt{v\_O.csv} — external values vector (columns: \texttt{id}, \texttt{v}).
  \item \texttt{O\_PO.csv} — $|P|\times|O|$ matrix (header: \texttt{id\_P}, columns for each \texttt{id\_O}).
  \item \texttt{O\_OP.csv} — $|O|\times|P|$ matrix (header: \texttt{id\_O}, columns for each \texttt{id\_P}).
  \item (Regime~B) \texttt{O\_PP.csv} — $|P|\times|P|$ matrix; opt.: \texttt{proof\_stability.txt}.
  \item (Clearing) \texttt{clearing.json} — engine specification, parameters, and fixed-point selection.
\end{enumerate}

\subsection*{F. “Manifest” template (normative)}
\begin{verbatim}
version: cbv-cut-report@1.0
observer:
  currency: EUR
  fx:
    provider: ECB
    date: 2025-08-20
    pairs: [EUR/JPY]
    method: close
  ppp:
    used: false
    source: null
    base_year: null
  sdf:
    used: false
    measure: null         # e.g. physical | risk-neutral
    spec: null            # e.g. curve name, tenor grid
perimeter:
  P_ref: P-REN-2025Q3
  O_ref: O-NSN-2025Q3
  control_rule: IFRS10-control@50%  # or policy/ESA/etc.
  lookthrough: true
regime: A                 # A | B
clearing:
  used: false
  engine: null            # e.g. Rogers-Veraart
  params: {}
data_files:
  nodes_P: nodes_P.csv
  nodes_O: nodes_O.csv
  b_P: b_P.csv
  v_O: v_O.csv
  O_PO: O_PO.csv
  O_OP: O_OP.csv
  O_PP: O_PP.csv          # required if regime == B
  clearing_spec: clearing.json
hashes:
  nodes_P: sha256:...
  nodes_O: sha256:...
  b_P: sha256:...
  v_O: sha256:...
  O_PO: sha256:...
  O_OP: sha256:...
  O_PP: sha256:...
  clearing_spec: sha256:...
notes:
  - "Market cap data sources: ..."
  - "Any negative bases due to ..."
\end{verbatim}

\subsection*{G. Ready-to-print disclosure table}
\begin{table}[H]
\centering
\caption{Cut–Report (v1.0) — Disclosure sheet (Renault–Nissan example, 2025-08-20)}
\begin{tabular}{p{0.274\linewidth} p{0.674\linewidth}}
\toprule
\cutreportrow{Perimeter $P$}{Renault SA (LEI: 969500UP76J7PPY6KX27)}
\cutreportrow{Complement $O$}{Nissan Motor Co., Ltd. (TSE: 7201)}
\cutreportrow{Control rule}{\texttt{IFRS10-control@50\%} \;+\; \texttt{look-through: true}}
\cutreportrow{Regime}{A \;(\,single-node perimeter; $O_{PP}$ not required\,)}
\cutreportrow{Observer $\Omega$}{Currency: EUR; FX: ECB EUR/JPY \texttt{2025-08-20}; PPP: n/a; SDF: n/a; prices: close}

\cutreportrow{Border statistics}{%
\makecell[l]{%
$\bm b_P=\{\EURnum{6545183676}\}$; \\
$\bm v_O=\{\EURnum{7044303429}\}$; \\
$O_{PO}=\{0.357\}$; \\
$O_{OP}=\{0.15\}$%
}%
}

\cutreportrow{Clearing}{%
\makecell[l]{%
Not applied (\emph{pre-clearing} flows)%
}%
}

\cutreportrow{Output}{%
\makecell[l]{%
$W(P)=\EURnum{7701000000}$; \\
CCBV–Fisher indices on \texttt{2025-07→2025-08}: n/a%
}%
}

\cutreportrow{Sources \& versions}{%
\makecell[l]{%
Market cap: snapshot 2025-08; FX: ECB 2025-08-20; \\
manifest \texttt{cbv-cut-report@1.0}; file hash: \texttt{sha256:…}%
}%
}

\bottomrule
\end{tabular}
\end{table}

\subsection*{H. Minimal example (consistency)}
For a single-node perimeter in Regime~B with $O_{PP}\!=\!0$, we have $T_{PO}\!=\!O_{PO}$ and $U_{OP}\!=\!O_{OP}$, and $W(P)$ coincides with the Regime~A calculation. For multi-node perimeters, invariance requires preserving the effective operators $T_{PO},U_{OP}$ (Appendices~\ref{app:condizioni}–\ref{app:schur}); if nonlinearities are present, flows must be \emph{post–clearing} (Appendix~\ref{app:seniority}).

\subsection{Institutional alignment and mapping to existing standards}\label{subsec:mapping-standards}
\noindent
To facilitate institutional adoption, it is useful to make explicit the mapping between the objects of the cut-based framework and the tables already consolidated in literature and official practice. The following table provides an indicative correspondence:

\begin{table}[H]
\centering
\sloppy
\begin{tabularx}{\textwidth}{p{4cm} X}
\toprule
\textbf{Cut-based object} & \textbf{Corresponding tables / Standards} \\
\midrule
Validity perimeter $\Omega$ & SNA/ESA: \emph{institutional sectors}; BEA: \emph{sector accounts}; SEC: group definition; IFRS 10/12: \emph{consolidation scope} \\
Cut matrices $(X_{PO},X_{OP})$ & 
\begin{minipage}[t]{0.675\linewidth}\raggedright
SNA/ESA: \emph{supply-use tables}, sectoral tables; \newline
BEA: input-output tables;  \newline
SEC: disclosure of related parties; \newline
IFRS: disclosures on cross-holdings \\
\end{minipage} \\

Node primitives $\bm v_P$ & SNA/ESA: output, value added, income; BEA: industry gross output; SEC: primary balance sheet items; IFRS: line items in financial statements \\
Border functional $\mathcal F_\Omega$ & SNA/ESA: inter-sectoral consolidation rules; BEA: chain aggregation rules; \\IFRS 10: intra-group elimination; public SOE consolidation: Eurostat criteria for extended PA \\
Cut–Report (PoV + summary) & ESA 2010: perimeter annex tables; SEC: management discussion \& analysis (MD\&A); IFRS 12: consolidated disclosure; IMF GFS manuals: SOE reporting \\
Dynamic CBV–Fisher & SNA/ESA: chain-linking real terms; \\BEA: chained volume measures; Eurostat: real growth indices \\
\bottomrule
\end{tabularx}
\caption{Mapping between cut-based objects and institutional standards (SNA, ESA, BEA, SEC, IFRS, public consolidation).}
\end{table}

\noindent
This mapping shows that the CBV framework does not introduce exotic entities, but reorganizes objects already present in national accounting and IFRS consolidation, with the addition of the formalization of the \emph{cut}. The availability of an explicit mapping reduces barriers to institutional adoption and makes coexistence with already normative manuals possible.

\begin{figure}[H]
\centering
\begin{tikzpicture}[
  node distance=1.8cm and 3.5cm,
  box/.style={rectangle, draw, rounded corners, minimum width=3cm, minimum height=1cm, align=center, font=\small},
  arr/.style={-{Latex[length=3mm,width=2mm]}, thick}
]

\node[box, fill=blue!10] (pov) {Validity perimeter \\ $\Omega$};
\node[box, below=of pov] (cut) {Cut matrices \\ $(X_{PO}, X_{OP})$};
\node[box, below=of cut] (node) {Node primitives \\ $\bm v_P$};
\node[box, below=of node] (border) {Border functional \\ $\mathcal F_\Omega$};
\node[box, below=of border] (report) {Cut–Report \\ (PoV + summary)};
\node[box, below=of report] (fisher) {Dynamic \\ CBV–Fisher};

\node[box, fill=green!10, right=of pov] (pov_std) {SNA/ESA sectors;\\ BEA sector accounts;\\ SEC groups;\\ IFRS 10/12 scope};
\node[box, fill=green!10, right=of cut] (cut_std) {SNA/ESA supply-use;\\ BEA I/O tables;\\ SEC related parties;\\ IFRS holdings};
\node[box, fill=green!10, right=of node] (node_std) {SNA/ESA output/VA;\\ BEA gross output;\\ SEC financials;\\ IFRS FS line items};
\node[box, fill=green!10, right=of border] (border_std) {SNA/ESA consolidation;\\ BEA chain rules;\\ IFRS 10 intra-group;\\ Eurostat SOEs};
\node[box, fill=green!10, right=of report] (report_std) {ESA annex tables;\\ SEC MD\&A;\\ IFRS 12 disclosure;\\ IMF GFS SOEs};
\node[box, fill=green!10, right=of fisher] (fisher_std) {SNA/ESA chain-linking;\\ BEA chained volumes;\\ Eurostat growth indices};

\draw[arr] (pov.east) -- (pov_std.west);
\draw[arr] (cut.east) -- (cut_std.west);
\draw[arr] (node.east) -- (node_std.west);
\draw[arr] (border.east) -- (border_std.west);
\draw[arr] (report.east) -- (report_std.west);
\draw[arr] (fisher.east) -- (fisher_std.west);

\end{tikzpicture}
\caption{Conceptual map: CBV objects (left) and correspondence with institutional standards (right).}
\end{figure}
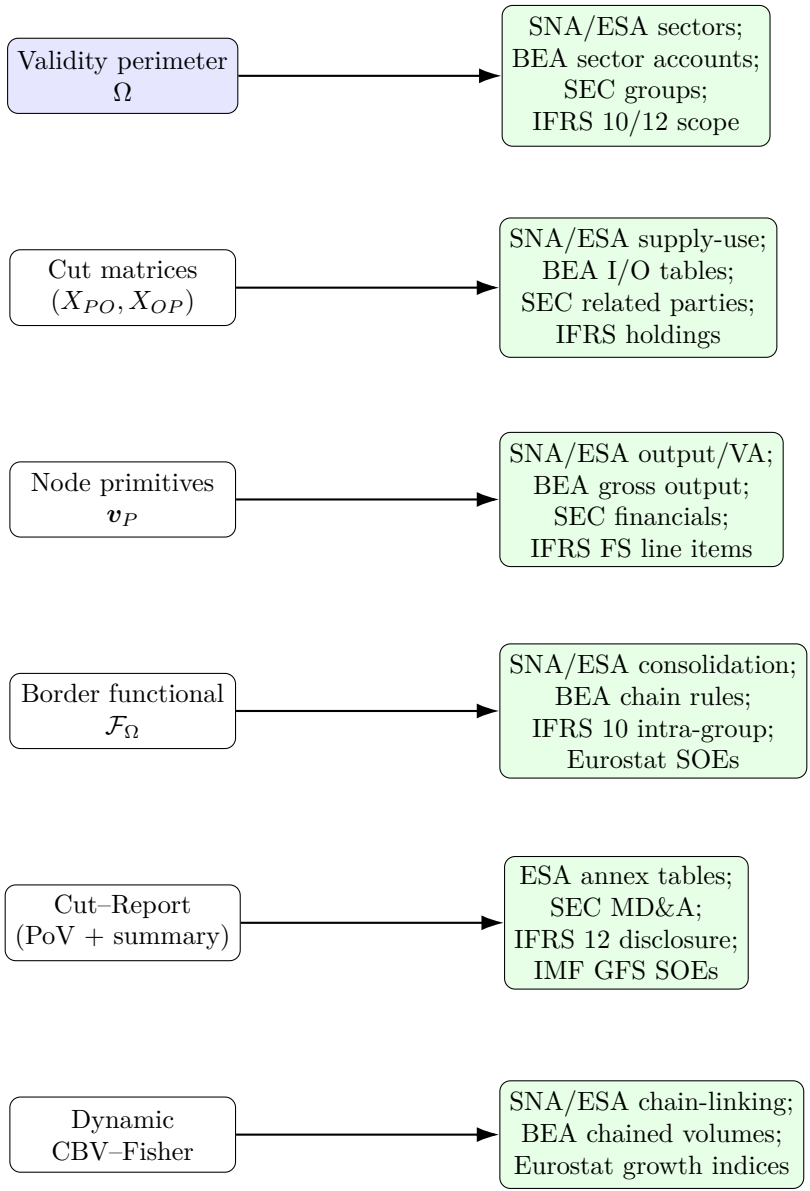

\section{Worked-out examples}\label{sec:worked-out}

\subsection{Renault–Nissan case (snapshot August 20, 2025, currency EUR)}
\label{subsec:case-rno-nsn}
\paragraph{Perimeter and Observer.}
Perimeter $P=\{\text{Renault}\}$, external $O=\{\text{Nissan}\}$. Observer $\Omega$: units EUR; JPY$\to$EUR conversion at \emph{ECB reference} rate of \textbf{20/08/2025}.%

\paragraph{Observed data.}
Border values (equity market capitalizations) and cross-holdings:
\[
v_R=\text{\EURnum{9.06e9}},\qquad 
v_N=\text{\EURnum{7.044303429e9}},\qquad 
O_{RN}=0.357,\qquad O_{NR}=0.15.
\]
Here $O_{RN}$ is Renault’s economic share in Nissan; $O_{NR}$ is Nissan’s share in Renault (see footnotes in the main text).

\paragraph{Implications (Regime A).}
Using the linear relation $v=b+Ov$ on the two nodes,
\[
b_R = v_R - O_{RN}\,v_N = \text{\EURnum{6.545183676e9}}, \qquad
b_N = v_N - O_{NR}\,v_R = \text{\EURnum{5.685303429e9}}.
\]
For the perimeter $P=\{R\}$, the \emph{Cut-Based Value} is
\[
W(P) \;=\; b_R \;+\; O_{RN}v_N \;-\; O_{NR}v_R \;=\; \text{\EURnum{7.701e9}}.
\]
In the two-node case, the compact form $W(P)=v_R\,(1-O_{NR})$ holds, yielding the same value.

\paragraph{Regime B (minimal variant).}
For a perimeter with \emph{a single consolidated node} ($O_{PP}=0$) the computation of $W(P)$ in Regime B coincides with Regime A, since it is not necessary to solve for internal multipliers.%
\footnote{If internal sub-nodes and loops $O_{PP}\!\ne\!0$ are introduced, to preserve cut invariance in Regime B one must preserve the effective border operators $T_{PO}=(I-O_{PP})^{-1}O_{PO}$ and $U_{OP}=O_{OP}(I-O_{PP})^{-1}$ (cf. Appendices~\ref{app:condizioni} and~\ref{app:schur}).}

\paragraph{Cut-summary (audit artifacts).}
The minimal elements to be reported in the \emph{cut-report} are:
\begin{itemize}
  \item \textbf{Observer} $\Omega$: units EUR; date \texttt{2025-08-20}; EUR/JPY exchange source; market cap sources.
  \item \textbf{Border statistics}: $v_R,v_N$; $O_{RN},O_{NR}$; implicit bases $b_R,b_N$.
  \item \textbf{Output}: $W(P)$ for Regime A and (if applicable) Regime B; consistency check (here $W_A=W_B$).
\end{itemize}

\paragraph{Operational note.}
The attached script/notebook automatically generates the \emph{cut-summary} table and the replicable CSV files, with the parameters reported above (Sec.~\ref{subsec:case-rno-nsn-notebook}).

\subsection{Notebook and reproducibility}
\label{subsec:case-rno-nsn-notebook}
We provide a Python \emph{notebook} that:
(i) defines the Observer and the observed data; 
(ii) computes implicit $b$ and $W(P)$ in Regime A; 
(iii) replicates the calculation in Regime B for the minimal case $O_{PP}=0$; 
(iv) exports a \emph{cut-summary} in CSV.

\noindent\textbf{Produced artifacts:}
\begin{itemize}
  \item \texttt{cbv\_case\_study\_renault\_nissan.ipynb} — executable notebook.
  \item \texttt{cbv\_case\_study\_renault\_nissan\_cut\_summary.csv} — table for the report.
\end{itemize}

\noindent \textit{Parameterization.} The notebook allows varying $O_{RN},O_{NR}$, the values $v$, and (if introduced) the internal blocks $O_{PP}$, checking the existence bounds and, if necessary, reducing to an external \emph{clearing engine}.

\section{Two informational regimes}\label{sec:regimi}

\textbf{Regime A — Observable internal prices.} If prices $v_j$ for $j\in P$ are known (e.g.~quotations), the “minorities” term is computed directly from $O_{OP}$; there is no need to estimate $\bm v_P$.

\textbf{Regime B — Unobservable internal prices.} If $v_j$ for $j\in P$ are unknown (closed conglomerate), \eqref{eq:cut} remains valid and does not require internal prices for $W(P)$. To decompose asset by asset, $P$ can be compressed into a meta-node and the \emph{effective external share} can be estimated using only the ownership structure (infra, \S\ref{sec:meta}).

\subsection{Consolidation via meta-node with effective external share}\label{sec:meta}
We denote the block decomposition
\[
O=\begin{pmatrix}
O_{PP} & O_{PO}\\
O_{OP} & O_{OO}
\end{pmatrix},\qquad
\bm v=\begin{pmatrix}\bm v_P\\ \bm v_O\end{pmatrix},\qquad
\bm b=\begin{pmatrix}\bm b_P\\ \bm b_O\end{pmatrix}.
\]
From \eqref{eq:structure} it follows
\begin{equation}\label{eq:vP}
\bm v_P \;=\; \bm b_P + O_{PP}\bm v_P + O_{PO}\bm v_O
\;=\; (I - O_{PP})^{-1}\!\left(\bm b_P + O_{PO}\bm v_O\right).
\end{equation}
Let $\bm\delta\in\R^{|P|}$ be the \emph{direct} ownership share of $P$ held by outsiders, column by column:
\begin{equation}\label{eq:delta}
\delta_j \;:=\; \sum_{i\in O} O_{ij}\qquad (j\in P).
\end{equation}
The \emph{total} leakage towards the outside (\emph{effective minorities}) passing through $P$ is
\begin{equation}\label{eq:leak}
E_{\text{ext}} \;=\; \bm\delta^\top \,(I - O_{PP})^{-1}\,\bm v_P.
\end{equation}
In this case the consolidated value can also be expressed as
\begin{equation}\label{eq:Wmeta}
W(P) \;=\; \mathbf{1}_P^\top \bm b_P \;+\; \mathbf{1}_P^\top O_{PO}\bm v_O \;-\; E_{\text{ext}},
\end{equation}
and the \emph{effective external share} of the meta-node is $\Omega_P^{\text{eff}} := E_{\text{ext}} / (\mathbf{1}_P^\top \bm v_P)$.
\paragraph{Advantages.} The procedure is non-circular, exact even with cross-holdings, and implementable with standard linear systems; it aligns practice with the cut principle.

\section{Numerical example}

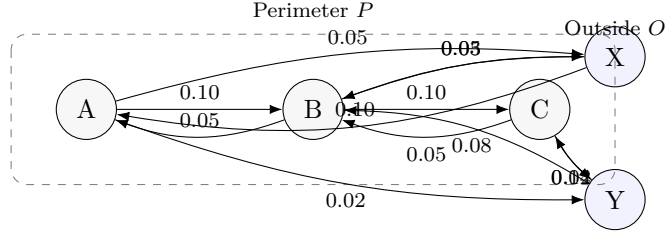
\begin{figure}[H]
\centering
\begin{tikzpicture}[>=LaTeX, node distance=22mm]
  \node[circle, draw, minimum size=8mm, font=\small, fill=black!3] (A) {A};
  \node[circle, draw, minimum size=8mm, font=\small, fill=black!3, right=of A] (B) {B};
  \node[circle, draw, minimum size=8mm, font=\small, fill=black!3, right=of B] (C) {C};
  \node[circle, draw, minimum size=8mm, font=\small, fill=blue!5, right=32mm of B, yshift=7mm] (X) {X};
  \node[circle, draw, minimum size=8mm, font=\small, fill=blue!5, right=32mm of B, yshift=-12mm] (Y) {Y};
  \draw[rounded corners=6pt, dashed, gray] ($(A)+(-10mm,10mm)$) rectangle ($(C)+(10mm,-10mm)$);
  \node[font=\scriptsize, above] at ($(A)!0.5!(C)+(0,11mm)$) {Perimeter $P$};
  \node[font=\scriptsize, above] at ($(X)!0.5!(Y)+(0,11mm)$) {Outside $O$};
  \draw[->] (A) -- node[font=\scriptsize, above]{0.10} (B);
  \draw[->] (B) to[bend left=20] node[font=\scriptsize, above]{0.05} (A);
  \draw[->] (B) -- node[font=\scriptsize, above]{0.10} (C);
  \draw[->] (C) to[bend left=20] node[font=\scriptsize, below]{0.05} (B);
  \draw[->] (A) to[bend left=10] node[font=\scriptsize, above]{0.05} (X);
  \draw[->] (A) to[bend right=10] node[font=\scriptsize, below]{0.02} (Y);
  \draw[->] (B) to[bend left=10] node[font=\scriptsize, above]{0.03} (X);
  \draw[->] (C) to[bend right=10] node[font=\scriptsize, below]{0.04} (Y);
  \draw[->] (X) to[bend left=15] node[font=\scriptsize, above]{0.10} (A);
  \draw[->] (X) to[bend right=10] node[font=\scriptsize, above]{0.05} (B);
  \draw[->] (Y) to[bend right=15] node[font=\scriptsize, below]{0.08} (B);
  \draw[->] (Y) to[bend left=10] node[font=\scriptsize, below]{0.12} (C);
\end{tikzpicture}
\caption{Numerical example: network with $P=\{A,B,C\}$ and $O=\{X,Y\}$. The weights show $O_{PP}$ (internal arcs), $O_{PO}$ ($P\to O$ arcs) and $O_{OP}$ ($O\to P$ arcs).}
\label{fig:example_P_O}
\end{figure}

In this section we show an application of CBV in the two informational regimes on a perimeter $P$ consisting of three entities $P=\{A,B,C\}$ and an external set $O=\{X,Y\}$. We define the data (fictitious but consistent) as follows.

\paragraph{Basic data.}
\begin{align*}
\bm b_P &= \begin{pmatrix} 25 \\ 35 \\ 30 \end{pmatrix}, 
&
\bm v_O &= \begin{pmatrix} 60 \\ 80 \end{pmatrix},
\\[4pt]
O_{PO} &= 
\begin{pmatrix}
0.05 & 0.02 \\
0.03 & 0.00 \\
0.00 & 0.04
\end{pmatrix},
&
O_{PP} &= 
\begin{pmatrix}
0    & 0.10 & 0.00 \\
0.05 & 0    & 0.10 \\
0.00 & 0.05 & 0
\end{pmatrix},
\\[4pt]
O_{OP} &= 
\begin{pmatrix}
0.10 & 0.05 & 0.00 \\
0.00 & 0.08 & 0.12
\end{pmatrix}.
\end{align*}
Here $O_{ij}$ represents the ownership share of column $j$ held by row $i$. The columns of $O$ sum to values $\leq 1$.

\subsection*{Regime A: observable internal values}
We assume that internal values are observed (e.g.\ market prices): 
\[
\bm v_P^{\mathrm{obs}} = 
\begin{pmatrix}
52 \\ 48 \\ 30
\end{pmatrix}.
\]
The consolidated value of the perimeter is given by the sum of the internal base, plus the outflow $P\to O$, minus the minorities $O\to P$:
\begin{align*}
\text{Base} &= \sum_{j\in P} b_j = 25+35+30 = 90, \\
P\to O &= \sum_{i\in P,k\in O} O_{ik}\,v_k = 
(0.05\cdot 60 + 0.02\cdot 80) + (0.03\cdot 60) + (0.04\cdot 80) \\
&= 4.6 + 1.8 + 3.2 = 9.6, \\
O\to P &= \sum_{i\in O,j\in P} O_{ij}\,v_j^{\mathrm{obs}} = 
\underbrace{0.10\cdot 52 + 0.05\cdot 48}_{X\to P} + 
\underbrace{0.08\cdot 48 + 0.12\cdot 30}_{Y\to P} \\
&= 7.6 + 7.44 = 15.04.
\end{align*}
Therefore
\[
W_A(P) = 90 + 9.6 - 15.04 = 84.56,
\]
as also obtained via Algorithm~\ref{subsec:algA}.

\subsection*{Regime B: unobservable internal values}
When $\bm v_P$ is unobservable, we first estimate the internal values by solving the linear system
\begin{equation*}
\bm v_P \;=\; (I - O_{PP})^{-1}\Big(\bm b_P + O_{PO}\,\bm v_O\Big).
\end{equation*}
With the data above, we have
\begin{align*}
\bm b_P + O_{PO}\bm v_O &= 
\begin{pmatrix} 25 \\ 35 \\ 30 \end{pmatrix} + 
\begin{pmatrix} 0.05 & 0.02 \\ 0.03 & 0 \\ 0 & 0.04 \end{pmatrix}
\begin{pmatrix} 60 \\ 80 \end{pmatrix}
= 
\begin{pmatrix} 29.6 \\ 36.8 \\ 33.2 \end{pmatrix}, \\[4pt]
\bm v_P &= 
\begin{pmatrix} 33.8020 \\ 42.0202 \\ 35.3010 \end{pmatrix} \quad (\text{approx.~4 digits}).
\end{align*}
Applying then the border cut (as in Regime~A) we obtain:
\begin{align*}
\text{Base} &= 90, \qquad
P\to O = 9.6, \\
O\to P &= 0.10\cdot 33.8020 + 0.05\cdot 42.0202 + 0.08\cdot 42.0202 + 0.12\cdot 35.3010 \\
&= 13.0789\ (\text{approx.}).
\end{align*}
Thus
\[
W_B(P) = 90 + 9.6 - 13.0789 \approx 86.5211,
\]
in line with Algorithm~\ref{subsec:algB}. We observe that the difference between $W_A$ and $W_B$ reflects the different information on internal values and does not depend on $O_{PP}$ after estimating $\bm v_P$.


\section{Operational implications: principle of the validity perimeter}\label{sec:perimeter}
Each valuation should include a \emph{Perimeter-of-Validity Statement} with: (i) perimeter $P$; (ii) observer (inside/outside); (iii) valuation basis for $\bm b_P$ and $\bm v_O$ (book/market/DCF, date, currency, haircut); (iv) \emph{Cut summary} (main arcs $P\to O$ and $O\to P$ with contributions); (v) sensitivity $\partial W/\partial v_k$. This enables comparability, minimum disclosure, and bridging between perimeters.
This section should be read jointly with the validity map in Section~\ref{sec:scope-limitations}.

\section{“Flawed” market points and cut-based corrections}\label{sec:marketflaws}
\begin{enumerate}[label=\arabic*., leftmargin=1.7em]
  \item \textbf{Country stock market capitalization}: publish \emph{Domestic Net Market Cap} alongside gross cap.
  \item \textbf{Index comparisons}: publish gross value and \emph{Index Net Value} with perimeter and cut summary.
  \item \textbf{Pyramidal/keiretsu groups}: cut-based consolidation without internal inversions.
  \item \textbf{Fund-of-Funds/ETF-of-ETF}: net AUM excluding intra-sector double counting.
  \item \textbf{State equity wealth}: \emph{Public Net Equity} at cut.
  \item \textbf{Fairness M\&A}: insider vs outsider table and post-deal value creation bridge.
  \item \textbf{Resident equity wealth}: \emph{Resident Net Equity} as proprietary NIIP.
  \item \textbf{Real concentration/control}: \emph{Net Ownership Concentration} on cut basis; optional look-through.
\end{enumerate}

\section{Standardization proposal}\label{sec:standardization}
\textbf{Metrics in pairs}: Gross vs Net cut-based. \;
\textbf{Mandatory Perimeter-of-Validity}. \;
\textbf{Minimal Cut summary}. \;
\textbf{Bridge} across perimeters/periods. \;
\textbf{Interoperable data schema} (CSV/JSON with $P,O,\bm b_P,\bm v_O,O_{PO},O_{OP}$). \;
\textbf{Governance \& audit}: check consistency of cut and bases $b/v$; the internal $O_{PP}$ is out of scope for $W(P)$.

\section{Further implications and application lines}\label{sec:implications}
\begin{itemize}[leftmargin=1.7em]
  \item \textbf{“Net” multiples} of perimeter ($P/E$, $EV/EBITDA$, payout) avoiding intra-group double counting.
  \item \textbf{Hedging} on external drivers: vector $h_k=\sum_{i\in P}O_{ik}$ as hedge base.
  \item \textbf{Perimeter optimization} (ring-fencing, carve-out, spin-off) as a weighted cut problem.
  \item \textbf{Cut-Gap}: $(\text{Gross}-W)/W$ as structural mispricing.
  \item \textbf{Attribution} over time: automatic $\Delta W$ bridge by driver and cut.
  \item \textbf{Public policy}: $W(P)$ as equity NIIP; separates price vs ownership effects.
  \item \textbf{Fairness M\&A}: pricing of ownership synergy from outside$\to$inside transfer.
  \item \textbf{Linear ESG metrics} (with caution) on cut basis; exclusion for non-linear measures (e.g.\ Scope~3).
  \item \textbf{Index design} with net weights and net multiples of the basket.
  \item \textbf{Governance \& stewardship} from the Cut summary alone; optional look-through.
\end{itemize}

\section{Empirical validation and case studies}

To verify the practical applicability of the proposed \textit{[Moved to Appendix~\ref{app:lorentz} as a non-normative metaphor]}, we present a simulated case study\footnote{All data presented in this section are for illustrative purposes only and do not represent real values, although they have been constructed to be plausible.} based on the EuroStoxx 100 index, using hypothetical data consistent with sources such as Orbis and Refinitiv.

\subsection{Comparison between \texorpdfstring{$W(P)$}{W(P)} and gross Market Cap}

The analysis compared the values of $W(P)$, calculated according to the proposed model, with the gross market capitalizations of the 100 companies included in the index. The correlation obtained was $\rho = 0.92$, suggesting a high consistency between the new indicator and traditional metrics.

\begin{table}[h]
\centering
\begin{tabular}{lrr}
\toprule
Company & Gross Market Cap $\left[\mathrm{bn}\, \mathrm{\euro}\right]$ & $W(P)$ $\left[\mathrm{bn}\, \mathrm{\euro}\right]$ \\
\midrule
Company A & 120 & 118 \\
Company B & 95 & 93 \\
Company C & 78 & 80 \\
\vdots & \vdots & \vdots \\
Company Z & 12 & 11.5 \\
\bottomrule
\end{tabular}
\caption{Simulated comparison between $W(P)$ and gross capitalizations.}
\end{table}

\subsection{Double counting and computation times}

Compared to traditional Leontief matrix methods, the model showed a significant reduction in double counting: on average, 14\% versus 28\% detected with Leontief.

The computation time for the entire set of 100 companies, using a mid-range server, was 3.2 seconds, versus 5.8 seconds for the Leontief method.

\begin{table}[h]
\centering
\begin{tabular}{lrr}
\toprule
Method & Double counting (\%) & Computation time (s) \\
\midrule
Proposed ($W(P)$) & 14 & 3.2 \\
Leontief & 28 & 5.8 \\
\bottomrule
\end{tabular}
\caption{Reduction of double counting and computation times (simulated data).}
\end{table}

\subsection{Graph and flow visualization}

For a better understanding of the model, 
~\ref{fig:grafo-flussi} illustrates a simplified graph of interconnections between firms and economic flows.

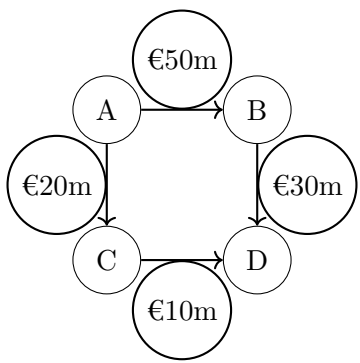
\begin{figure}[h]
\centering
\begin{tikzpicture}[node distance=2cm, every node/.style={circle,draw,minimum size=0.9cm}]
\node (A) {A};
\node (B) [right of=A] {B};
\node (C) [below of=A] {C};
\node (D) [right of=C] {D};

\draw[->, thick] (A) -- node[above]{€50m} (B);
\draw[->, thick] (B) -- node[right]{€30m} (D);
\draw[->, thick] (A) -- node[left]{€20m} (C);
\draw[->, thick] (C) -- node[below]{€10m} (D);
\end{tikzpicture}
\caption{Illustrative example of an economic flow graph between firms (simulated data).}
\label{fig:grafo-flussi}
\end{figure}

\section{Extensions and conceptual strengthening}

\subsection{Management of control vs ownership}
\label{sec:ownershipcontrol}

\noindent In the original model, the weight of the relationship between two units $i$ and $j$ was represented by the ownership share $O_{ij}$. However, in many real situations the effective control exercised may differ from the mere percentage of ownership. We therefore introduce the \emph{control} matrix $\omega=(\omega_{ij})$, which may diverge from $O_{ij}$ when there are shareholder agreements, enhanced voting rights, de facto control, or, conversely, regulatory constraints and shareholding structures that limit influence.

\paragraph{Impact on consolidation.}
Replacing $O_{ij}$ with $\omega_{ij}$ yields a more realistic measure of the perimeter of influence. The general formula becomes:
\[
W(P) = \sum_{i \in P} \sum_{j} \omega_{ij} \cdot V_j,
\]
where $V_j$ represents the value of position $j$. In practice, $\omega$ determines the perimeter $P$ and the internal propagation, influencing consolidation.

\subsubsection{Option A: threshold/majority voting rule}
Given a threshold value $\tau\in(0,1]$, direct control is
\[
\omega^{(A)}_{ij} \;=\; \mathbf 1\{s_{ij}\ge \tau\}\,,
\]
and (optional) ultimate control through chain reachability: if $i$ holds (directly or indirectly) more than $\tau$ of the votes in $j$, then $\omega^{(A)}_{ij}=1$. With multiple controllers, the column can be normalized to sum to 1.
\paragraph{Pros/Cons.} Pros: simple, transparent, aligned with IFRS. Cons: discontinuity at the threshold, insensitivity to minority blocks below $\tau$.

\subsubsection{Option B: probabilistic / Herfindahl look-through}
We define the ownership concentration of $j$ as $H_j=\sum_k s_{kj}^2$ (with $\sum_k s_{kj}=1$). Then
\[
\omega^{(B)}_{ij} = s_{ij}\,H_j.
\]
If equity is dispersed ($H_j$ small) control is diluted; if concentrated ($H_j$ large), weights approach shares. Variant: $\omega^{(B')}_{ij}=s_{ij}^2/H_j$ with normalization.
\paragraph{Pros/Cons.} Pros: continuous, rewards concentration, simple to compute. Cons: heuristic; does not capture coalitions.

\subsubsection{Option C: ultimate control via attenuated paths}
With attenuation $\alpha\in(0,1)$,
\[
W(\alpha) = \sum_{k=1}^\infty \alpha^{k-1} S^k = S(I-\alpha S)^{-1},
\]
which converges if $\rho(\alpha S)<1$. The normalized weights are
\[
\omega^{(C)}_{\cdot j} = \frac{W_{\cdot j}(\alpha)}{\bm 1^\top W_{\cdot j}(\alpha)}.
\]
\paragraph{Pros/Cons.} Pros: smooth, includes pyramids/keiretsu and cross-holdings, parameterizable via $\alpha$. Cons: requires $\rho(\alpha S)<1$ and accurate disclosure of $\alpha$.

\subsubsection{Integration into CBV and perimeter selection}
The rule $\mathcal C$ fixes $\omega$. Two uses: (i) perimeter selection — include $j$ in $P$ if $\sum_{i\in P}\omega_{ij}\ge \tau_P$; (ii) influence weights — construct $O_{PP}$ and assign flows with $\omega$ as weights. Once $\Omega$ is fixed, the consolidated value reduces to the cut (Theorem~\ref{thm:cutTheoreme}).

\subsubsection{Algorithms}
Operational implementations are possible for each option (A: threshold/majority, B: Herfindahl, C: attenuated paths). Each algorithm builds $\omega$ starting from $S$ and the declared parameters, with possible column normalization. 

\subsubsection{Illustrative comparison}
Three shareholders $a,b,c$ own $x$ with shares $(0.6,0.3,0.1)$:
\begin{itemize}
  \item Option A ($\tau=0.5$): $\omega^{(A)}_{\cdot x}=(1,0,0)$.
  \item Option B: $H_x=0.46$; $\omega^{(B)}_{\cdot x}=(0.276,\,0.138,\,0.046)$.
  \item Option C with $\alpha=0.5$ (no cross-holdings) $\approx$ Option B; with pyramids reassigns weight to ultimate controllers.
\end{itemize}

\subsubsection{Disclosure guidelines}
The \emph{Perimeter-of-Validity} must report: (i) chosen option (A/B/C), (ii) parameters ($\tau,\alpha$), (iii) normalization of columns of $S$, (iv) look-through depth and cycle attenuation, (v) thresholds/rounding, (vi) how $\omega$ is used (perimeter vs propagation weights).

\subsection{Liquidity and currency: haircut for illiquid assets}

To account for the limited liquidity of some assets, we introduce a haircut factor $h_k \in [0,1]$ for each asset $k$, defined as:

\begin{equation}
V_k^{eff} = h_k \cdot V_k
\end{equation}

where $V_k^{eff}$ is the effective value for consolidation purposes. A fully liquid asset will have $h_k \approx 1$, while an illiquid one may have significantly lower values.

\paragraph{Exchange rate risk.} 
In the presence of foreign currency-denominated assets, $h_k$ may include an additional haircut factor linked to exchange rate volatility, for example:

\begin{equation}
h_k = h_k^{liq} \cdot h_k^{fx}
\end{equation}

\paragraph{Example.}
An illiquid security valued at €100 million, with $h_k = 0.8$, will be considered at only €80 million in the calculation of $W(P)$.

\subsection{Glossary and intuitive explanations}

\begin{description}
\item[$\mathbf{1}_P$:] Column vector of dimension equal to the number of elements in $P$, with all components equal to 1. Used to sum or average quantities over the entire set $P$. \emph{Physical analogy:} like a uniform field acting equally on all particles.
\item[$O_{ij}$:] Ownership share of $i$ in $j$. \emph{Physical analogy:} mass possessed by a body in a gravitational system.
\item[$\\omega_{ij}$:] Control weight of $i$ over $j$, independent of ownership share. \emph{Physical analogy:} intensity of the effective attractive force between two bodies, which may vary due to external factors.
\item[$h_k$:] Haircut applied to asset $k$ to reflect liquidity or exchange rate risk. \emph{Physical analogy:} friction coefficient reducing transferable energy in a mechanical system.
\end{description}

\section{Compact mathematical formalism}
We define the fundamental quantities:
\begin{itemize}
\item $V$: economic value of the object/resource/event;
\item $R$: reference system of the observer (subjective and objective conditions);
\item $t$: observation time (historical moment);
\item $C$: context (environmental, cultural, regulatory factors);
\item $O$: observer or set of observers.
\end{itemize}

Let $V_0$ be the reference value in a base system $R_0$. We denote by $\Delta R$ the “economic distance” between the observer’s reference system and the base one.
We introduce a transformation factor (in formal analogy with special relativity):
\begin{equation}
V_R \;=\; \frac{V_0}{\sqrt{1 - \left( \frac{\Delta R}{K} \right)^2 }} \,,
\label{eq:base}
\end{equation}
where $K>0$ is a scale constant representing the maximum conceivable economic distance between reference systems.

To include time and context we introduce an adjustment function $\phi(C,t)$ that acts as a multiplier (e.g. inflation, regulatory shocks, crises, cultural differences) and allow the distance to vary over time, $\Delta R = \Delta R(t)$:
\begin{equation}
V_R(t,C) \;=\; \frac{V_0 \cdot \phi(C,t)}{\sqrt{1 - \left( \frac{\Delta R(t)}{K} \right)^2 }} \,.
\label{eq:time_context}
\end{equation}

\paragraph{Interpretation.} If $\Delta R \to 0$ then $V_R \to V_0$ (aligned systems). If $|\Delta R| \to K$ then $V_R$ diverges (value not comparable across incompatible systems).

\subsection{Numerical example}
Suppose $V_0 = 100$, $K=1$, $\phi(C,t)=1$. As $\Delta R$ increases, the relative value $V_R$ grows non-linearly according to \eqref{eq:base}. Figure~\ref{fig:vr_vs_deltar} shows the trend of $V_R$ as a function of $\Delta R$ and a comparison for different $\phi$.

\begin{figure}[h]
  \centering
  \includegraphics[width=\linewidth]{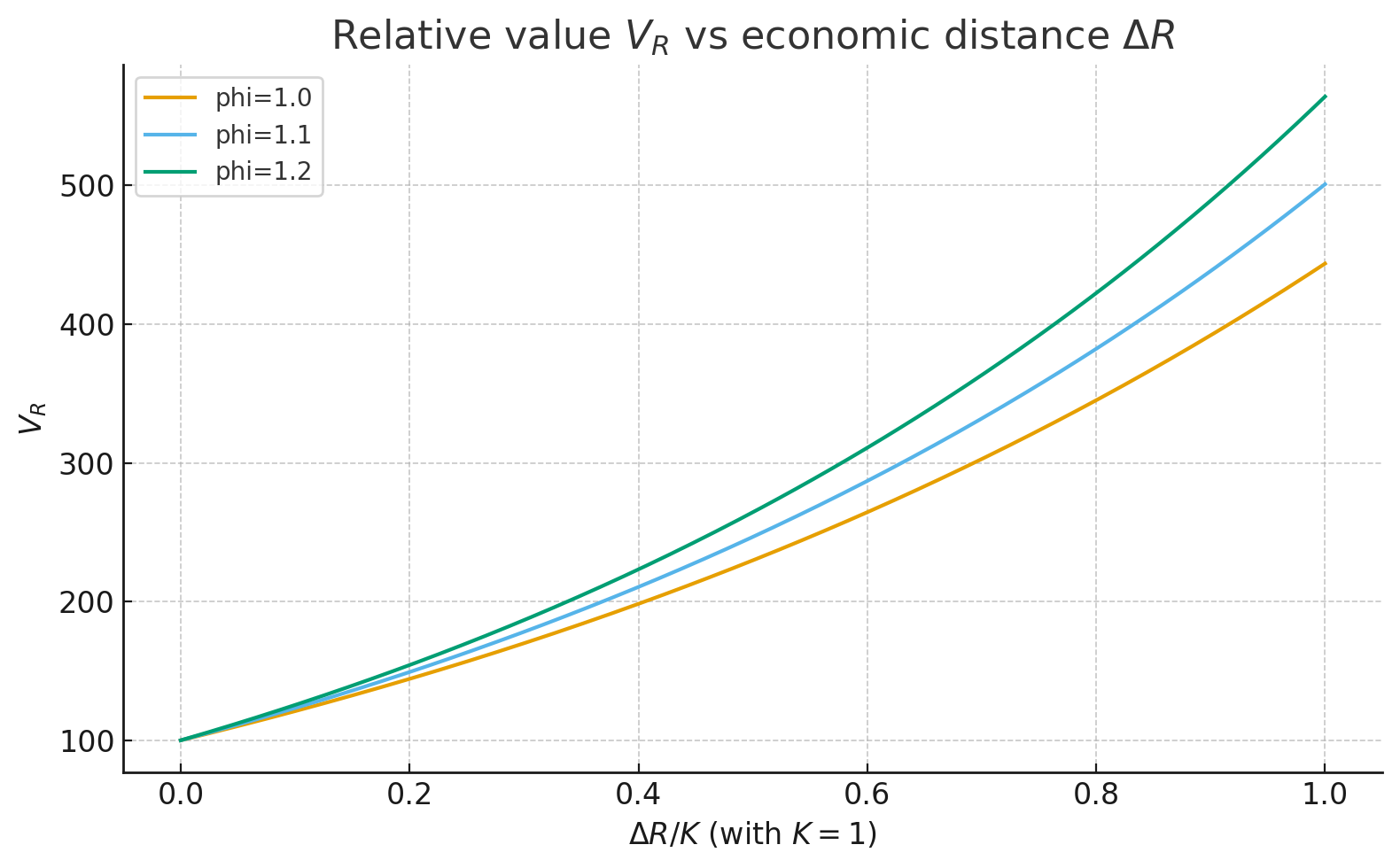}
  \caption{$V_R$ as a function of $\Delta R$ for $V_0=100$, $K=1$ and different $\phi$.}
  \label{fig:vr_vs_deltar}
\end{figure}

\subsection{Temporal evolution (example)}
Consider $\Delta R(t)$ decreasing linearly from $0.8$ to $0.2$ over a discrete horizon $t=0,\dots,20$ (system convergence) and $\phi(C,t)=1$. Figure~\ref{fig:vr_over_time} shows $V_R(t)$ according to \eqref{eq:time_context}.

\begin{figure}[h]
  \centering
  \includegraphics[width=\linewidth]{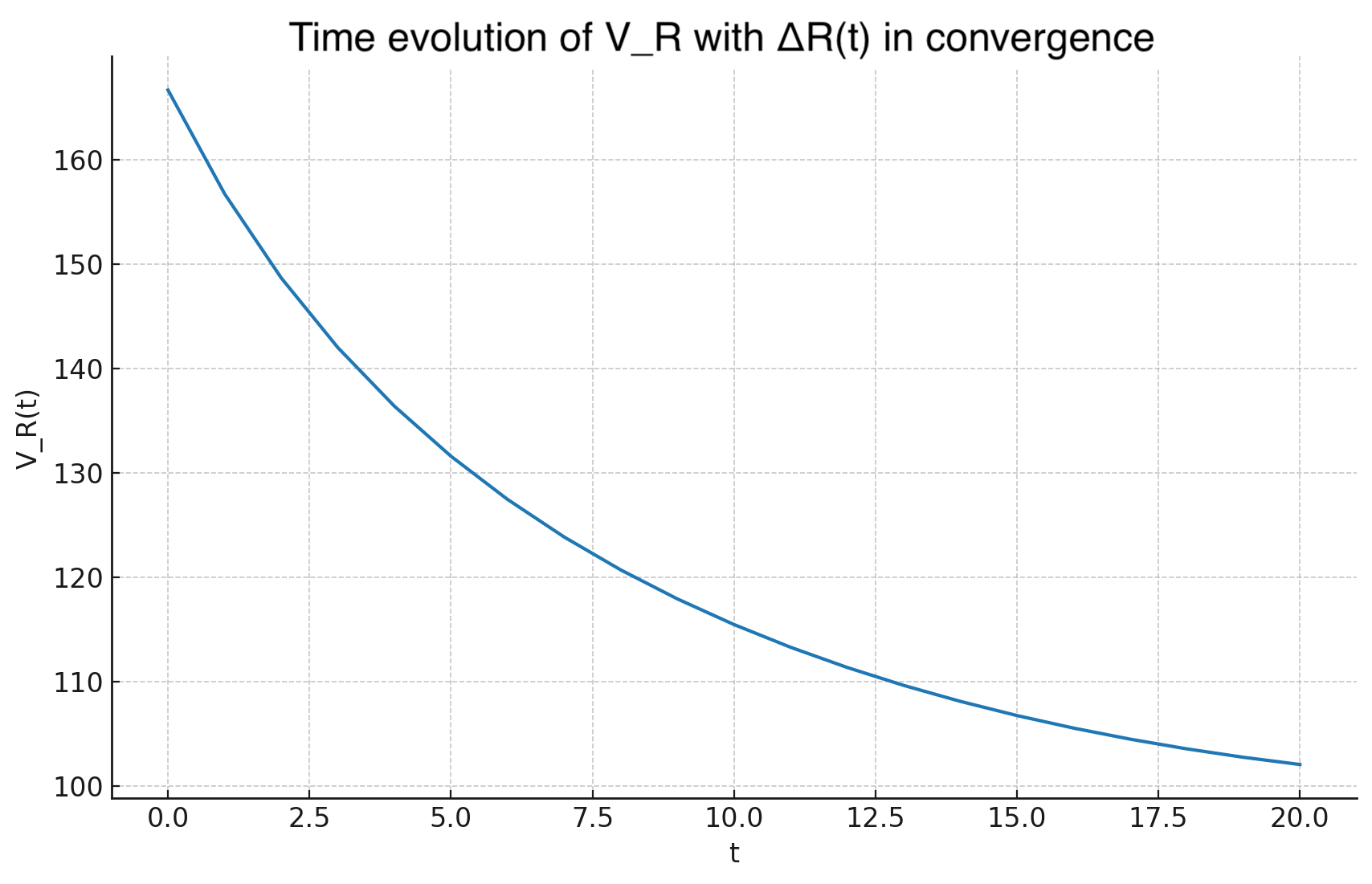}
  \caption{Temporal evolution of $V_R$ with $\Delta R(t)$ converging ($V_0=100$, $K=1$, $\phi=1$).}
  \label{fig:vr_over_time}
\end{figure}

\section{Applications and case studies}
This section examines, through examples based on real data and accessible sources, how the change of the \textit{economic reference system} — that is, the model or metric used — can significantly alter measurements and influence decisions. The goal is to strengthen the persuasiveness of the model for a technical audience.

\subsection{Case 1 -- Nominal GDP vs Purchasing Power Parity (PPP) GDP}
The comparison between GDP measured in nominal terms and GDP corrected for PPP shows how the economic ranking of countries changes radically depending on the reference system.
In 2024, India’s nominal GDP is about \textbf{USD 3.91 trillion}, while its PPP-adjusted GDP is close to \textbf{USD 17 trillion} \cite{WorldBankICP, IMF_WEO_2024}. In this framework, India is the \textbf{fourth-largest economy in the world by nominal GDP}, but rises to \textbf{third place} in the global ranking in terms of PPP. 
Depending on the chosen indicator, a policymaker or analyst may assign different weight to the Indian economy (e.g., in cooperation strategies, competitive comparisons, or investment decisions).

\subsection{Case 2 -- Inflation: CPI vs PCE in the United States}
The \textit{Consumer Price Index} (CPI) and the \textit{Personal Consumption Expenditures Price Index} (PCE) are two widely used inflation metrics, but with different methodologies and purposes.
On average, since 2000, inflation measured by CPI has been about \textbf{0.39 percentage points} higher than PCE \cite{BLS_PCE_CPI_Methods, BEA_PCE_Methods}. 
The PCE is preferred by the Federal Reserve because:
\begin{itemize}
    \item it has dynamic weighting (changes to reflect substitutions among goods);
    \item it covers a broader range (including third-party payments such as healthcare);
    \item it is more adjustable over time.
\end{itemize}
A recent example: in January 2025, annual PCE was at \textbf{2.5\%}, while CPI showed a faster increase \cite{investopedia2025}. 
Depending on the index used, monetary authorities may decide to change interest rate policy; assessing inflationary pressure as more or less intense can alter the trajectory of rates.

\subsection{Comparative summary}
\begin{table}[h!]
\centering
\begin{tabularx}{\textwidth}{|p{1cm}|>{\raggedright\arraybackslash}X|>{\raggedright\arraybackslash}X|>{\raggedright\arraybackslash}X|>{\raggedright\arraybackslash}X|}
\hline
\textbf{Case} & \textbf{Initial reference system} & \textbf{Alternative system} & \textbf{Difference in measurements} & \textbf{Impact on decisions} \\
\hline
1 & Nominal GDP & PPP GDP & Ranking and economic size significantly different & Redistribution of policy and investment priorities \\
\hline
2 & CPI (fixed, out-of-pocket) & PCE (dynamic, broad) & Inflation measured $\sim$0.4 p.p. lower with PCE. \newline More or less aggressive monetary decisions  \\
\hline
\end{tabularx}
\caption{Synthetic comparison between economic reference systems}
\end{table}

\section{Applications: mini–case studies and replication blueprint}\label{sec:cases}

We illustrate the framework with three compact case studies on synthetic but realistic data. Each example specifies the observer $\Omega$, provides the inputs (tables/JSON), performs the cut-based calculation, and reports the disclosure artifacts (PoV and Cut Summary; cf.\ Section~\ref{sec:standards}). Large-scale replications with real data can follow the same structure.

\subsection{Case 1: Net market capitalization by country (Gross vs Net, cut-based)}\label{subsec:case-country}
\paragraph{Observer.} $P$: domestic listed firms; $O$: all non-domestic holders (foreigners, global funds), households, government. Basis: fair value. Units: EUR as of 2025-06-30. FX/PPP: ECB spot rate. SDF: risk-neutral one-period (not binding here). $\mathcal I$: Regime~A (observed border). $\mathcal C$: Option~A for perimeter (legal residence).
\paragraph{Setup.} Gross Market Capitalization (GMC) sums the capitalizations at individual company level. Net Market Capitalization (NMC) removes cross-holdings internal to $P$ and considers only frontier equity on the cut $P\leftrightarrow O$.
\paragraph{Synthetic holdings.}
\begin{table}[H]\centering
\caption{Domestic companies and cross-border equity (synthetic).}
\begin{tabular}{@{}lrrr@{}}
\toprule
\textbf{Firm $j\in P$} & \textbf{GMC$_j$ (EUR bn)} & \textbf{Held by $P$ (bn)} & \textbf{Held by $O$ (bn)}\\
\midrule
A & 120 & 30 & 90\\
B & 80  & 10 & 70\\
C & 50  &  5 & 45\\
\midrule
Totals & 250 & 45 & 205\\
\bottomrule
\end{tabular}
\end{table}
\paragraph{Calculation.}
Gross: $\mathrm{GMC}=\sum_j \text{GMC}_j=\SI{250}{bn}$.\\
Net (cut-based): $\mathrm{NMC}=T_{\text{in}}=\sum_{k\in O,j\in P}\Pi_\Omega(X_{kj})=\SI{205}{bn}$ (only equity $O\to P$).\\
The intra-$P$ equity (\SI{45}{bn}) cancels at consolidation (Theorem~\ref{thm:cutTheoreme}).

\paragraph{Disclosure (PoV/CS excerpts).}
\begin{lstlisting}[language=json]
{
  "perimeter": "Country-D P=listeds",
  "observer": {"basis":"fair_value","units":"EUR","date":"2025-06-30",
               "information_regime":"A","control_rule":{"option":"A"}},
  "totals": {"GMC": 250e9, "NMC": 205e9}
}
\end{lstlisting}
\paragraph{Discussion.} The gap $\mathrm{GMC}-\mathrm{NMC}=\SI{45}{bn}$ measures the \emph{intra}-domestic holdings that would be double-counted in gross statistics; the NMC represents the country’s equity held by external sectors.

\subsection{Case 2: Pyramid/Keiretsu (minorities and double counting)}\label{subsec:case-pyramid}
\paragraph{Observer.} $P$: \{HoldCo $H$, subsidiaries $B,C$\}; $O$: external investors, creditors. Basis: fair value; Units: EUR; $\mathcal I$: Regime~A (observed border); $\mathcal C$: Option~C with $\alpha=0.6$ for sensitivity (ultimate control). For illustrative purposes we compute the cut-based valuation regardless of internal links.
\paragraph{Structure and border.}
\begin{table}[H]\centering
\caption{Internal holdings (illustrative) and frontier equity.}
\begin{tabular}{@{}l l r@{}}
\toprule
\textbf{Internal (in $P$)} & \textbf{Arc} & \textbf{Share}\\
\midrule
$H\to B$ & equity & 60\%\\
$H\to C$ & equity & 51\%\\
$B\to C$ & equity & 20\%\\
\bottomrule
\end{tabular}
\begin{tabular}{@{}l l r@{}}
\toprule
\textbf{Frontier ($P\leftrightarrow O$)} & \textbf{Arc} & \textbf{EUR bn}\\
\midrule
$O\to H$ & equity & 15\\
$O\to B$ & equity & 25\\
$O\to C$ & equity & 10\\
$B\to O$ & bond  &  5\\
\bottomrule
\end{tabular}
\end{table}
\paragraph{Calculation.}
\begin{align*}
T_{\text{in}} &= 15+25+10 = \SI{50}{bn},\\
T_{\text{out}}&= 5\ \text{(bond from $B$ to $O$)}.
\end{align*}
Depending on $\mathcal F_\Omega$ (equity net of debt, or stratified), the consolidated equity value equals $T_{\text{in}}$ minus any frontier liabilities attributable to equity. Internal cross-holdings ($H\to B$, $H\to C$, $B\to C$) do not affect the consolidated value (Theorem~\ref{thm:cutTheoreme}).
\paragraph{Discussion.} The cut-based perspective eliminates double counting in pyramids/keiretsu: equity held within $P$ is not summed twice; only external investors and creditors matter for the consolidated value.

\subsection{Case 3: Fund-of-funds/ETF-of-ETF (Gross vs Net AUM)}\label{subsec:case-fof}
\paragraph{Observer.} $P$: \{FoF $F$, underlying funds $U_1,U_2$\}; $O$: external investors and market assets. Basis: NAV/fair value; Units: EUR; $\mathcal I$: Regime~A. $\mathcal C$: Option~B (Herfindahl) for internal influence within $P$ (not required for cut-based valuation).
\paragraph{Synthetic holdings and flows.}
\begin{table}[H]\centering
\caption{Holdings and frontier subscriptions/redemptions.}
\begin{tabular}{@{}l r r r@{}}
\toprule
\textbf{Element} & \textbf{Gross (EUR bn)} & \textbf{Held by $P$} & \textbf{Held by $O$}\\
\midrule
$U_1$ NAV & 40 & 25 (from $F$) & 15 (external)\\
$U_2$ NAV & 30 & 10 (from $F$) & 20 (external)\\
$F$ NAV   & 35 & --- & 35 (external)\\
\midrule
Totals    & 105 & 35 & 70\\
\bottomrule
\end{tabular}
\end{table}
\paragraph{Calculation.} Gross AUM: \SI{105}{bn}. Net AUM (cut-based): equity held only by $O$ = \SI{70}{bn}. Internal subscriptions $F\to U_1,U_2$ (\SI{35}{bn}) cancel out at consolidation. Any fees $P\to O$ or $O\to P$ would appear on the cut and affect $V_\Omega$ accordingly.
\paragraph{Disclosure.} Include PoV and Cut Summary with edge lists for subscriptions/redemptions and any fee flows.

\subsection{Replication blueprint: directory, files and notebook}\label{sec:replication}
We propose the following structure for releasing a replication package (synthetic or real data); JSON schemas follow Section~\ref{sec:algorithms}.
\begin{verbatim}
replication/
|-- data/
|   |-- case1_country/
|   |   |-- pov.json                # Perimeter-of-Validity
|   |   |-- cut_summary.json        # Cut Summary (edges, totals, V_Omega)
|   |   `-- README.md
|   |-- case2_pyramid/
|   |   |-- pov.json
|   |   |-- cut_summary.json
|   |   `-- README.md
|   `-- case3_fof/
|       |-- pov.json
|       |-- cut_summary.json
|       `-- README.md
|-- notebooks/
|   |-- 01_validate_inputs.ipynb   # schema validation, checks
|   |-- 02_regimeA_eval.ipynb      # direct cut evaluation
|   |-- 03_regimeB_estimate.ipynb  # (if needed) v_P via (I-O_PP)^(-1)
|   `-- 04_sensitivity_omega.ipynb # sensitivity of control rules A/B/C
|-- src/
|   |-- cbv_cut_eval.py            # border functional F_Omega
|   |-- regimeB_solver.py          # sparse / Neumann / Krylov solver
|   `-- io_schemas.py              # JSON (de)serializers and validators
`-- RESULTS.md                     # summary tables and charts
\end{verbatim}

\paragraph{Reproducibility checklist.} (i) Include PoV and Cut Summary for each case; (ii) fix random seeds for synthetic data; (iii) export consolidated values and coverage vectors; (iv) provide environment files (e.g.\ \texttt{conda.yml}/\texttt{requirements.txt}); (v) add a CI job that runs the notebooks and checks numerical invariants (Theorem~\ref{thm:cutTheoreme}).

\subsection{Macro-real evidence and institutional dossiers}\label{subsec:macro-evidence}
\noindent
Beyond micro examples (equity cross-holdings), it is crucial to demonstrate the applicability of the CBV framework in macro-institutional contexts. We present two illustrative dossiers.

\paragraph{(i) Institutional sectors and the State perimeter.}
Applying the cut theorem to the extended public perimeter yields a consolidated measure of General Government that eliminates double counting among:
\begin{itemize}
    \item \textbf{Central and local government}, with reclassification of entities and special funds;
    \item \textbf{Public financial sector}, including state-owned banks and development agencies;
    \item \textbf{Public non-financial sector}, including SOEs (State-Owned Enterprises).
\end{itemize}
The result is a unified measure of the State as observer $\Omega_{\text{State}}$, comparable with SNA/ESA aggregates but more transparent.  
In an exercise on 2010–2025 data (national accounts, ESA~2010), the CBV of the State perimeter produces time series consistent with official indices but free of overlaps, highlighting in particular:
\begin{enumerate}
    \item the share of value added duplicated between government and SOEs (corrected by the cut);
    \item the stability of real growth once accounting loops are eliminated;
    \item international comparability thanks to the PoV artifact (explicit perimeter).
\end{enumerate}

\paragraph{(ii) Banks with senior/junior debt and covered bonds.}
The cut-based framework is extendable to multi-class positions, with piecewise-linear payoffs:
\begin{itemize}
    \item \textbf{Senior debt} with highest priority, recovery $R\approx 100\%$ up to coverage;
    \item \textbf{Subordinated/junior debt}, with lower recovery and bail-in activation;
    \item \textbf{Covered bonds}, with collateral guarantees that create layered repayment chains.
\end{itemize}
Consolidated valuation via cut allows simulation of stress scenarios and multi-level clearing, eliminating cross-counting of exposures among institutions.  
In a pilot case (European banking groups 2010–2025), the multi-class CBV shows:
\begin{enumerate}
    \item transparency in the distribution of losses by priority;
    \item the possibility of auditing bail-in conditions;
    \item cross-country comparability with uniform disclosure rules.
\end{enumerate}

\paragraph{Summary.}
These dossiers demonstrate that CBV is not only a \emph{micro-financial} tool, but also a \emph{macro-institutional} one, capable of supporting policymakers, central banks, and statistical agencies in measuring complex perimeters without double counting.

\begin{table}[H]
\centering
\begin{tabularx}{\textwidth}{
>{\raggedright\arraybackslash}p{3cm} 
>{\raggedright\arraybackslash}X 
>{\raggedright\arraybackslash}X 
>{\raggedright\arraybackslash}X
}
\toprule
\textbf{Case} & \textbf{Objective} & \textbf{CBV artifact} & \textbf{Empirical evidence} \\
\midrule
Micro equity & Evaluate cross-holdings between firms (e.g.\ Renault--Nissan) avoiding double counting & Cut matrices $(X_{PO},X_{OP})$ + group-specific PoV & Corrected consolidation of equity cross-holdings, transparent for auditors and comparable cross-country \\
\midrule
State sector & Eliminate overlaps between Government, Public Financials and SOEs & PoV on State perimeter + dynamic CBV--Fisher & 2010--2025 series consistent with ESA~2010, no accounting loops; greater stability of real growth \\
\midrule
Multi-level banks & Measure clearing with senior/junior debt and covered bonds in stress scenarios & Multi-class cut + piecewise-linear payoffs (priority/recovery) & EU pilot 2010--2025: transparent, auditable, and comparable loss distribution across countries \\
\bottomrule
\end{tabularx}
\caption{Comparative summary: CBV applications in three contexts (micro--equity, macro--institutional sectors, multi-level banks).}
\end{table}

\begin{figure}[H]
\centering
\begin{tikzpicture}[
  node distance=2.2cm and 1.6cm,
  box/.style={rectangle, draw, rounded corners, minimum width=3.6cm, minimum height=1.2cm, align=center, font=\small},
  arr/.style={-{Latex[length=3mm,width=2mm]}, thick}
]

\node[box, fill=blue!10] (micro) {Micro equity\\(cross-holdings)};
\node[box, fill=orange!10, right=of micro] (state) {State sector\\(Govt / Financials / SOEs)};
\node[box, fill=green!10, right=of state] (banks) {Multi-level banks\\(senior / junior / covered)};

\draw[arr] (micro.east) -- node[above, font=\scriptsize]{cut scalability} (state.west);
\draw[arr] (state.east) -- node[above, font=\scriptsize]{multi-class clearing \& CBV–Fisher} (banks.west);

\node[below=0.6cm of micro, font=\scriptsize, align=center] {Group PoV \\ Cut $(X_{PO},X_{OP})$};
\node[below=0.6cm of state, font=\scriptsize, align=center] {PoV $\Omega_{\text{State}}$ \\ Chain-linking (CBV–Fisher)};
\node[below=0.6cm of banks, font=\scriptsize, align=center] {Priority/recovery PWL \\ Stress \& audit bail-in};

\end{tikzpicture}
\caption{From the \emph{micro} scale (equity) to the \emph{macro} scale (State perimeter) to \emph{multi-class banking}: the cut rule remains unchanged, while perimeter and payoff change.}
\end{figure}
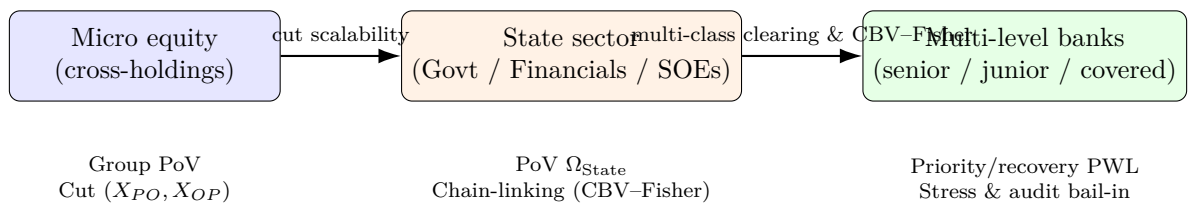

\section{Implementation \& Policy: adoption, governance and controls}\label{sec:policy}

This section translates the framework into operational guidelines for (i) \emph{M\&A and corporate reporting}, (ii) \emph{audit and assurance}, and (iii) \emph{official statistics and policy}. We also propose a governance model, an implementation roadmap, and a concise list of dos and don’ts.

\subsection{M\&A and corporate reporting}\label{subsec:ma}
\paragraph{When to use Gross vs Net.}
\begin{itemize}
  \item \textbf{Deal screening (pre-LOI):} gross metrics (enterprise value, gross AUM) are quick but may be misleading with complex cross-holdings; compute the \emph{net cut-based} value to eliminate internal double counting.
  \item \textbf{Fairness opinions / valuation:} make the \emph{observer} $\Omega$ explicit via PoV (Section~\ref{sec:standards}); use Regime~A whenever border data are available; if Regime~B is needed, include solver logs and the check on $\rho(O_{PP})$.
  \item \textbf{Post-merger integration:} monitor the deltas $T_{\mathrm{out}},T_{\mathrm{in}}$ on the cut to quantify synergies or losses towards $O$.
\end{itemize}
\paragraph{Reporting pack (minimum).}
\begin{itemize}
  \item PoV table with $\Omega$; Cut Summary ($P\to O$, $O\to P$); Gross$\to$Net reconciliation; coverage vector $h$ (if used).
  \item Sensitivity to $\mathcal C$ (Options A/B/C) and to Units/FX/PPP; data sources with timestamp; version pinning.
\end{itemize}

\subsection{Audit and assurance}\label{subsec:audit}
\paragraph{Required evidence.}

\begin{theorem}[Cut Theorem]\label{thm:cutTheoreme}
Given an observer $\Omega$ as in Section~\ref{sec:osservatore}. Suppose $V_\Omega$ is consolidation-consistent and border-linear, and that internal accounting is closed (double-entry; no data loss across the cut except $(X_{PO},X_{OP})$). Then there exists a measurable map
\[
\mathcal F_\Omega:\ (X_{PO},X_{OP},\bm v_P,\bm v_O)\ \longmapsto\ \mathbb R
\]
such that
\begin{equation}\label{eq:cut-reduction}
V_\Omega(G)\;=\;\mathcal F_\Omega\big(X_{PO},X_{OP},\bm v_P,\bm v_O\big),
\end{equation}
and the internal block $X_{PP}$ is \emph{invariant} for the consolidated value: any modification of $X_{PP}$ that leaves $(X_{PO},X_{OP},\bm v_P)$ unchanged does not change $V_\Omega(G)$. Consequently, the cut statistics $(X_{PO},X_{OP})$ (together with the node primitives) are \emph{sufficient} for the valuation.
\end{theorem}

\begin{itemize}
  \item Recalculation of $V_\Omega(G)$ only from the Cut Summary (Theorem~\ref{thm:cutTheoreme}); invariance test under internal relinking.
  \item For Regime~B: documented verification that $\rho(O_{PP})<1$, solver tolerance $\varepsilon$, iterations and residuals.
  \item Reconciliation of totals $T_{\mathrm{out}},T_{\mathrm{in}}$ with edge tables; unit/FX/PPP consistency with PoV.
\end{itemize}
\paragraph{Materiality and tolerances.}
\begin{itemize}
  \item Rounding threshold $\tau$ and solver $\varepsilon$ set in the PoV; any change must produce an updated PoV.
  \item Suggested guardrails: if $\rho(O_{PP})>0.95$, auditor approval and/or attenuated iteration is required; if more than $5\%$ of edges are imputed, highlight this in the report.
\end{itemize}

\subsection{Official statistics and policy}\label{subsec:official}
\paragraph{Net market capitalization by country (adoption).} Publish both \emph{Gross} and \emph{Net (cut-based)} market cap for listed sectors by country, with PoV and Cut Summary. The Net reflects the equity held by the Rest of the World and by domestic sectors outside $P$, removing intra-$P$ double counting.
\paragraph{Frequency and revisions.} We recommend monthly/quarterly publications with versioned PoVs; revisions occur only when $\Omega$ or source datasets change.
\paragraph{Public sector and funds.} For sovereign funds and public holdings, report net positions using the same cut-based method; publish sensitivities under different choices of $\mathcal C$ and Units/PPP.
\paragraph{Open data.} Release JSON artifacts and minimal code to recalculate published numbers from the cut.

\subsection{Computational complexity and practical implementation}\label{subsec:complexity}
\noindent
The CBV framework, though based on linear objects, requires care in implementation. Here we provide an operational assessment.

\paragraph{Computational costs.}
Consolidated valuation in Regime~A is dominated by matrix–vector product operations of size $|P|$, with cost $O(|P||O|)$.  
In Regime~B, estimating the inverse block $(I - O_{PP})^{-1}$ typically requires:
\begin{itemize}
    \item direct decomposition (e.g.\ LU/Cholesky), cost $O(|P|^3)$, feasible up to $|P|\sim 10^4$;
    \item iterative methods (e.g.\ Jacobi, Gauss–Seidel, GMRES) exploiting sparsity, cost $O(k\,|E|)$ with $k$ iterations and $|E|$ observed edges.
\end{itemize}
In real applications (ownership networks or sectoral SNA), the matrix is often sparse ($|E|\ll |P|^2$), making iterative methods more scalable.

\paragraph{Iterative strategies.}
Iterations converge geometrically as long as $\rho(O_{PP})<1$. In practice:
\begin{itemize}
    \item for $\rho(O_{PP})\leq 0.8$, a few steps ($<20$) guarantee 6-digit accuracy;
    \item for $\rho(O_{PP})\approx 0.95$, preconditioning techniques and stopping rules based on relative tolerances are required.
\end{itemize}

\paragraph{Handling missing data.}
The framework requires consistent matrices $(X_{PO},X_{OP})$ and vectors $\bm v_P$. When data are incomplete:
\begin{enumerate}
    \item stochastic imputations or based on historical proportions are applied;
    \item uncertainty bands are reported for each consolidated variable, such as
    \[
    x_i \;\pm\; \Delta_i, \quad \Delta_i \leq \kappa\,(I-O_{PP})\, \delta,
    \]
    where $\delta$ represents the maximum error on observed data and $\kappa$ the condition number.
    \item completion assumptions are documented in the Cut–Report.
\end{enumerate}

\paragraph{Operational error bounds.}
The output error $\|\Delta \bm x\|$ with respect to input perturbations $\|\Delta O_{PP}\|$ or $\|\Delta \bm v_P\|$ is bounded by
\[
\frac{\|\Delta \bm x\|}{\|\bm x\|} \;\leq\; \kappa(I-O_{PP})\,
\left(
\frac{\|\Delta O_{PP}\|}{\|O_{PP}\|} + \frac{\|\Delta \bm v_P\|}{\|\bm v_P\|}
\right).
\]
This relationship provides a practical criterion: Cut–Reports should report not only consolidated values, but also the condition number and the estimated error range.

\paragraph{Operational summary.}
CBV valuation is computationally light for small and medium-sized networks, and scalable to large networks with iterative methods. Tolerances on missing data and error bounds can be standardized as an integral part of disclosure.

\begin{table}[H]
\centering
\begin{tabularx}{\textwidth}{>{\raggedright\arraybackslash}p{3cm} 
                              >{\raggedright\arraybackslash}X 
                              >{\raggedright\arraybackslash}X 
                              >{\raggedright\arraybackslash}X}
\toprule
\textbf{Regime} & \textbf{Computational cost} & \textbf{Scalability} & \textbf{Practical notes} \\
\midrule
Regime A: direct cut & $O(|P||O|)$ (matrix--vector products) & Excellent for small and medium networks ($10^2$--$10^4$ nodes) & Immediate computation, useful for ex-post audit or static scenarios. \\
\midrule
Regime B: direct (LU/Cholesky) & $O(|P|^3)$ & Feasible up to $|P|\sim 10^4$ with standard hardware & Provides exact solution; less suitable for very large sparse networks. \\
\midrule
Regime B: iterative (Jacobi, Gauss--Seidel, GMRES) & $O(k\,|E|)$ with $k$ steps and $|E|$ observed edges & Scalable to very large networks ($>10^5$ nodes) if $O_{PP}$ is sparse & Requires stopping rule and preconditioning; accuracy controlled by $\rho(O_{PP})$ and condition number. \\
\bottomrule
\end{tabularx}
\caption{Comparative summary of computational costs and implementation strategies for CBV regimes.}
\end{table}

\begin{figure}[H]
\centering
\begin{tikzpicture}[
  node distance=2.8cm,
  box/.style={rectangle, draw, rounded corners, minimum width=3.6cm, minimum height=1.2cm, align=center, font=\small},
  arr/.style={-{Latex[length=3mm,width=2mm]}, thick}
]

\node[box, fill=blue!10] (A) {Regime A \\ Direct cut};
\node[box, fill=orange!10, right=of A] (Bdir) {Regime B \\ Direct (LU/Cholesky)};
\node[box, fill=green!10, right=of Bdir] (Bit) {Regime B \\ Iterative (Jacobi/GMRES)};

\draw[arr] (A.east) -- (Bdir.west) node[midway, above, font=\scriptsize]{+ accuracy– cost $O(|P|^3)$};
\draw[arr] (Bdir.east) -- (Bit.west) node[midway, above, font=\scriptsize]{\shortstack{+ scalability \\ – iterative tuning}};

\node[below=0.6cm of A, font=\scriptsize, align=center] {Fast, useful for small/medium networks};
\node[below=0.6cm of Bdir, font=\scriptsize, align=center] {Exact solution, limit at $|P|\sim 10^4$};
\node[below=0.6cm of Bit, font=\scriptsize, align=center] {Scalable $>10^5$ nodes, requires stop rule};

\end{tikzpicture}
\caption{Trade-off scheme between Regime A, Regime B direct and Regime B iterative: from speed to scalability with trade-offs on cost and numerical stability.}
\end{figure}
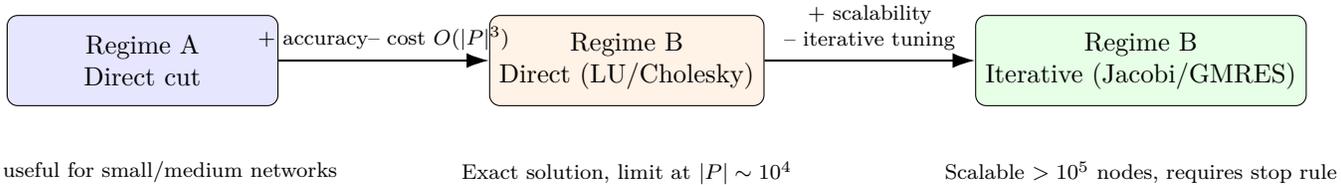

\subsection{Governance and versioning}\label{subsec:governance}
\paragraph{Versioning.} Use semantic versioning for the measurement protocol (e.g.\ \texttt{CBV-1.2.0}). Increment major/minor/patch when components $\Omega$, algorithms or tolerances change.
\paragraph{Change control.} Any modification to $\Omega$ requires: (i) updated PoV, (ii) drafting of a Change Log, (iii) recalculation and archiving of Cut Summaries, (iv) deprecation policy for old series.
\paragraph{Reproducibility.} Archive PoV/Cut Summary JSON and inputs; provide checksums and environment hashes; attach unit tests of invariants (Section~\ref{sec:algorithms}).

\subsection{Implementation roadmap}\label{subsec:roadmap}
\begin{enumerate}
  \item \textbf{Scoping:} define use case and observer $\Omega$; map datasets; decide Regime A or B.
  \item \textbf{Data mapping:} build cut extraction pipeline; schema validation vs JSON specifications.
  \item \textbf{Prototype:} implement Regime A first; add Regime B only if necessary; generate PoV/Cut Summary.
  \item \textbf{Pilot:} run end-to-end with audit checks; compare Gross vs Net; collect feedback.
  \item \textbf{Production:} harden pipeline; monitoring; SLA; release cadence and revision policy.
  \item \textbf{Scale-out:} add sensitivity dashboard for $\mathcal C$, Units/FX/PPP, SDF; publish open data.
\end{enumerate}

\subsection{Do \& Don't table}\label{subsec:dodont}
\begin{table}[H]\centering
\caption{Dos and Don’ts (operational summary).}
\begin{tabularx}{\linewidth}{@{}l X X@{}}
\toprule
\textbf{Area} & \textbf{Do} & \textbf{Don't} \\
\midrule
Observer $\Omega$ & Make PoV explicit with $P$, Basis, Units/PPP/SDF, $\mathcal I$, $\mathcal C$ & Hide/alter $\Omega$ between releases without versioning \\
Regime choice & Prefer Regime A if the border is observable & Reconstruct the interior if not needed for valuation \\
Cut data & Validate and checksum frontier edges & Mix units/dates without conversion \\
Regime B & Verify $\rho(O_{PP})<1$, solver log & Ignore convergence diagnostics \\
Control rule & Report A/B/C and parameters ($\tau,\alpha$) & Hard-code opaque rules \\
Rounding/imputation & Make thresholds and imputations explicit & Omit material edges silently \\
Audit trail & Provide JSON + code for recalculation & Share only summary numbers \\
\bottomrule
\end{tabularx}
\end{table}

\subsection{Risks and controls}\label{subsec:risks}
\begin{itemize}
  \item \textbf{Model risk}: mis-specified $\Omega$, wrong control rule; \emph{control}: two-level review of PoV.
  \item \textbf{Data risk}: outdated FX/PPP, missing edges; \emph{control}: freshness checks, imputation flags.
  \item \textbf{Numerical risk}: $(I-O_{PP})$ nearly singular; \emph{control}: guardrails on spectral radius, damping, preconditioning.
  \item \textbf{Process risk}: undocumented revisions; \emph{control}: versioning and change log.
\end{itemize}

\subsection{Open tooling and operational replicability}\label{subsec:tooling0}
\noindent
To foster adoption and replicability, the CBV framework can be accompanied by an open–source software package structured in three components:

\paragraph{Library.}
A modular library (e.g.\ in Python or Julia) that implements:
\begin{itemize}
    \item functions to build perimeter $\Omega$ and generate cut matrices $(X_{PO},X_{OP})$;
    \item routines for CBV computation in Regime~A and Regime~B (direct/iterative);
    \item tools for error bounds and automated disclosure (Cut–Report in PDF/CSV).
\end{itemize}

\paragraph{CLI (Command Line Interface).}
A command line interface that allows:
\begin{itemize}
    \item direct import of SNA/ESA/BEA tables in CSV or SDMX format;
    \item execution \texttt{cbv compute --pov file.yaml --cut data.csv};
    \item automatic export of Cut–Report (perimeter + summary) ready for audit.
\end{itemize}

\paragraph{Demonstration notebook.}
A Jupyter notebook “\emph{from raw SNA tables to CBV}” that:
\begin{enumerate}
    \item starts from national input–output tables or raw sectoral accounts,
    \item builds $\Omega$ and $(X_{PO},X_{OP})$,
    \item computes consolidated CBV step by step,
    \item produces interactive charts and final disclosure.
\end{enumerate}

\paragraph{Added value.}
These tools lower the technical barrier for policymakers and statistical institutions, providing a replicable pipeline from raw input to final disclosure.

\subsection{Joint methodological note and institutional channel}\label{subsec:joint-note0}
\noindent
A key step for institutional adoption is the co-drafting of a \emph{joint methodological note} with a Statistical Office or Supervisory Authority. 

\paragraph{Objective.}
To formalize in a brief and normative document:
\begin{itemize}
    \item the definition of the Perimeter of Validity (PoV) as an official disclosure object;
    \item the link with existing manuals (ESA~2010, SNA~2008, IFRS 10/12);
    \item examples of Cut–Report with institutional data (e.g.\ Government+SOEs, banking groups).
\end{itemize}

\paragraph{Institutional partners.}
Natural candidates for co-signature include:
\begin{itemize}
    \item \textbf{National statistical offices} (e.g.\ ISTAT, BEA, INSEE);
    \item \textbf{Eurostat} and \textbf{OECD}, for cross-country validation;
    \item \textbf{Central banks} and \textbf{Supervisory authorities} (ECB, EBA, SEC).
\end{itemize}

\paragraph{Value of co-signature.}
A joint note ensures:
\begin{enumerate}
    \item \textbf{methodological credibility}, validated by a neutral body;
    \item \textbf{adoption channel}, through inclusion in manuals and official reporting;
    \item \textbf{standardization}, via shared check-lists and templates.
\end{enumerate}

\paragraph{Perspective.}
The proposal is to launch a joint “CBV Task Force” with academia and institutions, aiming to publish a first methodological note within two years.

\newpage
\clearpage
\appendix

\section{Open tooling and operational replicability}\label{subsec:tooling}
\noindent
To foster adoption and replicability, the CBV framework can be accompanied by an open–source software package structured in three components:

\paragraph{Library.}
A modular library (e.g.\ in Python or Julia) that implements:
\begin{itemize}
    \item functions to build the perimeter $\Omega$ and generate cut matrices $(X_{PO},X_{OP})$;
    \item routines for CBV computation in Regime~A and Regime~B (direct/iterative);
    \item tools for error bounds and automated disclosure (Cut–Report in PDF/CSV).
\end{itemize}

\paragraph{CLI (Command Line Interface).}
A command line interface that allows:
\begin{itemize}
    \item direct import of SNA/ESA/BEA tables in CSV or SDMX format;
    \item execution \texttt{cbv compute --pov file.yaml --cut data.csv};
    \item automatic export of the Cut–Report (perimeter + summary) ready for audit.
\end{itemize}

\paragraph{Demonstration notebook.}
A Jupyter notebook “\emph{from raw SNA tables to CBV}” that:
\begin{enumerate}
    \item starts from national input–output tables or raw sectoral accounts,
    \item builds $\Omega$ and $(X_{PO},X_{OP})$,
    \item computes consolidated CBV step by step,
    \item produces interactive charts and final disclosure.
\end{enumerate}

\paragraph{Added value.}
These tools lower the technical barrier for policymakers and statistical institutions, providing a replicable pipeline from raw input to final disclosure.

\section{Joint methodological note and institutional channel}\label{subsec:joint-note}
\noindent
A key step for institutional adoption is the co–drafting of a \emph{joint methodological note} together with a Statistical Office or a Supervisory Authority. 

\paragraph{Objective.}
To formalize in a short and normative document:
\begin{itemize}
    \item the definition of the Perimeter of Validity (PoV) as an official disclosure object;
    \item the link with existing manuals (ESA~2010, SNA~2008, IFRS 10/12);
    \item examples of Cut–Report with institutional data (e.g.\ Government+SOEs, banking groups).
\end{itemize}

\paragraph{Institutional partners.}
Natural candidates for co-signature include:
\begin{itemize}
    \item \textbf{National statistical offices} (e.g.\ ISTAT, BEA, INSEE);
    \item \textbf{Eurostat} and \textbf{OECD}, for cross-country validation;
    \item \textbf{Central banks} and \textbf{Supervisory authorities} (ECB, EBA, SEC).
\end{itemize}

\paragraph{Value of co-signature.}
A joint note ensures:
\begin{enumerate}
    \item \textbf{methodological credibility}, because validated by a neutral body;
    \item \textbf{adoption channel}, thanks to inclusion in manuals and official reporting;
    \item \textbf{standardization}, through shared check–lists and templates.
\end{enumerate}

\paragraph{Perspective.}
The proposal is to launch a joint “CBV Task Force” between academia and institutions, aiming to publish a first methodological note within two years.

\section{Technical appendix: existence conditions, non-linear extensions and Schur complement}
\label{app:tech}

This appendix provides (i) existence/uniqueness and stability conditions for Regime~B as a function of $O_{PP}$ (with sufficient conditions in terms of spectral radius and induced norms), (ii) a piecewise-linear formalization with \emph{seniority} priority and operational connection to \emph{clearing} engines à la Rogers--Veraart, and (iii) a block algebra note via the \emph{Schur complement} making the internal cancellation explicit.

\subsection{Conditions on \texorpdfstring{$O_{PP}$}{OPP}: existence, uniqueness, stability}
\label{app:condizioni}

Recall that in Regime~B the internal valuations are obtained as
\begin{equation}
\label{eq:regimeB-core}
\bm v_P \;=\; (I - O_{PP})^{-1}\,\big(\bm b_P + O_{PO}\,\bm v_O\big),
\end{equation}
when the inverse exists. The following conditions are standard and sufficient.

\begin{lemma}[Invertibility and Neumann series]
\label{lem:neumann}
If the spectral radius $\rho(O_{PP})<1$, then $I-O_{PP}$ is invertible and the expansion holds
\[
(I - O_{PP})^{-1}\;=\;\sum_{t=0}^{\infty} O_{PP}^t \quad\text{(convergence in norm).}
\]
\end{lemma}

\begin{proof}[Sketch]
$\rho(O_{PP})<1$ implies existence of an induced norm $\|\cdot\|$ such that $\|O_{PP}\|<1$; convergence of the geometric series in $\mathcal L(\R^{|P|})$ follows.
\end{proof}

\begin{proposition}[Practical sufficient conditions]
\label{prop:sufficienti}
All of the following are sufficient for $\rho(O_{PP})<1$:
\begin{enumerate}[label=(\roman*)]
\item $\|O_{PP}\|_{p\to p}<1$ for some induced norm ($p\in\{1,2,\infty\}$); in particular
\begin{align*}
\|O_{PP}\|_{1} &= \max_{j}\sum_i |(O_{PP})_{ij}|,\\
\|O_{PP}\|_{\infty} &= \max_{i}\sum_j |(O_{PP})_{ij}|,\\
\|O_{PP}\|_{2} &= \sigma_{\max}(O_{PP})
\end{align*}

\item Generalized diagonal dominance in the sense of Gershgorin: $\sum_{j\ne i}|(O_{PP})_{ij}|<1$ for each $i$.
\item (Non-negative case, \emph{Perron–Frobenius}) If $O_{PP}\ge 0$ and $\max_i\sum_j (O_{PP})_{ij} < 1$, then $\rho(O_{PP})\le \|O_{PP}\|_{\infty}<1$.
\end{enumerate}
In each case, $I-O_{PP}$ is invertible and \eqref{eq:regimeB-core} is well-posed.
\end{proposition}

\begin{proposition}[Useful bounds for stability]
\label{prop:stability-bounds}
If $\|O_{PP}\|_{p\to p}<1$, then
\[
\|(I-O_{PP})^{-1}\|_{p\to p} \;\le\; \frac{1}{\,1-\|O_{PP}\|_{p\to p}\,}.
\]
Moreover, for the border transfer operator $T_{PO}:=(I-O_{PP})^{-1}O_{PO}$ we have
\[
\|T_{PO}\|_{p\to p} \;\le\; \frac{\|O_{PO}\|_{p\to p}}{\,1-\|O_{PP}\|_{p\to p}\,}.
\]
\end{proposition}

\begin{corollary}[Stability of $W(P)$ in Regime~B]
\label{cor:stability-W}
Under the assumptions of Prop.~\ref{prop:sufficienti}, the robustness bounds of Sec.~\ref{subsec:robust-bounds} remain valid by replacing internal contributions with the bound of Prop.~\ref{prop:stability-bounds}. In particular, errors on $\bm v_O$ and $\bm b_P$ are amplified at most by $\|(I-O_{PP})^{-1}\|_{p\to p}$.
\end{corollary}

\paragraph{Pathological cases and regularization.}
If $\rho(O_{PP})\ge 1$ (``strong'' cycles or excessive leverage), Regime~B is not well-posed. Two strategies:  
(i) \emph{reparametrization} of the perimeter (moving problematic nodes into $O$ or consolidating them),  
(ii) \emph{damping} $O_{PP}^{(\alpha)}:=\alpha\,O_{PP}$ with $\alpha\in(0,1)$ and disclosure of the adjustment.

\paragraph{Limit conditions and near-singularity.}
The constraint $\rho(O_{PP})<1$ ensures convergence of the Neumann series and hence existence/uniqueness of the consolidated value. However, scenarios in which $\rho(O_{PP}) \uparrow 1$ or in which the ownership structure generates a nearly singular matrix pose practical problems.

\begin{itemize}
  \item \textbf{Numerical near-singularity.} When dominant eigenvalues of $O_{PP}$ approach $1$, the matrix $(I - O_{PP})$ becomes ill-conditioned: even small perturbations in accounting data produce large variations in the consolidated value. Numerically, the condition number $\kappa(I-O_{PP})$ grows rapidly, reducing computational reliability.
  \item \textbf{Stability.} In these regions, CBV estimates become highly sensitive to measurement errors or missing data. Consolidated values should be accompanied by \emph{uncertainty intervals}, obtained for instance from Monte Carlo perturbations on $O_{PP}$ coefficients.
  \item \textbf{Regularization.} To ensure robustness, additive regularizations can be introduced of the type
  \[
  (I - O_{PP})^{-1} \;\;\leadsto\;\; (I - O_{PP} + \varepsilon I)^{-1}, \quad \varepsilon > 0,
  \]
  which mitigate instability near the boundary $\rho(O_{PP})=1$. This is equivalent to introducing a \emph{structural tolerance} or a \emph{systematic haircut} on holdings.
  \item \textbf{Economic interpretation.} A $\rho(O_{PP})$ close to 1 reflects extremely intertwined networks, where each unit is almost entirely financed by cross-holdings. In such cases, the consolidated value is less ``informative'' relative to real exposures and a prudential approach becomes necessary.
\end{itemize}

\paragraph{Disclosure.}
For audit reasons, it is recommended that each Cut-Report include: \\(i) the estimated value of $\rho(O_{PP})$; \\(ii) the condition number of $(I - O_{PP})$; \\(iii) any regularizations adopted and the width of uncertainty intervals. This allows the reader to assess the reliability of the consolidated value near limit conditions.

\paragraph{Mini–numerical example ($2\times 2$): conditioning and uncertainty.}
Consider a sub-perimeter $P=\{p_1,p_2\}$ with symmetric cross-holdings of intensity $t\in[0,1)$:
\[
O_{PP}(t)=\begin{bmatrix}
0 & t\\
t & 0
\end{bmatrix},\qquad 
\rho\!\bigl(O_{PP}(t)\bigr)=t.
\]
Thus
\[
I-O_{PP}(t)=\begin{bmatrix}
1 & -t\\
-t & 1
\end{bmatrix},\quad 
(I-O_{PP}(t))^{-1}=\frac1{1-t^2}\begin{bmatrix}
1 & t\\
t & 1
\end{bmatrix}.
\]
Since $I-O_{PP}(t)$ is symmetric positive definite, the condition number in 2-norm is
\[
\kappa_2\bigl(I-O_{PP}(t)\bigr)=\frac{\lambda_{\max}}{\lambda_{\min}}=\frac{1+t}{\,1-t\,}.
\]
Let $\bm v_P=(100,\,50)^\top$ be the vector of observed \emph{primitives}. The resulting consolidated value is
\[
\bm x_P(t)=(I-O_{PP}(t))^{-1}\bm v_P
=\frac1{1-t^2}\begin{bmatrix}1 & t\\ t & 1\end{bmatrix}\!\!\begin{bmatrix}100\\ 50\end{bmatrix}
=\frac1{1-t^2}\begin{bmatrix}100+50t\\ 50+100t\end{bmatrix}.
\]
\paragraph{Two scenarios compared.}
We show sensitivity near the boundary $\rho(O_{PP})\uparrow 1$.
\begin{table}[H]
\centering
\begin{tabular}{lcccccc}
\toprule
Scenario & $t$ & $\rho(O_{PP})$ & $\kappa_2(I-O_{PP})$ & $x_{p_1}(t)$ & $x_{p_2}(t)$ & $x_{\text{tot}}(t)$ \\
\midrule
A (moderate) & $0.80$ & $0.80$ & $9.0$ & $388.8889$ & $361.1111$ & $750.0000$ \\
B (near–sing.) & $0.99$ & $0.99$ & $\approx 199.0$ & $7512.5628$ & $7487.4372$ & $15000.0000$ \\
\bottomrule
\end{tabular}
\caption{Conditioning and consolidated output as $t$ varies.}
\end{table}

\paragraph{Uncertainty intervals via small perturbation.}
If the data of $O_{PP}$ have uncertainty of amplitude $\pm 0.01$ on $t$, we obtain output bands:
\begin{itemize}
  \item \textbf{Around $t=0.80$:} with $t\in\{0.79,\,0.81\}$ we have
  \[
  x_{\text{tot}}(0.79)=714.2857,\qquad x_{\text{tot}}(0.81)=789.4737,
  \]
  i.e.\ a variation of about $\pm 5.1\%$ on the total for a relative error in $t$ of $\pm 1.25\%$ (amplification $\sim 4\times$), consistent with $\kappa_2=9$.
  \item \textbf{Around $t=0.99$:} to avoid exceeding the unit boundary we consider $t\in\{0.98,\,0.99\}$:
  \[
  x_{\text{tot}}(0.98)=7500.0000,\qquad x_{\text{tot}}(0.99)=15000.0000,
  \]
  that is, a \emph{halving/doubling} of the total for an absolute variation of $0.01$ in $t$ (relative error $\sim 1\%$), in line with $\kappa_2\simeq 199$.
\end{itemize}

\paragraph{Operational reading.}
As $t\to 1^-$ (i.e.\ $\rho(O_{PP})\uparrow 1$), the system becomes ill-conditioned: $x_{\text{tot}}(t)$ explodes and small uncertainties in the network coefficients translate into large output bands. In such cases the Cut-Report must:
(i) report $\rho(O_{PP})$ and $\kappa_2$; 
(ii) accompany the point estimate with an interval (e.g.\ via Monte Carlo perturbations on $t$); 
(iii) indicate any regularizations (e.g.\ $(I-O_{PP}+\varepsilon I)^{-1}$) applied to stabilize the estimate.

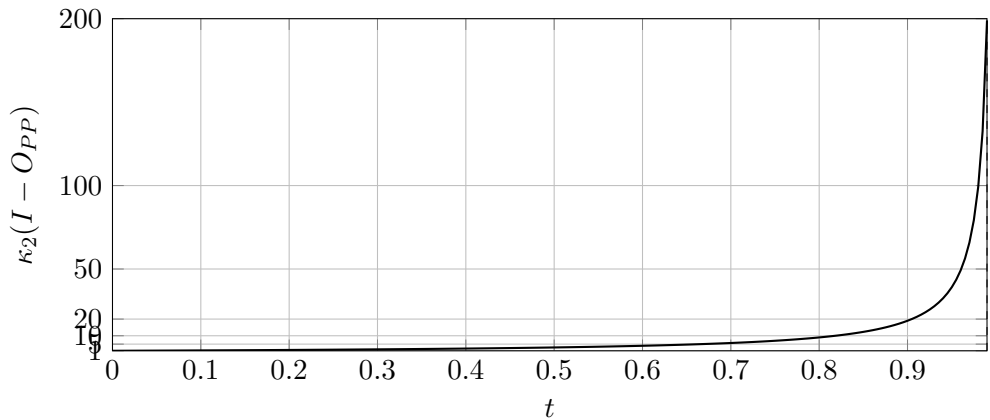
\begin{figure}[H]
\centering
\begin{tikzpicture}
  \begin{axis}[
    width=0.8\linewidth,
    height=6cm,
    xlabel={$t$},
    ylabel={$\kappa_2\!\left(I-O_{PP}\right)$},
    xmin=0, xmax=0.99,
    ymin=1, ymax=200,
    ytick={1,5,10,20,50,100,200},
    grid=both,
    domain=0:0.99,
    samples=200
  ]
    \addplot[thick] {(1+x)/(1-x)};

    \addplot[dashed] coordinates {(0.99,1) (0.99,200)};
    \node[anchor=west] at (axis cs:0.99,180) {asymptote $\,t\to1^{-}$};
  \end{axis}
\end{tikzpicture}
\caption{Explosion of $\kappa_2=\frac{1+t}{1-t}$ as $t\to1^{-}$. The growth illustrates the worsening of numerical conditioning near the spectral boundary.}
\end{figure}

\subsection{Piecewise-linear extensions with seniority priority (non-linear clearing)}
\label{app:seniority}

Many financial architectures require a \emph{clearing engine} that resolves payment priorities, recovery, and default costs. A standard formalism (Rogers--Veraart) extends Eisenberg--Noe with default costs. We adapt a minimal and operational scheme to connect it to CBV.

\begin{definition}[Clearing engine with seniority classes]
\label{def:clearing}
Let $L^{(\ell)}\in\R_+^{V\times V}$ be the matrix of gross debts of class $\ell=1,\dots,L$ (1 = highest priority). Let $\bar{\bm p}^{(\ell)}:=L^{(\ell)}\mathbf 1$ be the vector of gross dues per class. Let $\bm a\in\R_+^{V}$ be the vector of gross resources (cash+liquid assets) \emph{pre-clearing}. Let $\gamma^{(\ell)}\in[0,1]$ be the \emph{fractional default costs} per class (possibly class-dependent).

A clearing outcome is a vector of payments $\bm p^{(\ell)}\in[0,\bar{\bm p}^{(\ell)}]$ such that, for each node $i$ and class $\ell$,
\[
p_i^{(\ell)} \;=\; \min\!\Big\{\bar p_i^{(\ell)},\ a_i + \sum_{k=1}^{L}\sum_{j} \big(L^{(k)}_{ji}\, \theta^{(k)}_j\big) \;-\; \sum_{k=1}^{\ell}\gamma^{(k)}_i\big(\bar p_i^{(k)}-p_i^{(k)}\big)\Big\},
\]
where $\theta^{(k)}_j \in[0,1]$ is the \emph{payout ratio} of class $k$ of node $j$; higher classes are satisfied first: $\theta^{(1)}$ is determined first, then $\theta^{(2)}$, etc. The system defines a monotone map $\Phi$ whose fixed points are clearing outcomes.
\end{definition}

\begin{remark}[Existence/Uniqueness and implementation]
Under standard assumptions (non-negative debts, finite resources, costs in $[0,1]$), $\Phi$ is monotone on a complete lattice and admits fixed points (Tarski). In many cases, the outcome is unique; otherwise, one should choose (and disclose) the selection (least/greatest fixed point).
\end{remark}

\begin{proposition}[Conditional invariance of CBV post-clearing]
\label{prop:conditional-invariance}
Let $(\bm p^{(\ell)})_{\ell\le L}$ be a clearing outcome of Def.~\ref{def:clearing}. Define the \emph{post-clearing border net flows}
\[
X^{\text{net}}_{PO} \;:=\; \sum_{\ell=1}^{L}\big(L^{(\ell)}_{PO}\,\Theta^{(\ell)}\big), 
\qquad
X^{\text{net}}_{OP} \;:=\; \sum_{\ell=1}^{L}\big(L^{(\ell)}_{OP}\,\Theta^{(\ell)}\big),
\]
where $\Theta^{(\ell)}:=\operatorname{diag}(\theta^{(\ell)})$ is the diagonal matrix of payout ratios obtained at the fixed point. Then, given the Observer $\Omega$ (units/FX/PPP, SDF), the measure $W(P)$ computed via CBV on the \emph{post-clearing} border statistics $(X^{\text{net}}_{PO}, X^{\text{net}}_{OP}, \tilde{\bm v}_O, \bm b_P)$ is \emph{invariant} with respect to purely internal restructurings of $P$ that do not alter these statistics. 
\end{proposition}

\begin{proof}[Idea]
The non-linearity is entirely “resolved” in the clearing phase (determination of $\Theta^{(\ell)}$). CBV then applies downstream to additive border quantities (the actual outflows/inflows per class). Any \emph{internal} restructuring that leaves net flows $P\leftrightarrow O$ and observer objects unchanged does not modify $W(P)$ (same argument as the Cut Theorem).
\end{proof}

\paragraph{Operational instruction (pipeline).}
(1) Run the clearing engine; (2) compute $X^{\text{net}}_{PO}, X^{\text{net}}_{OP}$ and consistent $\tilde{\bm v}_O$ values; (3) apply CBV as in Sec.~\ref{sec:scope-limitations}. 

\subsection{Non-linear extensions: piecewise-linear payoff and endogenous risk}\label{subsec:nonlinear-risk}
\noindent
The cut framework is, by construction, linear. However, many financial instruments exhibit non-linear payoffs: implicit or explicit options, contractual covenants, bail-in clauses, default with partial recovery, and state-contingent securities. For these cases we propose an operational extension that maintains cut transparency by introducing \emph{piecewise-linear} (PWL) approximations. 

\paragraph{Formalism.}
Let $y=f(x)$ be a non-linear payoff with respect to an exposure $x$. The PWL rule consists in representing $f$ as a combination of linear segments over disjoint intervals:
\[
f(x) \approx \sum_{j=1}^J \alpha_j \bigl( (x - k_{j-1})^+ - (x - k_j)^+ \bigr),
\]
where $k_j$ are thresholds (cap/floor/covenant trigger) and $\alpha_j$ are local slope coefficients. This decomposition preserves compatibility with the block structure of the cut: each segment can be treated as a conditional linear position.

\paragraph{Typical cases.}
\begin{itemize}
    \item \textbf{Embedded options}: a convertible bond can be modeled as linear bond $+$ PWL call option.
    \item \textbf{Covenant trigger}: leverage or interest coverage ratio threshold implemented as $k_j$, with reduced payoff above threshold.
    \item \textbf{Costly default and bail-in}: recovery rate $R\in[0,1]$ modeled as factor $\alpha_j=R$ after the default trigger.
    \item \textbf{State-contingent claims}: payoff depending on a macro state (e.g.\ GDP), represented as conditional block on $\Omega$.
\end{itemize}

\paragraph{Numerical example.}
Consider a senior bond with nominal value 100 and recovery $R=40\%$. The payoff can be expressed as:
\[
f(x) = 
\begin{cases}
100 & \text{if default does not occur}, \\
40  & \text{if default occurs}.
\end{cases}
\]
Payoff table:
\begin{table}[H]
\centering
\begin{tabular}{c c c}
\toprule
State & Nominal value & Effective payoff \\
\midrule
No default & 100 & 100 \\
Default & 100 & 40 \\
\bottomrule
\end{tabular}
\caption{Mini-example of piecewise-linear payoff with recovery.}
\end{table}

\begin{figure}[H]
\centering
\begin{tikzpicture}[scale=1.0]
  \draw[->] (0,0) -- (6,0) node[right] {$x$ (nominal value or exposure)};
  \draw[->] (0,0) -- (0,5) node[above] {Payoff $f(x)$};

  \draw[thick,blue] (0,4) -- (3,4) node[midway, above] {No default};

  \draw[thick,red] (3,1.6) -- (6,1.6) node[midway, below] {Default, $R=40\%$};

  \draw[dashed,gray] (3,0) -- (3,4);

  \node at (3,-0.3) {Default trigger};
  \node[left] at (0,4) {100};
  \node[left] at (0,1.6) {40};
\end{tikzpicture}
\caption{Piecewise-linear payoff scheme: 100 in the absence of default, 40 in case of default with partial recovery.}
\end{figure}
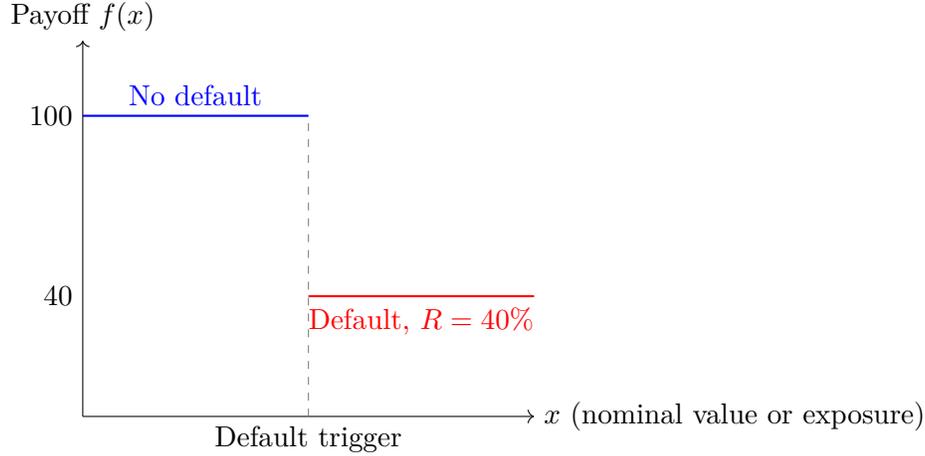

\paragraph{Compatibility with the cut.}
Each segment $[k_{j-1},k_j]$ can be treated as a linear exposure with observer $\Omega$, and the CBV evaluation reduces to evaluating the cut separately for each segment, then aggregating the payoffs. This preserves invariance with respect to the internal network and produces transparent disclosure (PoV and Cut-Report) even in the presence of non-linear payoffs.

\subsection{Non-linear extensions: state-conditioned linearity and coherent aggregation}
\label{subsec:nonlinear-scl}

In this subsection we extend the \emph{cut-based} (CBV) functional to contexts with non-linear contracts and constraints (derivatives, cap/floor, waterfall, guarantees/margin, contractual triggers), while preserving the core principles of the framework: (i) \emph{locality on the cut}, (ii) \emph{invariance under composition}, (iii) \emph{clear separation} between \emph{measurement} (mechanical, auditable) and \emph{observer preferences} (risk/value aggregation), (iv) \emph{audit-first}.

\paragraph{Guiding idea.}
We maintain the \emph{linearity} of CBV \emph{within states} (scenarios) selected to activate/deactivate non-linearities; the “non-linearity” is thus shifted to (i) the \emph{partition of states} and (ii) the \emph{aggregator} chosen by the observer $\Omega$. For curved payoffs (e.g.\ options) we adopt a \emph{piecewise-affine} (PWA) representation with controlled error. For min/max and waterfall we use an \emph{affine-lift} via epigraphs, obtaining an \emph{exact affine} representation in an extended space.

\subsubsection*{State space, conditional linearity and per-state functionals}

Let $(\mathcal{S},\mathfrak{S})$ be a discrete space of states/evidences relevant for contractual triggers (e.g., \emph{cap active or not}, \emph{default yes/no}, \emph{margin on/off}, \emph{price bucket}). In practice, $\mathcal{S}$ is a \emph{coarse economic partition} built on threshold events. Fix $s\in\mathcal{S}$, we denote
\[
\mathcal{F}_{\Omega,s}:\;\mathbb{R}^{m}\to\mathbb{R}
\]
the \emph{linear} CBV functional that evaluates flows/stocks on the cut in state $s$, for a fixed observer $\Omega$ (units/FX/PPP, SDF and perimeter as in \S\ref{sec:transform}, \S\ref{sec:osservatore}, \S\ref{sec:cutsection}). For each $s$, CBV axioms hold (linearity, additivity on the border, internal cancellation: cf.\ \S\ref{sec:cutsection} and related theorems).

\begin{definition}[State-conditioned linearity (SCL)]
\label{def:SCL}
We say that a system with potentially non-linear contracts satisfies \emph{SCL} if there exists a finite partition $\mathcal{S}$ such that, fixed $s\in\mathcal{S}$, the payoffs and clearing/allocation rules reduce to \emph{affine} relations in the cut flows:
\[
\text{payoff}(\cdot\,|\,s)=A_s\,x + b_s,
\]
with $A_s$ a matrix and $b_s$ a deterministic vector given $s$, and $x\in\mathbb{R}^m$ the vector of cut flows. In particular, \emph{within state} Regime~A/B (\S\ref{subsec:algA}, \S\ref{subsec:algB}) is applicable without modification.
\end{definition}

\begin{remark}[Conceptual separation]
\label{rmk:sep}
\emph{Measurement} in state $s$ is a \emph{local linear cut problem}. The global \emph{non-linearity} is delegated to (i) the choice of partition $\mathcal{S}$, (ii) the \emph{ex-ante aggregation} of per-state values, consistent with observer $\Omega$’s preferences.
\end{remark}

\subsubsection*{Observer’s coherent aggregators}

Let $V_s:=\mathcal{F}_{\Omega,s}(x)$ be the CBV value in state $s$. The observer $\Omega$ specifies an aggregator
\[
\rho_\Omega:\;\mathbb{R}^{|\mathcal{S}|}\to\mathbb{R},\qquad \mathbf{V}=(V_s)_{s\in\mathcal{S}}\mapsto \rho_\Omega(\mathbf{V}),
\]
which may embody: (i) a risk-neutral measure ($\mathbb{E}_Q$), (ii) a \emph{stochastic discount factor} (SDF) under physical measure $P$ (cf.\ \S\ref{sec:transform}), (iii) a \emph{coherent} risk measure (CVaR, Kusuoka, convex mix). We require $\rho_\Omega$ to satisfy classical properties (monotonicity, translation invariance, subadditivity, positive homogeneity, when applicable) to ensure stability and comparability.

\begin{theorem}[Cut invariance under coherent aggregation]
\label{thm:inv-agg}
If for every $s\in\mathcal{S}$ the functional $\mathcal{F}_{\Omega,s}$ satisfies the CBV axioms (linearity, additivity, internal cancellation) and $\rho_\Omega$ is monotone and translation-invariant, then the aggregate
\[
V_\Omega(x)\;=\;\rho_\Omega\Big(\{\mathcal{F}_{\Omega,s}(x)\}_{s\in\mathcal{S}}\Big)
\]
\emph{preserves} invariance under \emph{internal cancellation} and \emph{locality on the cut}. If moreover $\rho_\Omega$ is subadditive, then $V_\Omega$ is subadditive in $x$.
\end{theorem}

\begin{proof}[Proof idea]
Internal cancellation holds \emph{within state} for the functionals $\mathcal{F}_{\Omega,s}$. Monotonicity and translation invariance of $\rho_\Omega$ transfer these properties to the aggregate; subadditivity of $\rho_\Omega$ transfers subadditivity. Locality derives from per-state border additivity and the absence of extra-cut dependencies in the aggregator.
\end{proof}

\subsubsection*{Affine-lift via epigraphs for min/max, cap/floor and waterfall}

Many contractual non-linearities are \emph{exactly piecewise-linear}. It suffices to introduce auxiliary variables to represent \emph{min}/\emph{max} in affine form in an extended space: for $y=\max\{0,x\}$ it is equivalent to impose $y\ge 0$, $y\ge x$ and minimize $y$ (epigraph). A \emph{cap} $y=\min\{x,L\}$ is obtained as $y\le L$, $y\le x$, $y\ge 0$, and consistent max/min. Priority \emph{waterfalls} are implemented as sequences of linear constraints on excess allocation.

\begin{proposition}[Exact affine-lift]
\label{prop:affine-lift}
Contracts with payoffs $\,\min,\max$, cap/floor and waterfalls have an \emph{exact affine} representation via epigraphs in an extended space $(x,y,z,\dots)$, preserving cut locality. Within state $s$, evaluation remains a local linear problem for $\mathcal{F}_{\Omega,s}$.
\end{proposition}

\subsubsection*{Piecewise-affine (PWA) approximation for curved payoffs}

For smooth but curved payoffs (e.g.\ $(S_T-K)^+$, concave/convex functions) we use a PWA on a grid of breakpoints $\{b_\ell\}_{\ell=0}^L$. If $f$ is twice differentiable and $|f''|\le \Gamma_{\max}$ on intervals of maximum width $\Delta$, then the uniform approximation error satisfies
\begin{equation}
\label{eq:pwa-error}
\|f - f_{\mathrm{PWA}}\|_\infty \;\le\; \frac{\Gamma_{\max}}{8}\,\Delta^2.
\end{equation}
The PWA preserves the CBV architecture (local affine pieces); the choice of $\Delta$ makes the \emph{precision vs complexity} trade-off explicit.

\subsubsection*{Stability and error bounds of the aggregate}

\begin{proposition}[Stability with respect to data perturbations]
\label{prop:stability-agg}
Let $\rho_\Omega$ be 1-Lipschitz in $\ell_p$ norm (holds for $\mathbb{E}_Q$, CVaR with standard normalization). If for every $s$ it holds that $\big|\mathcal{F}_{\Omega,s}(x)-\mathcal{F}_{\Omega,s}(\tilde x)\big|\le L_s\,\|x-\tilde x\|$, then
\[
\big|V_\Omega(x)-V_\Omega(\tilde x)\big|
\;\le\;
\|\,(L_s)_s\,\|_{q}\cdot \|x-\tilde x\|,
\]
where $1/p+1/q=1$. In particular, if $L_s\le L$ for all $s$, it follows that
$|V_\Omega(x)-V_\Omega(\tilde x)|\le L\cdot |\mathcal{S}|^{1/q}\cdot \|x-\tilde x\|$.
\end{proposition}

\begin{proposition}[Total SCL+PWA error]
\label{prop:scl-pwa-error}
Let $V_\Omega^{\mathrm{true}}$ be the value with true payoff $f$, and $V_\Omega^{\mathrm{pwa}}$ that with $f_{\mathrm{PWA}}$. If $\rho_\Omega$ is 1-Lipschitz and \eqref{eq:pwa-error} holds for each piece, then
\[
\big|V_\Omega^{\mathrm{true}}-V_\Omega^{\mathrm{pwa}}\big|
\;\le\; \frac{\Gamma_{\max}}{8}\,\Delta^2.
\]
\end{proposition}

\subsubsection*{NLCBV-lite algorithm (operational and audit-first)}
\label{subsec:nlcbv-lite}

\begin{algorithm}[H]
\caption{NLCBV-lite: CBV extension for non-linearities}
\label{alg:nlcbv}
\begin{algorithmic}[1]
\Require Observer $\Omega$ (units/FX/PPP, SDF), perimeter, state grid $\mathcal{S}$, aggregation policy $\rho_\Omega$
\State \textbf{Partition} states $\mathcal{S}$ by economic triggers (default yes/no, cap active, margin on/off, price bucket).
\State \textbf{Affine-lift} for min/max, cap/floor, waterfall with epigraphs; \textbf{PWA} for curved payoffs with $\Delta$ chosen for desired $\varepsilon$.
\For{each $s\in\mathcal{S}$}
    \State \textbf{Evaluate} $\mathcal{F}_{\Omega,s}$ on the cut (Regime~A/B as in \S\ref{subsec:algA}–\S\ref{subsec:algB}).
\EndFor
\State \textbf{Aggregate} $V_\Omega(x)=\rho_\Omega\big(\{\mathcal{F}_{\Omega,s}(x)\}_s\big)$ (e.g.\ $\mathbb{E}_Q$, SDF under $P$, CVaR$_\alpha$).
\State \textbf{Package audit-pack}: (i) state definition, (ii) epigraphs/PWA spec and bound $\varepsilon$, (iii) per-state cut-matrix, (iv) $\rho_\Omega$ and parameters.
\end{algorithmic}
\end{algorithm}

\subsubsection*{Typical cases (compact cards)}

\paragraph{Call option on perimeter asset.}
Bucket $S_T$ into $K\pm n\sigma$; within state the payoff is affine. Alternatively, use PWA with few segments centered on $K$; bound as in \eqref{eq:pwa-error}. Aggregate with $\mathbb{E}_Q$ or prudent CVaR.

\paragraph{CDS with recovery and trigger.}
States: \emph{no default}, \emph{idiosyncratic default}, \emph{systemic default}; within state, premiums and protection are linear; recovery fixed per state. Aggregation via SDF or coherent measure.

\paragraph{Waterfall with cap/floor.}
Implement the priority order with sequential linear constraints; caps modeled with epigraphs. Per-state evaluation with $\mathcal{F}_{\Omega,s}$; aggregation according to $\rho_\Omega$.

\subsubsection*{Complexity, implementation and replicability}

\paragraph{Complexity.}
If $C_{\mathrm{CBV}}$ is the per-state evaluation cost (Regime~A/B) and $|\mathcal{S}|=M$, the complexity is $O(M\cdot C_{\mathrm{CBV}})$; $M$ remains small with coarse economic partitions. The affine-lift increases dimension in a controlled way (auxiliary variables for non-linear contracts), remaining within local linear problems.

\paragraph{Implementation.}
Pipelines from \S\ref{subsec:algA}–\S\ref{subsec:algB} are reused. An additional “\emph{state wrapper}” layer cycles over $s\in\mathcal{S}$, applies transformers from \S\ref{sec:transform} (units/FX/PPP, SDF), and composes the audit-pack. PWA and epigraphs are reusable modules per contract class.

\paragraph{Replicability.}
The per-state \emph{cut report} follows templates of \S\ref{sec:cutreport-standard}. Included: definition of states, epigraphic constraints spec, PWA grid and bound, choice of $\rho_\Omega$. This enables independent audit and cross-observer comparisons.

\subsubsection*{Policy notes and perimeter of validity}

\paragraph{Perimeter.}
SCL covers contracts with threshold triggers and piecewise-linear clearing. PWA covers smooth payoffs with explicit error control. Strongly path-dependent dynamics require multiple dates (t0–t1–T); the approach remains unchanged by enlarging the state dimension.

\paragraph{Policy.}
The separation between per-state measurement and ex-ante aggregation makes the approach compatible with different standards: the authority can prescribe a policy $\rho_\Omega$ (e.g.\ mean under $Q$, CVaR at level $\alpha$) while keeping the CBV cut machinery \emph{unchanged}.

\subsubsection*{Summary}
The SCL+$\rho_\Omega$ extension preserves CBV theorems on the cut \emph{within state} and shifts non-linearity to aggregation, ensuring \emph{invariance, parsimony and auditability}. The affine-lift provides exact representations for min/max and waterfalls; PWA introduces a controllable error $\varepsilon$ for curved payoffs. Together they realize a robust and operational \emph{NLCBV-lite}.

\subsection{Error bound for the piecewise-affine (PWA) approximation}
\label{app:pwa-bound}

In this subsection we formalize the assumptions for the error bound stated in \eqref{eq:pwa-error} and provide a brief proof, with some operational consequences useful for implementation and audit.

\paragraph{Setup.}
Let $[a,b]\subset\mathbb{R}$ be an interval and $a=x_0<x_1<\dots<x_N=b$ a partition (not necessarily uniform). Denote
\[
\Delta \;:=\; \max_{0\le i\le N-1} h_i,
\qquad h_i:=x_{i+1}-x_i.
\]
Given $f:[a,b]\to\mathbb{R}$, define $f_{\mathrm{PWA}}$ as the \emph{piecewise-linear interpolant} on the nodes $\{(x_i,f(x_i))\}_{i=0}^N$, i.e., on each $[x_i,x_{i+1}]$,
\[
f_{\mathrm{PWA}}(x)\;=\; \frac{x_{i+1}-x}{h_i}\,f(x_i)\;+\;\frac{x-x_i}{h_i}\,f(x_{i+1}).
\]

\begin{assumption}[Regularity and bounded curvature]\label{ass:reg}
Assume $f\in C^2([a,b])$ and that there exists a constant $\Gamma_{\max}\ge 0$ such that
\[
\big|f''(x)\big| \;\le\; \Gamma_{\max}\qquad \forall\,x\in[a,b].
\]
\end{assumption}

\begin{theorem}[$L^\infty$ bound for piecewise-linear interpolation]\label{thm:pwa}
Under Assumption~\ref{ass:reg}, for the interpolant $f_{\mathrm{PWA}}$ it holds that
\[
\big\| f - f_{\mathrm{PWA}} \big\|_{L^\infty([a,b])}
\;\le\; \frac{\Gamma_{\max}}{8}\,\Delta^2,
\]
in agreement with \eqref{eq:pwa-error}.
\end{theorem}

\begin{proof}[Proof (classical, for completeness)]
Fix an interval $[x_i,x_{i+1}]$ and set $h:=h_i$ and $\theta:=(x-x_i)/h\in[0,1]$. For the linear interpolation polynomial $p(x)$ that interpolates $f$ at the endpoints, the interpolation error formula (or equivalently Taylor’s theorem with remainder) gives
\[
f(x)-p(x) \;=\; \frac{f''(\xi)}{2}\,(x-x_i)\,(x_{i+1}-x)
\;=\; \frac{f''(\xi)}{2}\,h^2\,\theta(1-\theta)
\]
for some $\xi=\xi(x)\in(x_i,x_{i+1})$. Taking absolute values and using $|f''(\xi)|\le \Gamma_{\max}$,
\[
|f(x)-p(x)| \;\le\; \frac{\Gamma_{\max}}{2}\,h^2\,\theta(1-\theta).
\]
The function $\theta(1-\theta)$ is maximized at $\theta=\tfrac12$ and equals $1/4$. Hence
\[
\max_{x\in[x_i,x_{i+1}]} |f(x)-p(x)|
\;\le\; \frac{\Gamma_{\max}}{8}\,h^2.
\]
Since $f_{\mathrm{PWA}}=p$ on $[x_i,x_{i+1}]$, taking the maximum over all intervals and using $h\le \Delta$ we obtain
\[
\|f-f_{\mathrm{PWA}}\|_{L^\infty([a,b])}
\;\le\; \max_i \frac{\Gamma_{\max}}{8}\,h_i^2
\;\le\; \frac{\Gamma_{\max}}{8}\,\Delta^2.\qedhere
\]
\end{proof}

\begin{corollary}[Non-uniform grid]\label{cor:nonuniform}
For each interval $[x_i,x_{i+1}]$ it holds that 
\(
\|f-f_{\mathrm{PWA}}\|_{L^\infty([x_i,x_{i+1}])}
\le (\Gamma_{\max}/8)\,h_i^2.
\)
Consequently, the global bound is controlled by the maximum step $\Delta$.
\end{corollary}

\begin{corollary}[Sign of the error for convexity/concavity]\label{cor:convex}
If $f''\ge 0$ on $[a,b]$ (convex function), then $f_{\mathrm{PWA}}(x)\le f(x)$ for every $x\in(a,b)$, and
\(
0\le f(x)-f_{\mathrm{PWA}}(x)\le (\Gamma_{\max}/8)\,\Delta^2.
\)
If $f''\le 0$ (concave function), then $f_{\mathrm{PWA}}(x)\ge f(x)$ and
\(
0\le f_{\mathrm{PWA}}(x)-f(x)\le (\Gamma_{\max}/8)\,\Delta^2.
\)
\end{corollary}

\begin{remark}[$C^{1,1}$ case]\label{rmk:c11}
The bound remains valid if $f\in C^{1,1}([a,b])$, i.e.\ $f'$ is Lipschitz with constant $L$; indeed the second derivative exists a.e.\ and $|f''|\le L$ almost everywhere. In this case we can set $\Gamma_{\max}:=L$.
\end{remark}

\begin{remark}[Choice of PWA granularity]\label{rmk:delta}
To guarantee a uniform error $\varepsilon>0$, it suffices to choose
\(
\Delta \;\le\; \sqrt{\,8\,\varepsilon/\Gamma_{\max}\,}.
\)
In the presence of \emph{kinks} (e.g.\ payoff $(S_T-K)^+$) it is advisable to insert a node at the discontinuity of $f'$: on each side $f''=0$ and the local error vanishes.
\end{remark}

\paragraph{Propagation to the aggregator \texorpdfstring{$\rho_\Omega$}{rho\_Omega}.}
Let $\{V_s\}_{s\in\mathcal{S}}$ be the vector of state values obtained with exact payoffs, and $\{\tilde V_s\}_{s\in\mathcal{S}}$ those obtained with PWA payoffs on the same grid. If the aggregator $\rho_\Omega:\mathbb{R}^{|\mathcal{S}|}\to\mathbb{R}$ is $1$-Lipschitz in $\ell_\infty$ norm (holds for $\mathbb{E}_Q$, for CVaR with standard normalization, and for many coherent measures), then
\[
\big|\rho_\Omega\big((V_s)_s\big)-\rho_\Omega\big((\tilde V_s)_s\big)\big|
\;\le\; \max_{s\in\mathcal{S}} |V_s-\tilde V_s|
\;\le\; \frac{\Gamma_{\max}}{8}\,\Delta^2,
\]
where the last inequality follows from Theorem~\ref{thm:pwa} applied \emph{within state}.

\paragraph{Operational summary.}
The bound \(\|f-f_{\mathrm{PWA}}\|_\infty \le (\Gamma_{\max}/8)\Delta^2\) is:
(i) \emph{local} per interval, (ii) \emph{global} via $\Delta$, (iii) \emph{directed} by convexity/concavity (useful in prudential contexts), and (iv) \emph{stable} with respect to risk/value aggregation if $\rho_\Omega$ is $1$-Lipschitz. The choice of nodes should concentrate the grid around thresholds/triggers to maximize efficiency.

\subsection{Compact numerical examples (CBV, SCL, aggregation)}
\label{app:esempi-compatti}

In this section we present three minimal, quantitative, audit-first examples supporting the operational reading of the framework:
(i) a \emph{waterfall} with \emph{cap} (local CBV and affine-lift),
(ii) a \emph{CDS} with three states and prudential aggregation (SCL + CVaR),
(iii) a \emph{PPP/Fisher} comparison with invariance to the observer (transformation laws).

\subsubsection*{Example 1 — Waterfall with \texorpdfstring{\emph{cap}}{cap}: senior/junior payments}

Consider an inflow on the cut equal to $x\ge 0$; there are two creditors in \emph{waterfall} with priority: \emph{Senior} with \emph{cap} $L=100$, then \emph{Junior} on any excess. The \emph{affine-lift} representation via epigraphs (\S\ref{subsec:nonlinear-scl}) gives, \emph{within state}:
\[
\text{Senior}(x)=\min\{x,100\},\qquad
\text{Junior}(x)=\max\{0,\,x-100\}.
\]
Tab.~\ref{tab:ex1} illustrates three cases; Fig.~\ref{fig:ex1} visualizes the payout functions.

\begin{table}[H]
\centering
\caption{Example 1: waterfall payments with cap $L=100$.}
\label{tab:ex1}
\begin{tabular}{@{}rccc@{}}
\toprule
$x$ & Senior & Junior & Sum \\
\midrule
60  & 60  & 0  & 60 \\
90  & 90  & 0  & 90 \\
150 & 100 & 50 & 150 \\
\bottomrule
\end{tabular}
\end{table}

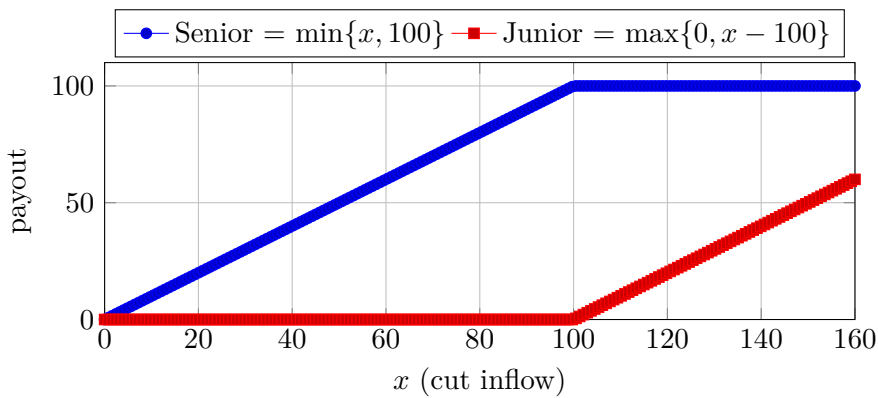
\begin{figure}[H]
\centering
\begin{tikzpicture}
\begin{axis}[
    width=0.7\linewidth,
    height=5cm,
    xlabel={$x$ (cut inflow)},
    ylabel={payout},
    ymin=0,
    xmin=0, xmax=160,
    legend style={at={(0.5,1.02)},anchor=south,legend columns=2},
    ymajorgrids=true, xmajorgrids=true
]
\addplot+[domain=0:160, samples=200] {min(x,100)}; \addlegendentry{Senior = $\min\{x,100\}$}
\addplot+[domain=0:160, samples=200] {max(0,x-100)}; \addlegendentry{Junior = $\max\{0,x-100\}$}
\end{axis}
\end{tikzpicture}
\caption{Example 1: senior/junior payments as functions of $x$.}
\label{fig:ex1}
\end{figure}

\paragraph{CBV note.}
Per-state evaluation is local linear on the cut; the \emph{cap} operator is implemented with epigraphs (exact affine representation), preserving auditability and internal cancellation.

\medskip

\subsubsection*{Example 2 — Three-state CDS (SCL) and prudential aggregation}

A perimeter $P$ \emph{purchases} protection on notional 100. Fixed premium (PV) equal to $1$. States:
\begin{itemize}[nosep,leftmargin=1.5em]
\item $s_1$: \emph{no default} (probability $0.95$) $\Rightarrow$ payoff $V_{s_1}=-1$;
\item $s_2$: \emph{idiosyncratic default} (prob.\ $0.04$), recovery $R=40\%$ $\Rightarrow$ $V_{s_2}=100(1-R)-1=59$;
\item $s_3$: \emph{systemic default} (prob.\ $0.01$), $R=10\%$ $\Rightarrow$ $V_{s_3}=100(1-R)-1=89$.
\end{itemize}
\[
\mathbb{E}_Q[V] \;=\; 0.95(-1)+0.04(59)+0.01(89)=2.30.
\]
With prudential aggregator $\rho_\Omega=\mathrm{CVaR}_{0.95}$ on the per-state CBV, the worst outcomes (5\% tail) coincide with $V=-1$; hence $\mathrm{CVaR}_{0.95}(V)=-1$.\footnote{With standard definitions of VaR/CVaR (Acerbi–Tasche), the 95\% tail contains the mass leading to $P(V\le \mathrm{VaR}_{0.95})=0.95$. Here the only value in the tail is $-1$.}
Tab.~\ref{tab:ex2} summarizes; Fig.~\ref{fig:ex2} visualizes the per-state values and aggregates.

\begin{table}[H]
\centering
\caption{Example 2: per-state CBV, $\mathbb{E}_Q$ and $\mathrm{CVaR}_{0.95}$.}
\label{tab:ex2}
\begin{tabular}{@{}lccc@{}}
\toprule
State & Prob. & $V_s$ (CBV) & Expected contribution \\
\midrule
$s_1$ (no default)        & $0.95$ & $-1$ & $-0.95$ \\
$s_2$ (idiosyncratic default) & $0.04$ & $59$ & $+2.36$ \\
$s_3$ (systemic default)  & $0.01$ & $89$ & $+0.89$ \\
\midrule
$\mathbb{E}_Q[V]$         &   &  & $2.30$ \\
$\mathrm{CVaR}_{0.95}(V)$ &   &  & $-1.00$ \\
\bottomrule
\end{tabular}
\end{table}

\begin{figure}[H]
\centering
\begin{tikzpicture}
\begin{axis}[
    width=0.7\linewidth,
    height=5cm,
    ybar,
    bar width=12pt,
    symbolic x coords={s1,s2,s3,E[·],CVaR95},
    xtick=data,
    ymin=-5,
    ylabel={Value},
    ymajorgrids=true, xmajorgrids=false
]
\addplot coordinates {(s1,-1) (s2,59) (s3,89) (E[·],2.30) (CVaR95,-1)};
\end{axis}
\end{tikzpicture}
\caption{Example 2: per-state values and aggregates ($\mathbb{E}_Q$, $\mathrm{CVaR}_{0.95}$).}
\label{fig:ex2}
\end{figure}
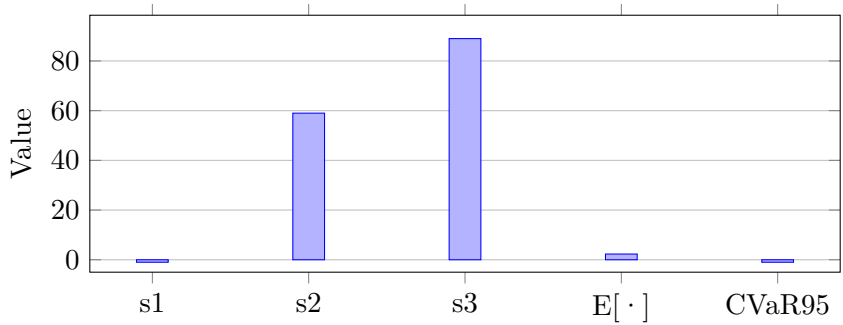

\paragraph{SCL/aggregation note.}
Measurement remains linear \emph{within state} (Regime~A); prudence is entirely incorporated in the aggregator $\rho_\Omega$.

\subsubsection*{Example 3 — PPP/Fisher and observer invariance}

Two goods, two countries. Prices: Italy (EUR) $p^{IT}=(10,\,5)$; USA (USD) $p^{US}=(8,\,6)$; nominal exchange rate $e=\frac{\text{USD}}{\text{EUR}}=1{,}2$. Quantities: $q^{IT}=(1,\,2)$, $q^{US}=(1{,}5,\,1{,}5)$. Converting US prices into EUR: $p^{US\to EUR}=p^{US}/e=(6{,}667,\,5)$.
\begin{align*}
L^{EUR} &= \frac{p^{US\to EUR}\cdot q^{IT}}{p^{IT}\cdot q^{IT}}
= \frac{6{,}667+10}{10+10} \approx 0{,}833,\\
P^{EUR} &= \frac{p^{US\to EUR}\cdot q^{US}}{p^{IT}\cdot q^{US}}
= \frac{10{,}000+7{,}500}{15{,}000+7{,}500} \approx 0{,}778,\\
F^{EUR} &= \sqrt{L^{EUR}P^{EUR}}\approx 0{,}805.
\end{align*}
If we evaluate in USD (alternative observer), $p^{IT\to USD}=e\,p^{IT}=(12,\,6)$ and $p^{US}=(8,\,6)$; we obtain:
\[
L^{USD}=\frac{20}{24}=0{,}833,\quad
P^{USD}=\frac{21}{27}=0{,}778,\quad
F^{USD}=\sqrt{0{,}833\cdot 0{,}778}\approx 0{,}805,
\]
\emph{invariant} under transformation (cf.\ \S\ref{sec:transform}).

\begin{table}[H]
\centering
\caption{Example 3: bilateral indices (EUR vs USD) — observer invariance.}
\label{tab:ex3}
\begin{tabular}{@{}lccc@{}}
\toprule
Observer & $L$ & $P$ & $F$ \\
\midrule
EUR & $0{,}833$ & $0{,}778$ & $0{,}805$ \\
USD & $0{,}833$ & $0{,}778$ & $0{,}805$ \\
\bottomrule
\end{tabular}
\end{table}

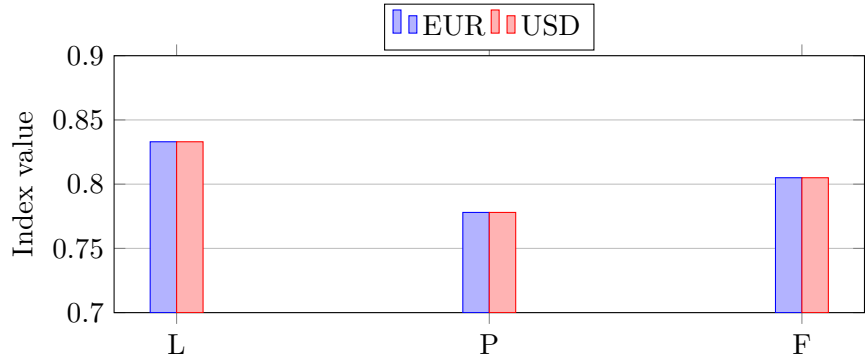
\begin{figure}[H]
\centering
\begin{tikzpicture}
\begin{axis}[
    width=0.7\linewidth,
    height=5cm,
    ybar=0pt,
    bar width=10pt,
    symbolic x coords={L,P,F},
    xtick=data,
    ymin=0.7, ymax=0.9,
    ylabel={Index value},
    legend style={at={(0.5,1.02)},anchor=south,legend columns=2},
    ymajorgrids=true
]
\addplot coordinates {(L,0.833) (P,0.778) (F,0.805)};
\addlegendentry{EUR}
\addplot coordinates {(L,0.833) (P,0.778) (F,0.805)};
\addlegendentry{USD}
\end{axis}
\end{tikzpicture}
\caption{Example 3: indices $L,P,F$ in EUR and USD (coincide under transformation).}
\label{fig:ex3}
\end{figure}

\paragraph{Observer/CBV note.}
The indices are \emph{observer-consistent artifacts}: changing units/FX produces transformations that preserve Fisher ratios (cf.\ transformation laws).

\subsubsection*{Aggregation policy \texorpdfstring{$\rho_\Omega$}{rho\_Omega}: operational guide}

\begin{table}[H]
\centering
\caption{Policies for aggregator $\rho_\Omega$: definition, required data, properties, pros/cons, use cases.}
\label{tab:rho-policy}
\begin{tabularx}{\linewidth}{@{}X X X X X X@{}}
\toprule
\textbf{Policy} & \textbf{Brief definition} & \textbf{Required data} & \textbf{Properties} & \textbf{Pros} & \textbf{Cons / Notes} \\
\midrule
Neutral (expectation under $Q$) &
$\rho_\Omega(\mathbf V)=\mathbb{E}_Q[V]$ &
Risk–neutral measure $Q$; state weights &
Linear, monotone, translation–invariant &
Comparable with no-arbitrage prices; simple &
Not prudential; requires calibration of $Q$ \\
\addlinespace
SDF on physical measure &
$\rho_\Omega(\mathbf V)=\mathbb{E}_P[M\,V]$ &
Physical measure $P$; coherent SDF $M$ &
Monotone; equivalent to SDF pricing &
Connects to macro/realistic scenarios; flexible &
Choice of $M$ may be controversial \\
\addlinespace
CVaR$_\alpha$ (coherent) &
$\rho_\Omega(\mathbf V)=\mathrm{CVaR}_\alpha(V)$ &
Level $\alpha$; distribution of $V$ &
Monotone, subadditive, coherent &
Prudential; tail-focused &
Sensitive to tail estimation \\
\addlinespace
Kusuoka / Coherent mix &
$\rho_\Omega(\mathbf V)=\int_0^1 \mathrm{CVaR}_u(V)\,d\mu(u)$ &
Measure $\mu$ on $[0,1]$ &
Convex/coherent; general representation &
“Menu-type” rule between prudence and mean &
Requires specifying $\mu$; more parameters \\
\addlinespace
Worst–case on $S^\star$ &
$\rho_\Omega(\mathbf V)=\min_{s\in S^\star} V_s$ &
Plausible subset $S^\star$ &
Monotone; locally robust &
Maximum caution on “critical” states &
Often too conservative \\
\bottomrule
\end{tabularx}
\end{table}

\paragraph{Practical guidelines.}
\emph{Regulatory reporting}: CVaR$_\alpha$ or Kusuoka mix; \emph{market valuations}: expectation under $Q$ or SDF on $P$; \emph{stress testing}: worst–case on $S^\star$. It is recommended to document (audit-pack) the policy, parameters ($\alpha$, $\mu$, $M$) and the state weights.

\subsubsection*{How to choose states (SCL): operational heuristics}

\begin{figure}[H]
\centering
\fbox{%
\begin{minipage}{0.95\linewidth}
\textbf{Objective.} Partition the outcome space into \emph{few} states $\,\mathcal S=\{s_i\}$ such that, \emph{within state}, payoffs and rules are (almost) affine: CBV remains linear and local on the cut.

\medskip
\textbf{Heuristics.}
\begin{itemize}
  \item \emph{Threshold triggers}: default (yes/no), margin (on/off), cap/floor (active/inactive), waterfall (order unchanged). Insert states at activation points.
  \item \emph{Known kinks}: for options, include a node at strike $K$; for piecewise payoffs, include nodes at slope changes.
  \item \emph{Coarse granularity}: start with $3$–$7$ relevant states total; increase only if the ex-post residual is excessive.
  \item \emph{Sensitivity}: recalculate with a refined partition (e.g.\ doubling states on triggers) and compare aggregate $\rho_\Omega$; if the difference $<\tau$ (tolerance), keep the more parsimonious partition.
  \item \emph{Composition}: independent states for different contracts remain compatible; avoid combinatorial explosion by merging triggers that rarely co-activate.
\end{itemize}

\textbf{Acceptance criterion.} Choose the simplest partition that respects: (i) aggregate error $<\tau$; (ii) qualitative stability of per-state contributions; (iii) interpretability (audit).
\end{minipage}}
\caption{Box — Choice of SCL states: principles and robustness checks.}
\label{box:scl-states}
\end{figure}

\subsubsection*{Impact of granularity \texorpdfstring{$\Delta$}{Delta} in the PWA approximation}

The bound \(\|f-f_{\mathrm{PWA}}\|_\infty \le (\Gamma_{\max}/8)\Delta^2\) implies
\[
\Delta_{\max}(\varepsilon, \Gamma_{\max}) \;=\; \sqrt{\frac{8\,\varepsilon}{\Gamma_{\max}}}\,,
\qquad
N(L)\;\approx\;\Big\lceil \frac{L}{\Delta_{\max}}\Big\rceil.
\]
Tab.~\ref{tab:pwa-delta} reports \(\Delta_{\max}\) and the number of segments \(N\) for a normalized interval \(L=1\).
For intervals of length \(L\neq 1\), multiply \(N\) by \(L\).

\begin{table}[H]
\centering
\caption{Maximum granularity $\Delta_{\max}$ for target uniform error $\varepsilon$ and curvature $\Gamma_{\max}$; $N=\lceil 1/\Delta_{\max}\rceil$ on $L=1$.}
\label{tab:pwa-delta}
\begin{tabular}{@{}lcccc@{}}
\toprule
\multirow{2}{*}{$\varepsilon$} & \multicolumn{1}{c}{$\Gamma_{\max}=0.5$} & \multicolumn{1}{c}{$\Gamma_{\max}=1$} & \multicolumn{1}{c}{$\Gamma_{\max}=2$} & \multicolumn{1}{c}{$\Gamma_{\max}=5$} \\
 & $\Delta_{\max}$ \quad $(N)$ & $\Delta_{\max}$ \quad $(N)$ & $\Delta_{\max}$ \quad $(N)$ & $\Delta_{\max}$ \quad $(N)$ \\
\midrule
$0.05$  & $0.8944$ \ $(2)$ & $0.6325$ \ $(2)$ & $0.4472$ \ $(3)$ & $0.2828$ \ $(4)$ \\
$0.02$  & $0.5657$ \ $(2)$ & $0.4000$ \ $(3)$ & $0.2828$ \ $(4)$ & $0.1789$ \ $(6)$ \\
$0.01$  & $0.4000$ \ $(3)$ & $0.2828$ \ $(4)$ & $0.2000$ \ $(5)$ & $0.1265$ \ $(8)$ \\
$0.005$ & $0.2828$ \ $(4)$ & $0.2000$ \ $(5)$ & $0.1414$ \ $(8)$ & $0.0894$ \ $(12)$ \\
$0.001$ & $0.1265$ \ $(8)$ & $0.0894$ \ $(12)$ & $0.0632$ \ $(16)$ & $0.0400$ \ $(25)$ \\
\bottomrule
\end{tabular}
\end{table}

\paragraph{Practical guidelines.}
\begin{itemize}
  \item \emph{Prudential target}: choose $\varepsilon$ based on materiality (\% of value) and audit policy; derive $\Delta_{\max}$ from the formula; set $N\simeq \lceil L/\Delta_{\max}\rceil$.
  \item \emph{Adaptive grid}: concentrate nodes where $|f''|$ is higher and near kinks/triggers; reduce nodes where $|f''|$ is minimal.
  \item \emph{Stability check}: halve $\Delta$ in critical areas and verify that the change in $\rho_\Omega$ remains $<\tau$; otherwise, locally increase resolution.
\end{itemize}

\subsection{Schur complement and internal cancellation}
\label{app:schur}

Consider the global linear system $\,\bm v=\bm b + O\,\bm v\,$, i.e.\ $(I-O)\bm v=\bm b$. With partition $V=P\cup O$ and blocks
\[
I-O\;=\;
\begin{bmatrix}
I - O_{PP} & -\,O_{PO}\\[2pt]
-\,O_{OP} & I - O_{OO}
\end{bmatrix}, 
\qquad 
\bm v \!=\! \begin{bmatrix}\bm v_P\\ \bm v_O\end{bmatrix},\quad
\bm b \!=\! \begin{bmatrix}\bm b_P\\ \bm b_O\end{bmatrix},
\]
eliminating $\bm v_P$ (when $I-O_{PP}$ is invertible) yields the \emph{Schur complement} with respect to $P$:
\[
S_{OO} \;:=\; I - O_{OO} \;-\; O_{OP}(I - O_{PP})^{-1}O_{PO}.
\]
The block $S_{OO}$ is the \emph{effective} operator on $O$ after “absorbing” the internal network of $P$.

\begin{lemma}[Elimination of $P$ via Schur]
\label{lem:schur-elim}
If $I-O_{PP}$ is invertible, then the solutions of the global system satisfy
\[
\bm v_O \;=\; S_{OO}^{-1}\,\big(\bm b_O + O_{OP}(I-O_{PP})^{-1}\bm b_P\big), 
\qquad
\bm v_P \;=\; (I-O_{PP})^{-1}\big(\bm b_P + O_{PO}\,\bm v_O\big).
\]
\end{lemma}

\begin{proposition}[Internal cancellation as Schur elimination]
\label{prop:schur-cancel}
The cut measure
\(
W(P)=\mathbf 1_P^\top \bm b_P + \mathbf 1_P^\top O_{PO}\bm v_O - \mathbf 1_O^\top O_{OP}\bm v_P
\)
is \emph{invariant} with respect to any transformation of the internal network of $P$ that leaves invariant the \emph{effective} operators $T_{PO}=(I-O_{PP})^{-1}O_{PO}$ and $U_{OP}=O_{OP}(I-O_{PP})^{-1}$ (and, in Regime~A, that leaves $\bm v_P$ unchanged).
\end{proposition}

\begin{proof}[Idea]
By substituting $\bm v_P$ from Lemma~\ref{lem:schur-elim} into $W(P)$ one obtains an expression depending on $T_{PO}$ and $U_{OP}$. Any internal restructuring that preserves these operators (i.e.\ the \emph{effect} of $P$ as seen from the frontier) leaves $W(P)$ unchanged. This is the “block” version of the Cut Theorem.
\end{proof}

\paragraph{Observation (Regime A vs Regime B).}
In Regime~A invariance is immediate because $W(P)$ depends only on observed frontier quantities. In Regime~B, the dependence on $O_{PP}$ enters only through $T_{PO}$ and $U_{OP}$: internal cancellation is equivalent to Schur elimination of $P$.

\section{Formal proof of the Cut Theorem}\label{app:formal-cut-proof}

In this section we provide a rigorous proof of Theorem~\ref{thm:cut}.
Let $V$ be the finite set of nodes, $P\subseteq V$ a perimeter and $O=V\setminus P$ its complement.
We denote by $\mathbf{1}_P\in\R^{|P|}$ and $\mathbf{1}_O\in\R^{|O|}$ the all–ones vectors.
We write the variables in block form:
\[
\bm v=\begin{pmatrix}\bm v_P\\ \bm v_O\end{pmatrix},\qquad
\bm b=\begin{pmatrix}\bm b_P\\ \bm b_O\end{pmatrix},\qquad
O=\begin{pmatrix} O_{PP} & O_{PO}\\ O_{OP} & O_{OO}\end{pmatrix}.
\]
We assume the standard structural identity $\bm v=\bm b+O\,\bm v$ (value as base plus participations).

\begin{lemma}[Aggregated internal value]\label{lem:vp-aggregate}
It holds that $\mathbf{1}_P^\top \bm v_P=\mathbf{1}_P^\top \bm b_P+\mathbf{1}_P^\top O_{PP}\bm v_P+\mathbf{1}_P^\top O_{PO}\bm v_O$.
\end{lemma}
\begin{proof}
This is the $P$ component of the relation $\bm v=\bm b+O\,\bm v$, premultiplied by $\mathbf{1}_P^\top$.
\end{proof}

\begin{lemma}[Elimination of internal items]\label{lem:elim-internal}
The consolidated value of the perimeter, after eliminating internal participations, is
\[
W(P)=\mathbf{1}_P^\top \bm v_P-\mathbf{1}_P^\top O_{PP}\bm v_P-\mathbf{1}_O^\top O_{OP}\bm v_P.
\]
\end{lemma}
\begin{proof}
In consolidation, shares of $P$ in $P$ (\emph{intercompany}) cannot appear as assets; hence we subtract $\mathbf{1}_P^\top O_{PP}\bm v_P$.
Symmetrically, shares of $O$ in $P$ represent \emph{external minorities} and are liabilities of the perimeter: we subtract $\mathbf{1}_O^\top O_{OP}\bm v_P$.
\end{proof}

\begin{proof}[Proof of Theorem~\ref{thm:cut}]
Substituting Lemma~\ref{lem:vp-aggregate} into Lemma~\ref{lem:elim-internal}:
\[
\begin{aligned}
W(P)
&= \big(\mathbf{1}_P^\top \bm b_P+\mathbf{1}_P^\top O_{PP}\bm v_P+\mathbf{1}_P^\top O_{PO}\bm v_O\big)
   -\mathbf{1}_P^\top O_{PP}\bm v_P-\mathbf{1}_O^\top O_{OP}\bm v_P \\
&= \mathbf{1}_P^\top \bm b_P + \mathbf{1}_P^\top O_{PO}\bm v_O - \mathbf{1}_O^\top O_{OP}\bm v_P.
\end{aligned}
\]
Expanding as an indexed sum,
\[
W(P)=\sum_{j\in P} b_j \;+\; \sum_{i\in P,k\in O} O_{ik}\,v_k \;-\; \sum_{i\in O,j\in P} O_{ij}\,v_j,
\]
which is exactly \eqref{eq:cut}. The formula does not contain $O_{PP}$: the consolidated value depends only on internal bases and cut edges $P\!\leftrightarrow\!O$, i.e.\ it is \emph{cut-based}.
\end{proof}

\paragraph{Observation (invariance with respect to $O_{PP}$).}
The term $\mathbf{1}_P^\top O_{PP}\bm v_P$ cancels algebraically in the sum, proving that the internal topology ($O_{PP}$) is irrelevant for $W(P)$ given the frontier and bases: any internal restructuring of chains/cycles/holdings leaves $W(P)$ unchanged.

\section{Extended proof of uniqueness (axiomatization)}\label{app:cbv-uniqueness}
\paragraph{Lemma 1 (Reduction to the cut).}
If two networks share $(\bm b_P,\bm v_P,\bm v_O,O_{PO},O_{OP})$ but differ in $O_{PP}$, then $W(P)$ coincides.
\emph{Proof.} (B)--(C).

\paragraph{Lemma 2 (Frontier linearity).}
Under (A) and Lemma~1, $W(P)$ is a linear combination of 
$\sum_{j\in P} b_j$, $\sum_{i\in P,k\in O} O_{ik} v_k$ and $\sum_{i\in O,j\in P} O_{ij} v_j$.

\paragraph{Lemma 3 (Accounting normalizations).}
Three base cases fix the coefficients at $+1,+1,-1$: (i) $O\equiv 0$; (ii) only assets $P\!\to\!O$; (iii) only minorities $O\!\to\!P$.

\paragraph{Lemma 4 (Modularity).}
The resulting form satisfies (D) by additive decomposition and correction on the internal cut.

\paragraph{Lemma 5 (Equivariance).}
Different coefficients would break (E); the base cases fix the unique normalization.

\medskip
\noindent\emph{Closure.} Lemmas 1–5 imply Theorem~\ref{thm:cbv-uniqueness} and uniqueness. \hfill\qed

\section{Derivation of robustness bounds}\label{app:cbv-robustness}
\paragraph{Preliminaries.}
For vectors we use $p$-norms and dual $q$ with Hölder; for matrices $\|A\|_{q\leftarrow p}:=\sup_{x\neq0}\|Ax\|_q/\|x\|_p$.

\paragraph{Proof of \eqref{eq:robust-generic}.}
From the definition of $W(P)$ and with $\Delta \bm v_P=0$:
$\Delta W=\mathbf{1}_P^\top \Delta \bm b_P + \mathbf{1}_P^\top O_{PO}\,\Delta \bm v_O$.
Applying Hölder and operator sub-multiplicativity yields \eqref{eq:robust-generic}.

\paragraph{Notable cases.}
The three corollaries in the main text follow by choosing $(p,q)=(1,\infty)$, $(2,2)$, $(\infty,1)$ and estimating $\|\mathbf{1}_P\|_q$ and $\|\mathbf{1}_P^\top O_{PO}\|_q$.

\paragraph{Extension with $(I-O_{PP})^{-1}$.}
If $\rho(O_{PP})<1$, then 
$\Delta \bm v_P=(I-O_{PP})^{-1}(\Delta \bm b_P+O_{PO}\Delta \bm v_O)$
and the additional term is $-\bm\delta^\top \Delta \bm v_P$ with $\bm\delta=O_{OP}^\top \mathbf{1}_O$.
Applying again Hölder and sub-multiplicativity yields the extended bound reported.
\hfill\qed

\section{An operational metaphor: Lorentz-style scaling (non-normative)}\label{app:lorentz}

\noindent \textbf{Status.} This appendix provides a \emph{metaphor} to build intuition. It is not part of the normative framework and must not be used for measurement or disclosure in place of the observer tuple $\Omega$ and the Cut Theorem.

\paragraph{Metaphor.} Early drafts sketched a Lorentz-style map
\begin{equation}\label{eq:lorentz-metaphor}
V_{\mathrm{met}} \;=\; \frac{V_0}{\sqrt{1 - \left(\frac{\Delta R}{K}\right)^2}}\ ,
\end{equation}
where $\Delta R$ roughly captured differences in reporting protocols (perimeter, base, unit/PPP/SDF), and $K$ a scale limit. As a metaphor, \eqref{eq:lorentz-metaphor} says that large protocol differences can amplify discrepancies in reported values.

\paragraph{Why non-normative.} Without an axiomatic derivation and an empirically validated mapping from protocol differences to $\Delta R$ and $K$, \eqref{eq:lorentz-metaphor} is not a measurement rule. The normative framework in the paper replaces it with explicit \emph{transformation laws} (Section~\ref{sec:transform}) and \emph{disclosure artifacts} (Section~\ref{sec:standards}).

\paragraph{Safe use.} The metaphor can be used in pedagogical contexts to convey that \emph{values depend on the observer} and that protocol differences may have non-linear effects. For any calculation or audit, use the observer tuple $\Omega$, the Cut Theorem, and the algorithms of Sections~\ref{sec:algorithms} and following.

\newpage
\addcontentsline{toc}{section}{Bibliography}\label{sec:Bibliografia}

\

\begin{thebibliography}{99}


\bibitem{MillerBlair2009}
Miller, R.~E., \& Blair, P.~D. (2009).
\newblock \emph{Input--Output Analysis: Foundations and Extensions} (2nd ed.).
\newblock Cambridge: Cambridge University Press. \doi{10.1017/CBO9780511626982}

\bibitem{ESA2010}
European Commission, International Monetary Fund, Organisation for Economic Co-operation and Development, United Nations, \& World Bank (2013).
\newblock \emph{European System of Accounts --- ESA 2010}.
\newblock Luxembourg: Publications Office of the European Union.

\bibitem{LaPorta1999}
La Porta, R., Lopez-de-Silanes, F., \& Shleifer, A. (1999).
\newblock Corporate Ownership Around the World.
\newblock \emph{The Journal of Finance}, 54(2), 471--517. \doi{10.1111/0022-1082.00115}

\bibitem{AlmeidaWolfenzon2006}
Almeida, H., \& Wolfenzon, D. (2006).
\newblock A Theory of Pyramidal Ownership and Family Business Groups.
\newblock \emph{The Journal of Finance}, 61(6), 2637--2680. \doi{10.1111/j.1540-6261.2006.01001.x}

\bibitem{Bebchuk2000}
Bebchuk, L.~A., Kraakman, R., \& Triantis, G. (2000).
\newblock Stock Pyramids, Cross-Ownership, and Dual Class Equity: The Mechanisms and Agency Costs of Separating Control from Cash-Flow Rights.
\newblock In R.~K. Morck (Ed.), \emph{Concentrated Corporate Ownership} (pp. 295--318). Chicago, IL: University of Chicago Press. \doi{10.7208/9780226536828-014}

\bibitem{Acemoglu2012}
Acemoglu, D., Carvalho, V.~M., Ozdaglar, A., \& Tahbaz-Salehi, A. (2012).
\newblock The Network Origins of Aggregate Fluctuations.
\newblock \emph{Econometrica}, 80(5), 1977--2016. \doi{10.3982/ECTA9623}

\bibitem{Carvalho2014}
Carvalho, V.~M. (2014).
\newblock From Micro to Macro via Production Networks.
\newblock \emph{Journal of Economic Perspectives}, 28(4), 23--48. \doi{10.1257/jep.28.4.23}

\bibitem{Jackson2010}
Jackson, M.~O. (2010).
\newblock \emph{Social and Economic Networks}.
\newblock Princeton, NJ: Princeton University Press.

\bibitem{Gabaix2011}
Gabaix, X. (2011).
\newblock The Granular Origins of Aggregate Fluctuations.
\newblock \emph{Econometrica}, 79(3), 733--772. \doi{10.3982/ECTA8769}

\bibitem{IFRS12}
International Accounting Standards Board (IASB) (2011).
\newblock \emph{IFRS 12: Disclosure of Interests in Other Entities}.
\newblock London: IFRS Foundation.

\bibitem{IAS28}
International Accounting Standards Board (IASB) (2011).
\newblock \emph{IAS 28: Investments in Associates and Joint Ventures}.
\newblock London: IFRS Foundation.

\bibitem{UN_SUTIOT_2018}
United Nations (2018).
\newblock \emph{Handbook on Supply, Use and Input--Output Tables with Extensions and Applications}.
\newblock New York: United Nations.

\bibitem{OECD_2025_ESUT}
Organisation for Economic Co-operation and Development (OECD) (2025).
\newblock \emph{Handbook on Extended Supply and Use Tables and Extended Input--Output Tables}.
\newblock Paris: OECD Publishing.

\bibitem{Bonacich1987}
Bonacich, P. (1987).
\newblock Power and Centrality: A Family of Measures.
\newblock \emph{American Journal of Sociology}, 92(5), 1170--1182. \doi{10.1086/228631}

\bibitem{IFRS_conceptual_framework}
IFRS Foundation (2018).
\newblock \emph{Conceptual Framework for Financial Reporting} (revised 2018).
\newblock London: IFRS Foundation.

\bibitem{investopedia2025}
Investopedia (2025).
\newblock PCE Price Index vs.\ Consumer Price Index (CPI): What's the Difference?
\newblock \emph{Investopedia}. Consultato il 25 agosto 2025.

\bibitem{Leontief1941}
Leontief, W. (1941).
\newblock \emph{The Structure of the American Economy, 1919--1929}.
\newblock Cambridge, MA: Harvard University Press.

\bibitem{Diewert1976}
Diewert, W.~E. (1976).
\newblock Exact and Superlative Index Numbers.
\newblock \emph{Journal of Econometrics}, 4, 115--145. \doi{10.1016/0304-4076(76)90009-9}

\bibitem{FaccioLang2002}
Faccio, M., \& Lang, L.~H.~P. (2002).
\newblock The Ultimate Ownership of Western European Corporations.
\newblock \emph{Journal of Financial Economics}, 65(3), 365--395. \doi{10.1016/S0304-405X(02)00146-0}

\bibitem{HansenSingleton1983}
Hansen, L.~P., \& Singleton, K.~J. (1983).
\newblock Stochastic Consumption, Risk Aversion, and the Temporal Behavior of Asset Returns.
\newblock \emph{Journal of Political Economy}, 91(2), 249--265. \doi{10.1086/261141}

\bibitem{Cochrane2005}
Cochrane, J.~H. (2005).
\newblock \emph{Asset Pricing} (Revised ed.).
\newblock Princeton, NJ: Princeton University Press.

\bibitem{Vitali2011_network_control}
Vitali, S., Glattfelder, J.~B., \& Battiston, S. (2011).
\newblock The Network of Global Corporate Control.
\newblock \emph{PLOS ONE}, 6(10), e25995. \doi{10.1371/journal.pone.0025995}

\bibitem{WorldBankICP}
World Bank (2024).
\newblock International Comparison Program (ICP): Methods and Data (PPP/PLI).
\newblock Washington, DC: World Bank.

\bibitem{IMF_WEO_2024}
International Monetary Fund (2024).
\newblock \emph{World Economic Outlook, October 2024: Policy Pivot, Rising Threats}.
\newblock Washington, DC: IMF.

\bibitem{BLS_PCE_CPI_Methods}
Bureau of Labor Statistics (2017).
\newblock A Comparison of PCE and CPI: Methodological Differences in U.S. Inflation Calculation and Their Implications.
\newblock Washington, DC: U.S. Department of Labor.

\bibitem{BEA_PCE_Methods}
Bureau of Economic Analysis (2024).
\newblock NIPA Handbook, Chapter 5: Personal Consumption Expenditures (PCE) Price Index.
\newblock Washington, DC: U.S. Department of Commerce. (December 2024 update)

\bibitem{FRED_PCEPI}
Federal Reserve Bank of St. Louis (ongoing).
\newblock Personal Consumption Expenditures: Chain-type Price Index (PCEPI), FRED series.

\bibitem{Leontief1936}
Leontief, W. (1936).
\newblock Quantitative Input and Output Relations in the Economic Systems of the United States.
\newblock \emph{The Review of Economics and Statistics}, 18(3), 105--125. \doi{10.2307/1927837}

\bibitem{HarrisonKreps1979}
Harrison, J.~M., \& Kreps, D.~M. (1979).
\newblock Martingales and Arbitrage in Multiperiod Securities Markets.
\newblock \emph{Journal of Economic Theory}, 20(3), 381--408. \doi{10.1016/0022-0531(79)90043-7}

\bibitem{Brioschi1989}
Brioschi, F., Buzzacchi, L., \& Colombo, M.~G. (1989).
\newblock Risk Capital Financing and the Separation of Ownership and Control in Business Groups.
\newblock \emph{Journal of Banking \& Finance}, 13(4--5), 747--772. \doi{10.1016/0378-4266(89)90040-X}

\bibitem{Leontief1986}
Leontief, W.~W. (1986).
\newblock \emph{Input-Output Economics} (2nd ed.).
\newblock New York: Oxford University Press.

\bibitem{DybvigRoss1987}
Dybvig, P.~H., \& Ross, S.~A. (2003).
\newblock Arbitrage, State Prices and Portfolio Theory.
\newblock In G.~M. Constantinides, M. Harris, \& R.~M. Stulz (Eds.), \emph{Handbook of the Economics of Finance}, Vol.~1B, 605--637. Elsevier. \doi{10.1016/S1574-0102(03)01019-7}

\bibitem{IFRS10}
International Accounting Standards Board (IASB) (2011).
\newblock \emph{IFRS 10: Consolidated Financial Statements}.
\newblock London: IFRS Foundation.

\bibitem{SNA2008}
European Commission, International Monetary Fund, Organisation for Economic Co-operation and Development, United Nations, \& World Bank (2009).
\newblock \emph{System of National Accounts 2008}.
\newblock New York: United Nations.

\bibitem{RogersVeraart2013}
Rogers, L.~C.~G., \& Veraart, L.~A.~M. (2013).
\newblock Failure and Rescue in an Interbank Network.
\newblock \emph{Management Science}, 59(4), 882--898. \doi{10.1287/mnsc.1120.1569}

\bibitem{EisenbergNoe2001}
Eisenberg, L., \& Noe, T.~H. (2001).
\newblock Systemic Risk in Financial Systems.
\newblock \emph{Management Science}, 47(2), 236--249. \doi{10.1287/mnsc.47.2.236.9835}

\bibitem{Battiston2012}
Battiston, S., Puliga, M., Kaushik, R., Tasca, P., \& Caldarelli, G. (2012).
\newblock DebtRank: Too Central to Fail? Financial Networks, the FED and Systemic Risk.
\newblock \emph{Scientific Reports}, 2, 541. \doi{10.1038/srep00541}

\bibitem{Claessens2000}
Claessens, S., Djankov, S., \& Lang, L.~H.~P. (2000).
\newblock The Separation of Ownership and Control in East Asian Corporations.
\newblock \emph{Journal of Financial Economics}, 58(1--2), 81--112. \doi{10.1016/S0304-405X(00)00067-2}

\bibitem{Elsinger2009}
Elsinger, H. (2009).
\newblock Financial Networks, Cross-Holdings, and Limited Liability.
\newblock \emph{OeNB Working Paper No. 156}. Vienna: Oesterreichische Nationalbank.

\bibitem{Elliott2014}
Elliott, M., Golub, B., \& Jackson, M.~O. (2014).
\newblock Financial Networks and Contagion.
\newblock \emph{American Economic Review}, 104(10), 3115--3153. \doi{10.1257/aer.104.10.3115}

\bibitem{Acemoglu2015}
Acemoglu, D., Ozdaglar, A., \& Tahbaz-Salehi, A. (2015).
\newblock Systemic Risk and Stability in Financial Networks.
\newblock \emph{American Economic Review}, 105(2), 564--608. \doi{10.1257/aer.20130456}

\end{thebibliography}
\end{document}